\documentclass[a4paper,12pt]{amsart}
\usepackage[round]{natbib}
\usepackage{url} 
\usepackage[utf8]{inputenc}
\usepackage[all]{xy}
\usepackage{amssymb}
\usepackage{amsmath}
\usepackage{mathtools}
\usepackage{upgreek} 
\usepackage{prodint} 
\usepackage{csquotes}
\usepackage{subfig}
\usepackage{booktabs} 
\usepackage[shortlabels]{enumitem}
\usepackage{graphicx}
\graphicspath{{Figures/}} 
\usepackage{amsaddr}
\usepackage{color}
\usepackage{appendix}
\usepackage{mwe}
\usepackage{subfig}
\usepackage{tikz} 
\usetikzlibrary{arrows.meta,shapes} 
\usepackage{pgf,tikz}
\usetikzlibrary{arrows,shapes.arrows,shapes.geometric,shapes.multipart,decorations.pathmorphing,positioning,shapes.swigs,}

\usepackage{multirow}
\usepackage{soul}
\usepackage{xcolor}
\usepackage[left=1in, right=1in, top = 1.3in, bottom = 1.3in]{geometry}
\usepackage{empheq}

\usepackage{setspace}
\theoremstyle{definition} 
\newcommand\independent{\protect\mathpalette{\protect\independenT}{\perp}}
\def\independenT#1#2{\mathrel{\rlap{$#1#2$}\mkern2mu{#1#2}}}
\newtheorem{lemma}{Lemma}

\newtheorem{theorem}{Theorem}
\newtheorem{proposition}{Proposition}
\newtheorem{definition}{Definition}

\title[]{Causal inference with recurrent and competing events}

\author{Matias Janvin$^1$, Jessica G. Young$^{2,3,4}$, Pål C. Ryalen$^5$, Mats J. Stensrud$^1$}
 \address{$^1$Institute of Mathematics, École Polytechnique Fédérale de Lausanne, Switzerland}
 \address{$^2$Department of Epidemiology, Harvard T.H. Chan School of Public Health, USA}
 \address{$^3$CAUSALab, Harvard T.H. Chan School of Public Health, USA}
 \address{$^4$Department of Population Medicine, Harvard Medical School and Harvard Pilgrim Health Care Institute, USA}
\address{$^5$Department of Biostatistics, University of Oslo, Norway}

\date{\today}

\MakeOuterQuote{"}\EnableQuotes
\DeclareUnicodeCharacter{00A0}{ }

\allowdisplaybreaks 

\begin{document}
\maketitle

\clearpage

\begin{abstract}
    Many research questions concern treatment effects on outcomes that can recur several times in the same individual. For example, medical researchers are interested in treatment effects on hospitalizations in heart failure patients and sports injuries in athletes. Competing events, such as death, complicate causal inference in studies of recurrent events because once a competing event occurs, an individual cannot have more recurrent events. Several statistical estimands have been studied in recurrent event settings, with and without competing events. However, the causal interpretations of these estimands, and the conditions that are required to identify these estimands from observed data, have yet to be formalized. Here we use a formal framework for causal inference to formulate several causal estimands in recurrent event settings, with and without competing events.  We clarify when commonly used classical statistical estimands can be interpreted as causal quantities from the causal mediation literature, such as (controlled) direct effects and total effects. Furthermore, we show that recent results on interventionist mediation estimands allow us to define new causal estimands with recurrent and competing events that may be of particular clinical relevance in many subject matter settings.  We use causal directed acyclic graphs and single world intervention graphs to illustrate how to reason about identification conditions for the various causal estimands based on subject matter knowledge. Furthermore, using results on counting processes, we show that our causal estimands and their identification conditions, which are articulated in discrete time, converge to classical continuous time counterparts in the limit of fine discretizations of time. We propose estimators and study their properties. Finally, we use the proposed estimators to compute the effect of blood pressure lowering treatment on the recurrence of acute kidney injury using data from the Systolic Blood Pressure Intervention Trial (SPRINT).

\end{abstract}

\section{Introduction\label{sec:introduction}}

Practitioners and researchers are often interested in treatment effects on outcomes that can recur in the same individual over time. Such outcomes include hospitalizations in heart failure patients~\citep{anker_time_2012}, fractures in breast cancer patients with skeletal metastases~\citep{chen_tests_2004} and rejection episodes in recipients of kidney transplants~\citep{cook_marginal_1997}. However, in many studies of recurrent events, individuals also experience competing events, such as death.   These events may substantially complicate causal inference.  

For example, in the Systolic Blood Pressure Intervention Trial~\citep{sprint2015randomized}, investigators found that intensive blood pressure lowering therapy increased the expected number of acute kidney injury episodes (a possible harmful side effect of blood pressure treatment) compared to standard blood pressure treatment. However, individuals on intensive blood pressure therapy had a lower incidence of all-cause mortality.

In this example, all-cause mortality is a \textsl{competing event} for the outcome of interest (number of recurrences of acute kidney injury) because once an individual dies they cannot subsequently experience the recurrent event.\footnote{What we define as a \textsl{competing event} is often called a \textsl{terminating} event in the recurrent events literature.}  Due to the randomized design these findings indeed have a ``causal'' interpretation: the results support an average harmful effect of intensive blood pressure treatment on the number of acute kidney injury recurrences. However, analogous to previous arguments in the case where the outcome of interest is an incident (rather than a recurrent) event subject to competing events \citep{young_causal_2020,stensrud_separable_2020-1}, this ``protection'' is difficult to interpret in light of the finding that mortality risk is lowered by intensive blood pressure treatment.  The increased number of acute kidney injury episodes in the intensive treatment arm might \textsl{only} be due to the treatment effect on mortality.

Early works on competing events \citep{tsiatis_nonidentifiability_1975,gail_review_1975,prentice1978analysis} considered the problem of identifying the survival function of the event of interest under elimination of competing events, and concluded that this is often unfeasible as it requires strong independence assumptions between the competing events, as well as a clear conceptualization of how to eliminate a competing event without affecting the risk of the other events. Instead,~\citet{prentice1978analysis} advocated for the cause-specific hazard as an identifiable estimand in the competing event setting.  Although these early works are clearly based on an insightful intuition, they are not grounded in a formal framework for characterizing causal effects and their identifying conditions, which makes it difficult to interpret the effect estimates from these procedures and to assess recommendations regarding analytic choices. For example, it has been clarified that cause-specific hazards do not have a desirable causal effect interpretation \citep{robins_new_1986,young_causal_2020,aalen2015does,martinussen_subtleties_2020,hernan_hazards_2010, stensrud2020test,stensrud2020conditional}.

The importance of characterizing the causal interpretation of statistical estimands is increasingly acknowledged both within and outside of the academic causal inference community~\citep{european_medicines_agency_qualification_nodate}. In a series of articles by the Recurrent Event Qualification Opinion Consortium~\citep{schmidli_estimands_2021,wei_properties_2021,fritsch_efficiency_2021}, six candidate causal estimands were proposed in recurrent event settings with competing events, defined by counterfactual contrasts under different treatment scenarios in the following: 1) the expected number of events in the study population;  2) the expectation over a composite of the recurrent and competing events in the study population; 3) the expected number of events under an intervention which prevents the competing event from occurring in the study population; 4) the expected number of events in a subset of the study population consisting of the principal stratum of individuals that would survive regardless of treatment; 5) the ratio of the expected number of recurrences to the restricted mean survival by the end of follow-up in the study population, and 6) the ratio of the expectation over a composite of the recurrent and competing events to the restricted mean survival by the end of follow-up in the study population.

In addition to defining these various counterfactual estimands, ~\citet{schmidli_estimands_2021} considered some aspects of their differences in interpretation, as well as, for some of the estimands, approaches to statistical analysis. However, they did not consider assumptions needed to identify any of these counterfactual estimands in a given study with a function of the observed data. Once a causal estimand is chosen, this identification step is required to justify a choice of approach to statistical analysis. Furthermore, \cite{schmidli_estimands_2021} did not consider how underlying questions about treatment mechanism may be important to the choice of estimand in recurrent events studies when treatment has a causal effect on competing events, as illustrated in the example above. %

In this work, we formalize the interpretation, identification and estimation of various counterfactual estimands in recurrent event settings with competing events using counterfactual causal models ~\citep{robins_new_1986,pearl_causality:_2009,richardson_single_2013, robins_alternative_2011, robins2020interventionist}. Building on ideas in~\citet{young_causal_2020} and~\citet{stensrud_separable_2020-1} for the case where the outcome of interest is an incident event (e.g. diagnosis of prostate cancer), we show that several of these estimands in recurrent event settings can be interpreted as special cases of causal effects from the mediation literature -- total, controlled direct, and separable effects -- by conceptualizing the competing event as a time-varying ``mediator'' \citep{robins1992identifiability,robins_alternative_2011,robins2020interventionist}.  We give identification conditions and derive identification formulas for these estimands and demonstrate how single world intervention graphs (SWIGs) ~\citep{richardson_single_2013} can be used to reason about identification conditions with subject matter knowledge. Our results will also formalize the counterfactual interpretation of statistical estimands for recurrent events from the counting process literature \citep{cook_statistical_2007, andersen_modeling_2019}, which has not adopted a formal causal (counterfactual) framework for motivating results.

The article is organized as follows.  In Sec.~\ref{sec:observed_data}, we present the structure of the observed data, without the complication of losses to follow-up. In Sec.~\ref{sec:estimands_with_competing_events}, we define and describe several causal estimands for recurrent events in settings with competing events. In Sec.~\ref{sec:choosing_estimand} we give our own prescription for choosing an appropriate causal estimand for recurrent outcomes. In Sec.~\ref{sec:censoring}, we consider how to treat the censoring of events, including by loss to follow-up. In Sec.~\ref{sec:identification} we discuss identifiability conditions and give identification formulas for the proposed causal estimands. Furthermore, we demonstrate the convergence of discrete time estimands to continuous time estimands, and establish the correspondence between the discrete time identification conditions and the classical independent censoring assumption in event history analysis.\footnote{While these results are shown for recurrent event outcomes, they also apply to the more classical competing event setting described in~\citet{young_causal_2020}, which constitutes a special case of the current work.} In Sec.~\ref{sec:estimation_maintext}, we describe statistical methods for the proposed estimands, and establish conditions for their consistency. In Sec.~\ref{sec:examples_revisited}, we illustrate our results using a data example on acute kidney injury under blood pressure treatment. Finally, in Sec.~\ref{sec:discussion}, we provide a discussion.

\section{Factual data structure \label{sec:observed_data}}
Consider a randomized trial, like SPRINT, where $i\in \{1,\dots,n\}$ i.i.d.\ individuals with elevated risk of cardiovascular disease were randomized to standard versus intensive blood pressure lowering therapy $A\in\{0,1\}$ ($0$ indicates assignment to standard treatment, $1$ assignment to intensive treatment).  Because the individuals are i.i.d., we suppress the subscript $i$. Let $k\in\{0,\dots,K+1\}$ denote $K+2$ consecutive ordered intervals of time comprising the follow-up (e.g. days, weeks, months)  with time interval $k=0$ corresponding to the interval of treatment assignment (baseline) and $k=K+1$ corresponding to the last possible follow-up interval, beyond which no information has been recorded. Without loss of generality, we choose a timescale such that all intervals have a duration of 1 unit of time until Sec.~\ref{sec:correspondence}.

Let $Y_k\in \{0,1,2,\dots\}$ denote the cumulative count of acute kidney injury episodes by the end of interval $k$ and $D_k\in\{0,1\}$ an indicator of death by the end of interval $k$. Define $D_0\equiv Y_0\equiv 0$, that is, individuals are alive and have not yet experienced any post-treatment recurrent events at baseline. Let $L_0$ be a vector of baseline covariates measured before the treatment assignment $A$, capturing pre-treatment common causes of acute kidney injury and death. For $k>0$, let $L_k\in\mathcal{L}$ denote a vector of time-varying covariates measured in interval $k$, containing the most recent blood pressure measurements.\footnote{Our presentation focuses on  \textsl{intention-to-treat} effects by defining $A$ as an indicator of baseline assignment to a particular treatment arm.  Our results trivially extend to accommodate effects of adherence to a particular protocol at baseline by instead taking $A$ to be the actual treatment strategy followed at baseline and by including common causes of treatment adherence, acute kidney injury and death in $L_0$.  In either case, indicators of time-varying adherence to the protocol may be important to include in $L_k$, $k>0$, for the purposes of identification to be discussed in later sections.}

The history of a random variable through $k$ is denoted by an overbar (i.e. $\overline{Y}_{k} \equiv (Y_0,\dots,Y_k)$ and $\overline{L}_{k} \equiv (L_0,\dots,L_k)$) and future events are denoted by underbars (i.e. $\underline{D}_{k} \equiv (D_k,\dots,D_{K+1})$).

We assume no loss to follow-up until Sec.~\ref{sec:censoring}, and we assume that variables are temporally (and topologically\footnote{A topological order is a linear ordering of nodes in a graph from first to last.}) ordered as $D_k, Y_k,L_k$ within each follow-up interval. We adopt the notational convention that any variable with a negative time index occurring in a conditioning set is taken to be the empty set $\emptyset$ (e.g.\ $P(A=a|L_{-1},B) = P(A=a|B)$ for an event $B$).

An individual cannot experience recurrent events after a competing event, such as death, has occurred: if an individual experiences death at time $k^\dagger$, then $Y_{j}=Y_{k^\dagger-1}$ and $D_{j}=1$  for all $j \geq k^\dagger$. Thus, the type of outcome that is the focus of this manuscript is defined in the factual data after death occurs.  This is in contrast to what we will refer to as a 'truncation by death' setting, where the outcome of interest is undefined after an individual experiences the competing event \citep{young_causal_2020, young2021identified, stensrud2020conditional,stensrud_discussion_2021}. For example, when the outcome of interest is quality of life in cancer patients, this is only defined for individuals who are alive unless the investigator chooses to assign an arbitrary quality of life value to dead individuals.  Ultimately, the true distinction between a ``competing event'' setting, where outcomes cannot occur post-death, and a ``truncation by death'' setting comes down to the estimands that the investigator is willing to consider.  When the outcome is undefined after death then certain estimands will not be available that are available when such outcomes are defined \citep{young2021identified}.  We consider this further in Section 3.

In what follows, we will use causal directed acyclic graphs~\citep{pearl_causality:_2009} (DAGs) to represent underlying data generating models. We assume that the DAG represents a Finest Fully Randomized Causally Interpreted Structural Tree Graph (FFRCISTG) model~\citep{robins_new_1986,richardson_single_2013}. Furthermore, we will assume that statistical independencies in the data are faithful to the DAG (see Appendix~\ref{sec:app_correspondence_censoring} for the definition of faithfulness that we adopt here). An example of a DAG, encoding a set of possible assumptions on the data generating model for the trial in the data example described in Sec.~\ref{sec:introduction} is shown in Fig.~\ref{fig:DAG_observed_data}.

\begin{figure} 
    \centering
\begin{tikzpicture}
    \begin{scope}[every node/.style={thick,draw=none}]
    \node[name=Yk] at (2,1.5){$Y_k$};
    \node[name=Dk] at (2,-1.5){$D_k$};
    \node[name=Dk1] at (5.5,-1.5){$D_{k+1}$};
    \node[name=Yk1] at (5.5,1.5){$Y_{k+1}$};
    \node[name=A] at (-2,0){$A$};
    \node[name=L0] at (-4,0){$L_0$};
    \node[name=L0prime] at (1,0){$L_{k-1}$};
\end{scope}
\begin{scope}[>={Stealth[black]},
              every node/.style={fill=white,circle},
              every edge/.style={draw=black,very thick}]
    \path [->] (A) edge (Yk);
    \path [->] (A) edge (L0prime);
    \path [->] (L0prime) edge (Yk);
    \path [->] (L0prime) edge (Yk1);
    \path [->] (L0prime) edge (Dk);
    \path [->] (L0prime) edge (Dk1);
    \path [->] (A) edge (Dk);
    \path [->] (A) edge[bend right=30] (Dk1);
    \path [->] (A) edge[bend left=30] (Yk1);
    \path [->] (Dk) edge (Yk);
    \path[->] (Yk) edge (Yk1);
    \path[->] (Yk) edge (Dk1);
    \path[->] (Dk1) edge (Yk1);
    \path[->] (Dk) edge (Dk1);
    \path [->] (L0) edge[bend right] (Dk);
    \path [->] (L0) edge[bend left] (Yk);
    \path [->] (L0) edge[bend right=35] (Dk1);
    \path [->] (L0) edge[bend left=35] (Yk1);
    \path [->] (L0) edge[bend right] (L0prime);
\end{scope}
\end{tikzpicture}

\caption{An example of a possible causal model describing the Systolic Blood Pressure Intervention Trial \citep{sprint2015randomized}, where individuals are randomized to intensive versus standard blood pressure therapy. The trial outcomes are recurrent episodes of acute kidney injury ($Y_k$) and survival ($D_k$).}
\label{fig:DAG_observed_data}
\end{figure}

\section{Counterfactual estimands\label{sec:estimands_with_competing_events}}
In this section, we consider various counterfactual estimands in settings with recurrent and competing events. We propose extensions of previously considered counterfactual estimands that quantify causal effects on incident failure in the face of competing events   \citep{stensrud_separable_2020-1,young_causal_2020,stensrud_generalized_2021} to the recurrent events setting. This includes a new type of \textsl{separable effect}, inspired by the seminal decomposition idea of \citet{robins_alternative_2011}, that may disentangle the treatment effect on recurrent acute kidney injury from its effect on survival. We also discuss additional counterfactual estimands in recurrent event settings.

We denote counterfactual random variables by superscripts, such that $Y^{a}_k$ is the recurrent event count that would be observed at time $k$ had, possibly contrary to fact, treatment been set to $A=a$. By causal effect, we mean a contrast of some functional (e.g. the mean) of the counterfactual distribution in the same subset of individuals.

\subsection{Total effect\label{sec:estimand_total}}
The counterfactual marginal mean number of recurrent events by time $k$ under an intervention that sets $A$ to $a$ is 
\begin{align*}
    E[Y^a_k] \quad\textrm{for}\quad k\in\{0,\dots,K+1\} ~.
\end{align*}

In turn, the counterfactual contrast
\begin{align}
    E[Y^{a=1}_k] ~\textrm{vs.}~ E[Y^{a=0}_k]  ~ \label{eq:total_effect_no_censoring}
\end{align}
quantifies a causal effect of treatment assignment on the mean number of recurrent events by $k$.  \citet{schmidli_estimands_2021} referred to this effect as the 'treatment policy' estimand.  However, in order to understand the interpretational implications of choosing this effect measure when competing events exist, it is important to understand that (\ref{eq:total_effect_no_censoring}) also coincides with an example of a \textsl{total effect} as historically defined in the causal mediation literature ~\citep{robins1992identifiability,young_causal_2020}.

In our running example, the total effect quantifies the effect of intensive versus standard blood pressure treatment ($A$) on recurrent acute kidney injury ($Y_k$) through all causal pathways, including pathways through survival ($D_k$), as depicted by all directed paths connecting $A$ and $Y$ nodes intersected by $D$ nodes in the causal diagram in Fig. ~\ref{fig:DAG_observed_data} \citep{young_causal_2020}. Therefore, a non-null value of the total effect is not sufficient to conclude that the treatment exerts direct effects on acute kidney injury (outside of death): the total effect may also (or only) be due to an (indirect) effect on survival, keeping individuals at risk of acute kidney injury for a longer (or shorter) period of time. 

In addition to the total effect on recurrent acute kidney injury ($Y_k$), we might consider the total (marginal) effect of treatment on survival, given by the marginal contrast in cumulative incidences
\begin{align}
    E[D^{a=1}_{k}]~\textrm{vs.}~ E[D^{a=0}_k] ~. \label{eq:total_effect_survival}
\end{align}
However, simultaneously considering the total effect of treatment on acute kidney injury and on survival is still insufficient to determine by which  mechanisms the  treatment affects acute kidney injury and death. 
For example, suppose that individuals in treatment arm $A=1$ experience the competing event shortly after treatment initiation. In this case, no recurrent events would be recorded in this treatment arm. Clearly, in this setting it would not be possible for an investigator to draw any conclusions about the mechanism by which treatment acts on the recurrent event outside of the competing event.

\subsection{Controlled direct effect}
Following \citet{robins1992identifiability}, consider the counterfactual mean number of events under an intervention that prevents the competing event from occurring and sets treatment to $A=a$,
\begin{align*}
    E[Y^{a,\overline{d}=0}_k] \quad\textrm{for}\quad  k\in\{0,\dots,K+1\} ~,
\end{align*}
where the overline in the superscript denotes an intervention on all respective intervention nodes in the history of the counterfactual, i.e.\ $Y_k^{a,\overline{d}=0} \equiv Y_k^{a,\overline{d}_{k}=0}$. 

In turn, the counterfactual contrast
\begin{align}
    E[Y^{a=1,\overline{d}=0}_k]~\textrm{vs.}~ E[Y^{a=0,\overline{d}=0}_k]~ \label{eq:direct_effect_no_censoring}
\end{align}
quantifies a causal effect of treatment assignment on the mean number of recurrent events by $k$ under an additional intervention that somehow ``eliminates competing events''.  \citet{schmidli_estimands_2021}  referred to this effect as the 'hypothetical strategy' estimand.  However, it is useful to notice that the effect (\ref{eq:direct_effect_no_censoring}) coincides with an example of a \textsl{controlled direct effect} as defined in the causal mediation literature ~\citep{robins1992identifiability,young_causal_2020}. The quantity $E[Y^{a,\overline{d}=0}_k]$ is closely related to the survival function under the elimination of competing events, as discussed in the early competing events literature by e.g.\ \citet{tsiatis_nonidentifiability_1975,prentice1978analysis,putter_tutorial_2007}, although without using a formal causal framework.

In our example, the controlled direct effect isolates direct effects of treatment on recurrent acute kidney injury by considering a (hypothetical) intervention which prevents death from occurring in all individuals. An important reservation against the controlled direct effect is that it is often difficult to conceptualize an intervention which prevents the competing event from occurring  ~\citep{young_causal_2020}. For example, there exists no practically feasible intervention that can eliminate death due to all causes. Without clearly establishing the intervention being targeted, the interpretation of the direct effect is ambiguous and its role in informing decision-making is unclear. The unclear role of the controlled direct effect in decision-making was reiterated by~\citet{schmidli_estimands_2021} in their discussion of the 'hypothetical strategy', although the authors did not discuss the role of the estimand in clarifying the mechanism by which treatment affects the outcome.

\subsection{Separable effects\label{sec:sep_effects_no_censoring}}

Following~ \citet{robins_alternative_2011,robins2020interventionist} and \citet{stensrud2020conditional, stensrud_separable_2020-1}, we will define an actionable notion of direct (indirect) effects that refers to an intervention that might be implemented currently or in the future. These effects require that the investigator pose candidates for modified versions of the study treatment, denoted $A_Y$ and $A_D$, with the following properties:  let $M_Y$ and $M_D$ be two random variables, and suppose that the following conditions hold for the original treatment $A$ and the modified treatments $A_Y,A_D$:
\begin{align}
    &\text{All effects of $A$, $A_Y$ and $A_D$ on $Y_k$ and $D_k$, $k\in\{0,\dots,K\}$, are intersected } \notag\\
    &\text{by $M_Y$ or $M_D$,  and } \notag\\
    &\quad M_Y^{a_Y=a,a_D} = M_Y^a \quad \text{for} \quad a_D\in\{0,1\} ~,\notag\\
    &\quad M_D^{a_Y,a_D=a}=M_D^a \quad\text{for}\quad a_Y\in\{0,1\} ~. \label{eq:modified_treatment_assumption}
\end{align}
"Intersection" refers to the paths in the respective causal DAG. Assumption (\ref{eq:modified_treatment_assumption}) is referred to as the modified treatment assumption and is discussed in~\citet{stensrud2020conditional}.  According to (\ref{eq:modified_treatment_assumption}), receiving $A_Y=A_D=a$ results in the same outcomes as receiving $A=a$ for $a\in\{0,1\}$. While a physical treatment decomposition is one way in which assumption (\ref{eq:modified_treatment_assumption}) may hold, it may also hold for modified treatments that are not a physical decomposition \citep{stensrud_generalized_2021,stensrud2020conditional}. The modified treatment assumption (\ref{eq:modified_treatment_assumption}) can in principle be falsified in a future six-armed trial where individuals are exposed to $A,A_Y,A_D\in\{0,1\}$ once the modified treatment becomes available~\citep{stensrud2020conditional}.

In the case of a decomposition, an individual receiving $A_Y=A_D=0$ has received the same treatment as $A=0$ (assignment to neither of the treatment components) and an individual receiving $A_Y=A_D=1$ the same treatment as $A=1$ (assignment to both treatment components). The marginal mean number of events under a hypothetical intervention where we jointly assign $A_Y=a_Y$ and $A_D=a_D$ for any combination of $a_D\in\{0,1\}$ and $a_Y\in\{0,1\}$, possibly such that $a_Y\neq a_D$ is 
\begin{align*}
     E[Y_k^{a_Y,a_D}] \quad\textrm{for}\quad k\in\{0,\dots,K+1\} ~.
\end{align*}
Contrasts of this estimand for different levels of $A_Y$ and $A_D$ constitute particular examples of \textsl{separable effects}~\citep{stensrud_separable_2020-1,stensrud_generalized_2021}, a type of interventionist mediation estimand \citep{robins_alternative_2011,robins2020interventionist,didelez2019defining}.  For example, the  separable effect of $A_Y$ evaluated at $A_D=0$ is
\begin{align}
    E[Y_k^{a_Y=1,a_D=0}] ~\textrm{vs.}~ E[Y_k^{a_Y=0,a_D=0}] \quad\textrm{for}\quad  k\in\{0,\dots,K+1\} ~. \label{eq:sep_effect_aY}
\end{align}
Expression (\ref{eq:sep_effect_aY}) quantifies the effect of only treating with the $A_Y$ component versus neither of the components. 

These estimands correspond to the effects of joint interventions on candidate modified treatments $A_Y$ and $A_D$, even when the modified treatment assumption (\ref{eq:modified_treatment_assumption}) does not hold. However, the modified treatment assumption (\ref{eq:modified_treatment_assumption}) is sufficient in order for the separable effects to explain the mechanism by which the original treatment $A$ exerts its effects on the recurrent outcome~\citep[Appendix A]{stensrud_generalized_2021}.

Returning to the data example, a well-known biological effect of angiotensin converting enzyme inhibitors (ACE) and angiotensin II receptor blockers (ARB) (two common antihypertensive medications) is that they reduce the renal filtration pressure by binding to receptors in the kidneys which dilate efferent glomerular arterioles, which in turn can lead to a substatial drug-induced fall in kidney function \citep{brunton_goodman_2018}.  In light of this, drug developers and doctors could be interested in the effect of a hypothetical modified version of an antihypertensive drug, which preserves its effects on systemic blood pressure but does not lead to dilation of efferent glomerular arterioles. In principle, such a modified drug might have similar cardioprotective effects as the original antihypertensive agent, but without the harmful side-effect that can lead to acute kidney injury.

This working background knowledge on the mechanisms by which the study treatment affects recurrent acute kidney injury and competing events allows us to pose candidates for $A_Y$ and $A_D$ in this example and, as we will discuss further below, interpret separable effects in terms of direct, indirect, or path-specific effects of $A$.  Specifically, the modified treatment assumption is conceivable in this example by defining $A_Y$ to be the component of blood pressure therapy that binds to efferent arterioles in the kidneys, causing their dilatation ($M_Y$), and $A_D$ as the remaining components of the treatment, including those that exert their effects by lowering systemic blood pressure ($M_D$). Thus, $A_Y$ and $A_D$ are the treatment levels of these two components under intensive versus standard therapy respectively.  A further discussion of this decomposition of blood pressure therapy into the aforementioned $A_Y$ and $A_D$ components is given in \citet{stensrud_generalized_2021}. 
 
Additional assumptions or \textit{isolation conditions} \citep{stensrud_generalized_2021}, are then required in order to \textit{interpret} any given separable effect as a direct, indirect, or otherwise path-specific effect of the original study treatment: if the $A_Y$ component has no effect on survival, then $E[Y_k^{a_Y=1,a_D}] \ \textrm{vs.}\ E[Y_k^{a_Y=0,a_D}]$ captures exclusively the effect of the $A_Y$ component on acute kidney injury not mediated by survival. We can formalize this statement using the condition of strong $A_Y$ partial isolation, inspired by ~\citet{stensrud_generalized_2021}:

A treatment decomposition satisfies strong $A_Y$ partial isolation if
\begin{align}
    \textrm{There are no causal paths from $A_Y$ to $D_k$ for all $k\in\{0,\dots,K+1\}$ ~.} \label{eq:isolation}
\end{align}
Under strong $A_Y$ partial isolation,  (\ref{eq:sep_effect_aY}) captures \emph{only} treatment effects on the recurrent event not via treatment effects on competing events, and is therefore a direct effect. In our example on blood pressure treatment, strong $A_Y$ partial isolation likely fails, as acute kidney injury may in and of itself increase the risk of death, and therefore effects through the path $A_Y\rightarrow Y_j \rightarrow D_{k>j}$ cannot be ruled out.

Consequently, (\ref{eq:sep_effect_aY}) also captures effects of $A_Y$ on $Y_k$ via $\overline{D}_k$, and therefore cannot be interpreted as a direct effect outside of $\overline{D}_k$. 

Another isolation condition, $A_D$ partial isolation, allows us to interpret separable effects as indirect effects of treatment on the recurrent outcome via effects on survival.
A brief account of the isolation conditions is given in Appendix~\ref{sec:app_isolation}, and is discussed in detail for the competing events setting in~\citet{stensrud_generalized_2021}. 
If we had access to a four arm randomized trial where individuals are observed under all four treatment combinations $(A_Y,A_D)\in\{0,1\}^2$, and there is no loss to follow-up, these effects could easily be identified and estimated by two-way comparisons of the four different treatment combinations. Such two-way comparisons would also allow the strong $A_Y$ partial isolation condition (\ref{eq:isolation}) to be tested: in particular, a non-null value of the two-way comparison $E[D_k^{a_Y=1,a_D}]$ vs. $E[D_k^{a_Y=0,a_D}]$ implies a violation of (\ref{eq:isolation}). Conversely, inspection of the contrast $E[Y_k^{a_Y,a_D=1}]$ vs. $E[Y_k^{a_Y,a_D=0}]$ can strengthen or weaken our belief in the $A_D$ partial isolation condition, although cannot be used to falsify the assumption.   Because we only observe two of the four treatment combinations in the trial described in Sec. \ref{sec:observed_data}, namely $A_Y=A_D=1$ and $A_Y=A_D=0$, the separable effects target effects that require identifying assumptions beyond those that hold by design in this two arm trial.  We will consider these assumptions in Sec.~\ref{subsection: separable effects identification}.

\subsection{Estimands with composite outcomes} \label{subsection: average rate estimands}

\citet{schmidli_estimands_2021} proposed the estimands
\begin{align}
    \frac{E[Y_{k}^{a}]}{E[\mu_{k}^{a}]} \quad\textrm{for}\quad  k\in\{1,\dots,K+1\}~, \label{eq:while_alive} \\
    E\left[ \frac{Y_{k}^a}{\mu_{k}^a} \right] \quad\textrm{for}\quad  k\in\{0,\dots,K+1\} \label{eq:average_individual_rate} ~,
\end{align}
where $\mu_k^{a}=\sum_{i=0}^k I(D_i^{a}=0)$ is the counterfactual restricted survival under an intervention that sets treatment to $a$.\footnote{In continuous time, the restricted survival can be written as $\int_0^t I(T^D\geq s)ds$, where $T^D$ is the time of the competing event. Taking the expectation gives  $\int_0^t S(s)ds$ for survival function $S(t)$, which is the restricted mean survival in continuous time~\citep{aalen_survival_2008}.}
Expressions (\ref{eq:while_alive})-(\ref{eq:average_individual_rate}) differ subtly:  (\ref{eq:average_individual_rate}) is the mean of a ratio and implicitly reflects the association between recurrent and competing events, whereas any information about this association is erased by (\ref{eq:while_alive}), which is a ratio of means. \citet{schmidli_estimands_2021} referred to (\ref{eq:while_alive}) as the `while alive strategy' estimand.  A contrast in (\ref{eq:while_alive})-(\ref{eq:average_individual_rate}) under different levels of $a$ captures both treatment effects on acute kidney injury and on the competing event.

Different types of composite outcomes have also been suggested. For example,~\citet{schmidli_estimands_2021} described the estimand
\begin{align*}
    E \left[ I(D_k^{a}=0)+Y_k^{a} \right] ~,
\end{align*}
which could also be extended by multiplying $D_k^{a}$ or $Y_k^{a}$ by a weight. Likewise,~\citet{claggett_quantifying_2018} introduced a reverse counting process, which can be formulated as
\begin{align*}
  E\left[ \sum_{j=1}^M I(Y_k^a < j)I(D_{k}^{a}=0) \right] 
\end{align*}
for recurrent outcomes. The estimand is the expectation over a counting process which starts at $M$ and decrements in steps of one every time the recurrent event occurs. If the terminating event occurs, the process drops to zero.

There are common limitations to all estimands in this subsection:
\begin{enumerate}[(I)]
    \item Neither can be used to draw formal conclusions about the mechanism by which the treatment affects the recurrent event and the event of interest for the same reason as the total effect (Sec.~\ref{sec:estimand_total}).
    \item The estimands (implicitly or explicitly) assign  weight to the competing and recurrent events by combining them into a single effect measure. However, the choice of 'weights' is not obvious and can differ on a case-by-case basis.
    \item The estimands represent a coarsening of the information in the cumulative incidence and mean frequency, and therefore provide less information than simultaneously inspecting the mean frequency of acute kidney injury and the cumulative incidence of death. Inspecting the mean frequency and cumulative incidence curves separately gives the additional advantage of showing the (absolute) magnitude of each estimand separately as functions of time, which is not visible from the composite estimand alone.
\end{enumerate}
Points (I)-(III) also apply to composite estimands in settings with truncation by death.

\subsection{Estimands that condition on the event history}
\label{sec: intensity based estimands}

The counterfactual intensity of the recurrent event process is defined as
\begin{align} 
    E[Y_{k+1}^a-Y_k^a\mid \overline{L}_k^a, \overline{Y}_k^a,\overline{D}_k^{a},A]  \quad\textrm{for}\quad  k\in\{0,\dots,K+1\} ~. \label{eq:intensity_without_competing}
\end{align}
Expression (\ref{eq:intensity_without_competing}) is a discrete time intensity of $Y_k^a$, conditional on the past history of recurrent events and measured covariates. One could then consider contrasts such as
\begin{align}
   & E[Y_{k+1}^{a=1}-Y_k^{a=1}\mid \overline{L}_k^{a=1}=\overline{l}_{k}, \overline{Y}_k^{a=1}=\overline{y}_{k},\overline{D}_k^{a=1}=0,A=1]\notag \\
    &\hspace{3cm}~\text{vs.}~ \label{eq:contrast_intensity_without_competing} \\
    &E[Y_{k+1}^{a=0}-Y_k^{a=0}\mid \overline{L}_k^{a=0}=\overline{l}_{k}, \overline{Y}_k^{a=0}=\overline{y}_{k},\overline{D}_k^{a=0}=0,A=0] ~. \notag
\end{align}
However, because (\ref{eq:intensity_without_competing}) conditions on the history of the recurrent event process up to time $k$, (\ref{eq:contrast_intensity_without_competing}) generally cannot be interpreted as a causal effect, even though it is a contrast of counterfactual outcomes. This is because it compares different groups of individuals -- those with a particular recurrent event and covariate process history under $a=1$ versus those with that same history under $a=0$. Thus,  a nonnull value of  (\ref{eq:contrast_intensity_without_competing}) does \emph{not} imply that $A$ has a nonnull causal effect on $Y$ at time $k$.
This is analogous to the difficulty in causally interpreting contrasts of hazards for survival outcomes, and has already been discussed extensively in the literature~\citep{robins_new_1986,young_causal_2020,martinussen_subtleties_2020,hernan_hazards_2010, stensrud2020test,stensrud2020conditional}.

An alternative estimand is the expanded notion of separable effects called \textsl{conditional} separable effects~\citep{stensrud2020conditional}, where consideration of causal effects is restricted to a particular subset of ``survivors''~\citep{stensrud2020conditional}. When strong $A_Y$ partial isolation holds, the conditional separable effect evaluated at $A_D=a_D$ is defined as the contrast
\begin{align*}
    E[Y_k^{a_Y=1,a_D}\mid D_k^{a_D}=0] ~\text{vs.}~E[Y_k^{a_Y=0,a_D}\mid D_k^{a_D}=0]~.
\end{align*}

Unlike (\ref{eq:contrast_intensity_without_competing}), the conditional separable effect can be interpreted as a contrast of counterfactual outcomes in the same subset of individuals.
Like the marginal separable effects discussed in Sec.~\ref{sec:sep_effects_no_censoring}, the conditional separable effects rely on assumptions that are testable in a future randomized trial~\citep{stensrud2020conditional}. However, the conditional separable effects require the assumption of strong $A_Y$ partial isolation in order to be well-defined, which is not required by the marginal separable effects. The conditional separable effects can be used even if the investigator considers the outcome of interest to be ill-defined after the competing event.

\subsection{Principal stratum estimand}
\citet{schmidli_estimands_2021} also considered the principal stratum estimand
\begin{align}
    E[Y^a_k \mid D_{k}^{a=0}=0,D_{k}^{a=1}=0] \quad\textrm{for}\quad  k\in\{0,\dots,K+1\} ~, \label{eq:principal_stratum}
\end{align}
which is closely related to the conditional separable effect. Contrasts of (\ref{eq:principal_stratum}), given by 
\begin{align*}
    E[Y^{a=1}_k \mid D_{k}^{a=0}=0,D_{k}^{a=1}=0] ~\text{vs.}~ E[Y^{a=0}_k \mid D_{k}^{a=0}=0,D_{k}^{a=1}=0] ~,
\end{align*}
correspond to principal stratum effects, e.g. the survivor average causal effect~\citep{robins_new_1986,frangakis_principal_2002,schmidli_estimands_2021}. Identification of (\ref{eq:principal_stratum}) was also considered by \citet{xu_bayesian_2022} in the semi-competing events setting. The principal stratum estimand targets an unknown subset of the population  \citep{ robins_new_1986,robins2007discussions, joffe2011principal,dawid2012imagine,stensrud_separable_2020-1,stensrud_translating_2022}. In cases where this subset is small, or non-existent, the principal stratum effects may play an unclear role in decision-making. Integrally linked to the unknown nature of the population to whom a principal stratum effect refers, this estimand depends on cross-world independence assumptions for identification that can not be falsified in any real-world experiment, in contrast to the (conditional) separable effects.

\subsection{Natural direct effect}\label{sec:natural_direct}
The natural (pure) direct effects, originally described by \citet{robins1992identifiability} and later reconsidered by \citet{pearl_direct_nodate}, give another way of defining treatment effects on the recurrent outcome which do not capture the effect on the competing event. One way of doing so is through the contrast
\begin{align*}
E\left[Y_k^{a=1,D_k^{a=0}} \right] ~\textrm{vs.}~ E\left[Y_k^{a=0,D_k^{a=0}} \right] ~.
\end{align*}

Like the controlled direct effect, the natural direct effect also requires the conceptualization of an intervention on the competing event.

Recent work has also considered identification of path specific effects which capture direct and indirect effects through longitudinal mediators~\citep{vansteelandt_mediation_2019,mittinty_longitudinal_2020} as well as natural effects formulated using random interventions on longitudinal mediators \citep{zheng_longitudinal_2017}.

\section{Choosing an estimand\label{sec:choosing_estimand}}

The choice of estimand for a particular problem must be motivated by subject matter
arguments. 
When there is no subject matter support for a causal effect of the treatment on the competing event (i.e. there are no directed arrows from $A$ into $D_k$ at any $k$) or when this mechanism does not create ambiguities with regard to mechanisms of the treatment then the total effect may be enough. 

However, if treatment effects on the competing event could create mechanisms that lead to an ambiguous interpretation of the total effect, then other estimands may help supplement  information quantified by the total effect.  Unlike other proposals for effects to quantify treatment mechanism outlined above, strong assumptions are required to even define the separable effects, putting aside even the issue of identifying them in the study data in hand, and to ascribe them a particular mechanistic interpretation.  Unfortunately, the alternative estimands provided do not avoid such assumptions but rather bury them: for example, an estimate obtained from a real-world study of a controlled direct effect defined relative to an ill-defined intervention on death, or a natural effect defined relative to setting death to a cross-world unobservable value, can never be refuted in the future without additional assumptions on par with the modified treatment assumption/isolation conditions required to understand a separable effect. The required transparency for proceeding with a separable effects analysis can, and in our view should, be viewed as a benefit of this approach: it shines needed light on the reality that using real-world data to answer mechanistic questions is hard and requires detailed assumptions about how the study treatment works. When an investigator is lacking that knowledge, the solution should not be to revert to untestable questions but to acknowledge the need for more time and thought to sharpen hypotheses. In such cases, one may proceed with a total effect, acknowledging its mechanism is not yet understood.  Alternatively, one may proceed with considering separable effects for yet to be elucidated candidates $A_Y$ and $A_D$.  Such an approach is arguably no more vague than previous (in)direct effect notions but, unlike those former notions, has a hope of being sharpened as more knowledge develops.

Finally, the identifying functions for separable effects coincides with those for certain path specific effects in certain settings, including those where full isolation holds. Thus, numerous advancements in statistics for path specific effects, such as natural effects, can still be leveraged for estimation of separable effects ~(see for example \citet{zheng_longitudinal_2017,vansteelandt_mediation_2019}).

\section{Censoring \label{sec:censoring}}
Define $C_{k+1}$, $k\in\{0,\dots,K\}$ as an indicator of loss to follow-up by $k+1$ such that, for an individual with $C_k=0, C_{k+1}=1$, the outcome (and covariate) processes defined in Sec. \ref{sec:observed_data} are only fully observed through interval $k$.  Loss to follow-up (e.g. due to failure to return for study visits) is commonly understood as a form of censoring.   We adopt a more general definition of censoring from ~\citet{young_causal_2020} which captures loss to follow-up but also possibly other events, depending on the choice of estimand.

\begin{definition}[Censoring, \citet{young_causal_2020}]
A censoring event is any event occurring in the study by $k+1$, for any $k\in\{0,\dots,K\}$, that ensures the values of all future counterfactual outcomes of interest under $a$ are unknown even for an individual receiving the intervention $a$.
\end{definition}

Loss to follow-up  by time $k$ is always a form of censoring by the above definition. However, other events may or may not be defined as censoring events depending on the choice of causal estimand. For example, competing events are censoring events by the above definition when the controlled direct effect is of interest, but are not censoring events when the total effect is of interest \citep{young_causal_2020}. This is because the occurrence of a competing event at time $k^\dagger$ prevents knowledge of $\underline{Y}^{a,\overline{d}=0}_{k^\dagger}$, but does not prevent knowledge of $\underline{Y}^a_{k^\dagger}$. By similar arguments, competing events are not censoring events when separable effects are of interest because they do not involve counterfactual outcomes indexed by $\overline{d}=0$ ("elimination of competing events"). When loss to follow-up is present in a study, we will define all effects relative to interventions that include ``eliminating loss to follow-up'' with the added superscript $\overline{c}=0$ to denote relevant counterfactual outcomes, e.g.\ $Y_k^{\overline{c}=0}$. For example, if loss to follow-up is due to the administrative end of a study, the intervention that eliminates loss to follow-up could be conceived as the hypothetical continuation of the study such that every individual is followed until the end of interval $K+1$. Contrasts of such effects are examples of controlled direct effects with respect to interventions on loss to follow-up. The identification assumptions outlined below are sufficient for identifying estimands with this additional interpretation.  \cite{young_causal_2020} discuss additional assumptions that would allow an interpretation without this additional intervention on loss to follow-up. In Sec.~\ref{sec:correspondence}, we establish the correspondence between the notion of censoring adopted in this article and the classical independent censoring assumption in event history analysis.

\section{Identification of the causal estimands\label{sec:identification}}
In this section, we give sufficient conditions for identifying the total, controlled direct and the separable effects as functionals of the observed data. Proofs can be found in Appendix \ref{sec:proof_identification}. Identification of estimands in Sections~\ref{sec: intensity based estimands}-\ref{sec:natural_direct} is beyond the scope of this work.

\subsection{Total effect\label{sec:identification_total_effect}}
 
Consider the following conditions for $k\in\{0,\dots,K\}$:

\textbf{Exchangeability}
\begin{align}
    \overline{Y}_{K+1}^{a,\overline{c}=0}   &\independent A | L_0 \label{eq:exchangeability_total_i}  ~,\\
    \underline{Y}_{k+1}^{a,\overline{c}=0} &\independent C_{k+1}^{a,\overline{c}=0} | \overline{L}_k^{a,\overline{c}=0}, \overline{Y}_k^{a,\overline{c}=0}, \overline{D}_k^{a,\overline{c}=0},\overline{C}_k^{a,\overline{c}=0}, A ~. \label{eq:exchangeability_total_ii}
\end{align}
Assumption (\ref{eq:exchangeability_total_i}) states that the baseline treatment is unconfounded given $L_0$. This holds by design with $L_0=\emptyset$ when treatment assignment $A$ is (unconditionally) randomized, such as in the blood pressure trial considered in our running example. Assumption (\ref{eq:exchangeability_total_ii}) states that the censoring is unconfounded. As we will discuss in Sec.~\ref{sec:correspondence}, this assumption is closely related to the independent censoring assumption in survival analysis.

\textbf{Positivity}
\begin{align}
&P(L_0=l_0) > 0 \implies P(A=a\mid L_0=l_0) >0 ~, \label{positivity_total_effect}\\
    &f_{A,\overline{L}_k,\overline{D}_k,\overline{C}_k,\overline{Y}_k}(a,\overline{l}_k,0,0,\overline{y}_k)>0  \notag\\
    &\quad \implies P(C_{k+1}=0\mid A=a, \overline{L}_k=\overline{l}_k,\overline{D}_k=0,\overline{C}_k=0,\overline{Y}_k=\overline{y}_k) > 0 ~. \label{positivity_total_effect_ii}
\end{align}
Assumption (\ref{positivity_total_effect}) states  that for every level of the baseline covariates, there are some individuals that receive either treatment. Once again, this will hold by design in a trial where $A$ is assigned by randomization, such as in the data example. The second assumption requires that, for any possible observed level of treatment and covariate history amongst those remaining alive and uncensored through $k$, some individuals continue to remain uncensored through $k+1$ with positive probability.

\textbf{Consistency}
\begin{align}
    &\text{If $A=a$ and $\overline{C}_{k+1}=0$,} \notag\\
    &\text{then $\overline{L}_{k+1}=\overline{L}_{k+1}^{a,\overline{c}=0},\overline{D}_{k+1}=\overline{D}_{k+1}^{a,\overline{c}=0}, \overline{Y}_{k+1}=\overline{Y}_{k+1}^{a,\overline{c}=0}$}, \overline{C}_{k+1}=\overline{C}_{k+1}^{a,\overline{c}=0}~.\label{eq:consistency_total_effect}
\end{align}

Let $\Delta X_k=X_{k}-X_{k-1}$ denote an increment of the process $X$.
In Appendix \ref{sec:proof_identification} we show that, under assumptions (\ref{eq:exchangeability_total_i})-(\ref{eq:consistency_total_effect}),
\begin{align}
E[&\Delta Y_i^{a,\overline{c}=0}] =\notag\\
    &\sum_{\Delta\overline{y}_i}\sum_{\overline{d}_i}\sum_{\overline{l}_{i-1}}\prod_{j=0}^{i} \notag\\
    &\quad \Delta y_i\cdot P(\Delta Y_j=\Delta y_j\mid \overline{D}_{j}=\overline{d}_j, \overline{C}_j=0,\overline{L}_{j-1}=\overline{l}_{j-1},\Delta \overline{Y}_{j-1}=\Delta \overline{y}_{j-1},A=a) \notag\\
    &\qquad \times P(D_{j}=d_{j}\mid \overline{C}_{j}=0, \overline{L}_{j-1}=\overline{l}_{j-1}, \Delta\overline{Y}_{j-1}=\Delta\overline{y}_{j-1}, \overline{D}_{j-1}=\overline{d}_{j-1}, A=a ) \notag\\
    &\qquad \times P(L_{j-1}=l_{j-1}\mid \Delta\overline{Y}_{j-1}=\Delta\overline{y}_{j-1}, \overline{D}_{j-1}=\overline{d}_{j-1},\overline{C}_{j-1}=0,\overline{L}_{j-2}=\overline{l}_{j-2},A=a)  \notag\\\label{eq:g-formula_total_effect}
\end{align}
for intervals $i\in\{0,\dots,K+1\}$. Expression (\ref{eq:g-formula_total_effect}) is an example of a g-formula~\citep{robins_new_1986}. Another equivalent formulation is
\begin{align}
    E[\Delta Y_i^{a,\overline{c}=0}] &=  
    E\bigg[ \frac{I(A=a)I(C_i=0)}{\pi_A(A)\prod_{j=0}^i \pi_{C_j}(C_j)} \cdot \Delta Y_i \bigg] ~, \label{eq:total_effect_IPCW_discrete}
\end{align}
where 
\begin{align*}
    \pi_{C_j}(\bullet) &= P(C_j=\bullet\mid \overline{C}_{j-1},\overline{D}_{j-1},\overline{L}_{j-1}, \overline{Y}_{j-1},A)~, \\
    \pi_A(\bullet)&=P(A=\bullet\mid L_0) ~.
\end{align*}
Expression (\ref{eq:total_effect_IPCW_discrete}) is an example of an inverse probability weighted (IPW) identification formula~\citep{robins_recovery_1992,rotnitzky_semiparametric_1995,hernan_marginal_2000}. In turn, the total effect defined in (\ref{eq:total_effect_no_censoring}) under an additional intervention that ``eliminates loss to follow-up'' can be expressed as contrasts of
\begin{align*}
    E[Y_{k+1}^{a,\overline{c}=0}] = \sum_{i=0}^{k+1} E[\Delta Y_i^{a,\overline{c}=0}] 
\end{align*}
 for different levels of $a$ with $E[\Delta Y_i^{a,\overline{c}=0}]$ identified by (\ref{eq:g-formula_total_effect}) or (\ref{eq:total_effect_IPCW_discrete}).  In the survival setting, with support $Y_k\in \{0,1\}$, (\ref{eq:g-formula_total_effect}) corresponds to Expression (30) in~\citet{young_causal_2020}. A key difference from the survival setting is that the conditional probability of new recurrent events now depends on the history of the recurrent event process, which may take many possible levels, whereas in the survival setting considered by~\citet{young_causal_2020}, the terms of the relevant g-formula are restricted to those with fixed event history consistent with no failure ($\overline{Y}_k=0$).
The identification formula for the total effect on the competing event (\ref{eq:total_effect_survival}) is shown in Appendix~\ref{sec:proof_identification}.

\subsubsection{Graphical evaluation of the exchangeability conditions}

\begin{figure} 
    \centering

\subfloat[]{
\resizebox{0.8\columnwidth}{!}{
\begin{tikzpicture}
            \tikzset{line width=1.5pt, outer sep=0pt,
            ell/.style={draw,fill=white, inner sep=2pt,
            line width=1.5pt},
            swig vsplit={gap=5pt,
            inner line width right=0.5pt}};
            \node[name=L,ell, shape=ellipse] at (-3,0){$L_0$};
            \node[name=Lk1,ell, shape=ellipse] at (3,0){$L_{k-1}^{a,\overline{c}=0}$};
            \node[name=UCY,ell, shape=ellipse] at (9.5,1.7){$U_{CY}$};
            \node[name=UAY,ell, shape=ellipse] at (3,3){$U_{AY}$};
            \node[name=Yk,ell, shape=ellipse] at (6,3){$Y_k^{a,\overline{c}=0}$};
            \node[name=Yk1,ell, shape=ellipse] at (11,3){$Y_{k+1}^{a,\overline{c}=0}$};
            \node[name=Dk,ell, shape=ellipse] at (6,0){$D_k^{a,\overline{c}=0}$};
            \node[name=Dk1,ell, shape=ellipse] at (11,0){$D_{k+1}^{a,\overline{c}=0}$};
            \node[name=Ck,shape=swig vsplit] at (6,-3){
            \nodepart{left}{$C_k^{a,\overline{c}=0}$}
            \nodepart{right}{$c_k=0$} };
            \node[name=Ck1,shape=swig vsplit] at (11,-3){
            \nodepart{left}{$C_{k+1}^{a,\overline{c}=0}$}
            \nodepart{right}{$c_{k+1}=0$} };
            \node[name=A,shape=swig vsplit] at (0,0){
            \nodepart{left}{$A$}
            \nodepart{right}{$a$} };
            \begin{scope}[>={Stealth[black]},
                          every edge/.style={draw=black,very thick}]
                \path[->,>={Stealth[red]}] (UCY) edge[red] (Yk1);
                \path[->,>={Stealth[red]}] (UCY) edge[red] (Ck1.140);
                \path[->,>={Stealth[red]}] (UAY) edge[red] (Yk);
                \path[->,>={Stealth[red]}] (UAY) edge[red] (A.120);
            \end{scope}
            \begin{scope}[transparency group, opacity=0.3] 
                \path[->,>={Stealth[black]}] (L) edge [bend right = 20] (Lk1);
                \path[->,>={Stealth[black]}] (L) edge (Yk);
                \path[->,>={Stealth[black]}] (L) edge[bend left=35] (Yk1);
                \path[->,>={Stealth[black]}] (L) edge (Ck);
                \path[->,>={Stealth[black]}] (L) edge[bend right=35] (Ck1);
                \path[->,>={Stealth[black]}] (L) edge[bend left=12] (Dk);
                \path[->,>={Stealth[black]}] (L) edge[bend left=12] (Dk1);
                \path[->,>={Stealth[black]}] (A) edge (Lk1);
                \path[->,>={Stealth[black]}] (A) edge[bend right=20] (Dk);
                \path[->,>={Stealth[black]}] (A) edge[bend right=20] (Dk1);
                \path[->,>={Stealth[black]}] (A) edge (Yk);
                \path[->,>={Stealth[black]}]  (A) edge (Ck);
                \path[->,>={Stealth[black]}]  (A) edge (Ck1);
                \path[->,>={Stealth[black]}]  (A) edge (Yk1);
                \path[->,>={Stealth[black]}] (Lk1) edge (Dk);
                \path[->,>={Stealth[black]}] (Lk1) edge[bend right=15] (Dk1);
                \path[->,>={Stealth[black]}] (Lk1) edge (Yk);
                \path[->,>={Stealth[black]}]  (Lk1) edge (Ck);
                \path[->,>={Stealth[black]}]  (Lk1) edge (Ck1);
                \path[->,>={Stealth[black]}]  (Lk1) edge (Yk1);
                \path[->,>={Stealth[black]}]  (Ck) edge[bend right] (Yk);
                \path[->,>={Stealth[black]}]  (Ck) edge (Ck1);
                \path[->,>={Stealth[black]}]  (Ck.70) edge[bend right] (Dk);
                \path[->,>={Stealth[black]}]  (Dk) edge (Ck1);
                \path[->,>={Stealth[black]}]  (Dk) edge (Dk1);
                \path[->,>={Stealth[black]}]  (Dk) edge (Yk);
                \path[->,>={Stealth[black]}]  (Yk) edge (Yk1);
                \path[->,>={Stealth[black]}]  (Yk) edge (Dk1);
                \path[->,>={Stealth[black]}]  (Yk) edge (Ck1);
                \path[->,>={Stealth[black]}]  (Ck1) edge[bend right]  (Yk1);
                \path[->,>={Stealth[black]}]  (Ck1.70) edge[bend right]  (Dk1);
                \path[->,>={Stealth[black]}]  (Dk1) edge (Yk1);
            \end{scope}
        \end{tikzpicture}
}
} 
\\
\subfloat[]{
 \resizebox{0.8\columnwidth}{!}{
\begin{tikzpicture}
            \tikzset{line width=1.5pt, outer sep=0pt,
            ell/.style={draw,fill=white, inner sep=2pt,
            line width=1.5pt},
            swig vsplit={gap=5pt,
            inner line width right=0.5pt}};
            \node[name=L,ell, shape=ellipse] at (-3,0){$L_0$};
            \node[name=Lk1,ell, shape=ellipse] at (3,0){$L_{k-1}^{a,\overline{c}=0}$};
            \node[name=UY,ell, shape=ellipse] at (8.5,4.5){$U_{Y}$};
            \node[name=UD,ell, shape=ellipse] at (8.5,-1.2){$U_{D}$};
            \node[name=UDY,ell, shape=ellipse] at (13,1.5){$U_{DY}$};
            \node[name=Yk,ell, shape=ellipse] at (6,3){$Y_k^{a,\overline{c}=0}$};
            \node[name=Yk1,ell, shape=ellipse] at (11,3){$Y_{k+1}^{a,\overline{c}=0}$};
            \node[name=Dk,ell, shape=ellipse] at (6,0){$D_k^{a,\overline{c}=0}$};
            \node[name=Dk1,ell, shape=ellipse] at (11,0){$D_{k+1}^{a,\overline{c}=0}$};
            \node[name=Ck,shape=swig vsplit] at (6,-3){
            \nodepart{left}{$C_k^{a,\overline{c}=0}$}
            \nodepart{right}{$c_k=0$} };
            \node[name=Ck1,shape=swig vsplit] at (11,-3){
            \nodepart{left}{$C_{k+1}^{a,\overline{c}=0}$}
            \nodepart{right}{$c_{k+1}=0$} };
            \node[name=A,shape=swig vsplit] at (0,0){
            \nodepart{left}{$A$}
            \nodepart{right}{$a$} };
            \begin{scope}[>={Stealth[black]},
                          every edge/.style={draw=black,very thick}]
                \path[->,>={Stealth[blue]}]  (UY) edge[blue] (Yk);
                \path[->,>={Stealth[blue]}]  (UY) edge[blue] (Yk1);
                \path[->,>={Stealth[blue]}]  (UD) edge[blue] (Dk);
                \path[->,>={Stealth[blue]}]  (UD) edge[blue] (Dk1);
                \path[->,>={Stealth[blue]}]  (UDY) edge[blue] (Dk1);
                \path[->,>={Stealth[blue]}]  (UDY) edge[blue] (Yk1);
            \end{scope}
            \begin{scope}[transparency group, opacity=0.3] 
                \path[->,>={Stealth[black]}] (L) edge [bend right = 20] (Lk1);
                \path[->,>={Stealth[black]}] (L) edge (Yk);
                \path[->,>={Stealth[black]}] (L) edge[bend left=35] (Yk1);
                \path[->,>={Stealth[black]}] (L) edge (Ck);
                \path[->,>={Stealth[black]}] (L) edge[bend right=35] (Ck1);
                \path[->,>={Stealth[black]}] (L) edge[bend left=12] (Dk);
                \path[->,>={Stealth[black]}] (L) edge[bend left=12] (Dk1);
                \path[->,>={Stealth[black]}] (A) edge (Lk1);
                \path[->,>={Stealth[black]}] (A) edge[bend right=20] (Dk);
                \path[->,>={Stealth[black]}] (A) edge[bend right=20] (Dk1);
                \path[->,>={Stealth[black]}] (A) edge (Yk);
                \path[->,>={Stealth[black]}]  (A) edge (Ck);
                \path[->,>={Stealth[black]}]  (A) edge (Ck1);
                \path[->,>={Stealth[black]}]  (A) edge (Yk1);
                \path[->,>={Stealth[black]}] (Lk1) edge (Dk);
                \path[->,>={Stealth[black]}] (Lk1) edge[bend right=15] (Dk1);
                \path[->,>={Stealth[black]}] (Lk1) edge (Yk);
                \path[->,>={Stealth[black]}]  (Lk1) edge (Ck);
                \path[->,>={Stealth[black]}]  (Lk1) edge (Ck1);
                \path[->,>={Stealth[black]}]  (Lk1) edge (Yk1);
                \path[->,>={Stealth[black]}]  (Ck) edge[bend right] (Yk);
                \path[->,>={Stealth[black]}]  (Ck) edge (Ck1);
                \path[->,>={Stealth[black]}]  (Ck.70) edge[bend right] (Dk);
                \path[->,>={Stealth[black]}]  (Dk) edge (Ck1);
                \path[->,>={Stealth[black]}]  (Dk) edge (Dk1);
                \path[->,>={Stealth[black]}]  (Dk) edge (Yk);
                \path[->,>={Stealth[black]}]  (Yk) edge (Yk1);
                \path[->,>={Stealth[black]}]  (Yk) edge (Dk1);
                \path[->,>={Stealth[black]}]  (Yk) edge (Ck1);
                \path[->,>={Stealth[black]}]  (Ck1) edge[bend right]  (Yk1);
                \path[->,>={Stealth[black]}]  (Ck1.70) edge[bend right]  (Dk1);
                \path[->,>={Stealth[black]}]  (Dk1) edge (Yk1);
            \end{scope}
        \end{tikzpicture}
 }
} 
\caption{Identification of total effect. Unmeasured variables are denoted by $U_\bullet$. (A) shows unmeasured confounders (common causes of treatment, loss to follow-up, and outcomes), which violate the exchangeability conditions (\ref{eq:exchangeability_total_i})-(\ref{eq:exchangeability_total_ii}) through the red paths: $U_{AY}$ violates (\ref{eq:exchangeability_total_i}) and $U_{CY}$ violates (\ref{eq:exchangeability_total_ii}). (B) shows examples of unmeasured effect modifiers, which are common in practice (arrows from $U_{DY}$ to $Y_k^{a,\overline{c}=0}$ and $ D_k^{a,\overline{c}=0}$ are not shown to reduce clutter). In the data example, common causes of $D$ and $Y$ could for example be the previous history of cardiovascular disease and blood pressure history. The action of such unmeasured effect modifiers, shown by blue paths, does not violate any of the exchangeability conditions (\ref{eq:exchangeability_total_i})-(\ref{eq:exchangeability_total_ii}).}
\label{fig:total_effect}
\end{figure}

In Fig.~\ref{fig:total_effect} we show a single world intervention graph (SWIG) for the intervention considered under the total effect. This is a transformation~\citep{richardson_single_2013,richardson_primer_2013} of the causal DAG in Fig.~\ref{fig:DAG_observed_data}, which also includes unmeasured variables illustrating sufficient data generating models under which exchangeability conditions (\ref{eq:exchangeability_total_i})-(\ref{eq:exchangeability_total_ii}) would be violated. In particular, (\ref{eq:exchangeability_total_i})-(\ref{eq:exchangeability_total_ii}) can be 
violated by the presence of unmeasured confounders (common causes of treatment, loss to follow-up, and outcomes) such as $U_{AY}$ or $U_{CY}$ in Fig~\ref{fig:total_effect} (A). This is well-known from before, and demonstrates how SWIGs can be used to reason about the identification conditions. 

However, (\ref{eq:exchangeability_total_i})-(\ref{eq:exchangeability_total_ii}) are not violated by unmeasured common causes of the outcomes $Y_k$ and $D_k$ such as $U_Y, U_{DY}$ and $U_{D}$ in Fig.~\ref{fig:total_effect} (B),  which we often expect to be present in practice. Examples of common causes of recurrent events and death in the data example include prognostic factors related to disease progression such as previous cardiovascular disease history and blood pressure history, many of which are measured in the observed data. In contrast, the controlled direct effect and separable effects are not identified in the presence of open backdoor paths between recurrent events and death, as we will see next.

\subsection{Controlled direct effects}
The identification of the (controlled) direct effect (\ref{eq:direct_effect_no_censoring}) proceeds analogously to the total effect, with the main difference being that we also intervene to remove the occurrence of the competing event. This amounts to re-defining the censoring event as a composite of loss to follow-up and the competing event. The identification conditions then take the following form for $k\in\{1,\dots,K\}$:

\textbf{Exchangeability}
\begin{align}
    \overline{Y}_{K+1}^{a,\overline{c}=\overline{d}=0}  &\independent A | L_0 ~, \label{eq:exchangeability_direct_i}\\
    \underline{Y}_{k+1}^{a,\overline{c}=\overline{d}=0} &\independent (C_{k+1}^{a,\overline{c}=\overline{d}=0}, D_{k+1}^{a,\overline{c}=\overline{d}= 0})| \overline{L}_k^{a,\overline{c}=\overline{d}=0}, \overline{Y}_k^{a,\overline{c}=\overline{d}=0}, \overline{D}_k^{a,\overline{c}=\overline{d}=0},\overline{C}_k^{a,\overline{c}=\overline{d}=0}, A ~. \label{eq:exchangeability_direct_ii}
\end{align}

\textbf{Positivity}
\begin{align}
    &f_{A,\overline{L}_k,\overline{D}_k,\overline{C}_k,\overline{Y}_k}(a,\overline{l}_k,0,0,\overline{y}_k)>0 \notag\\
    &\qquad \implies P(C_{k+1}=0, D_{k+1}=0\mid A=a, \overline{L}_k=\overline{l}_k,\overline{D}_k=0,\overline{C}_k=0,\overline{Y}_k=\overline{y}_k) >0 ~. \label{positivity_direct_effect}
\end{align}
We also assume the positivity assumption (\ref{positivity_total_effect}), and a modified version of the consistency assumption in (\ref{eq:consistency_total_effect}) which requires us to conceptualize an intervention on the competing event (see Appendix~\ref{sec:proof_identification} for further details).

Under assumptions (\ref{eq:exchangeability_direct_ii})-(\ref{positivity_direct_effect}), an identification formula is given by
\begin{align}
E[&\Delta Y_i^{a,\overline{c}=\overline{d}=0}]= \notag\\
    &\sum_{\Delta\overline{y}_{i}}
    \sum_{\overline{l}_{i-1}}\prod_{j=0}^{i} \notag\\
    &\quad \Delta y_i\cdot P(\Delta Y_j=\Delta y_j\mid \overline{D}_{j}=
    0, \overline{C}_j=0,\overline{L}_{j-1}=\overline{l}_{j-1},\Delta \overline{Y}_{j-1}=\Delta \overline{y}_{j-1},A=a) \notag\\
    &\qquad \times P(L_{j-1}=l_{j-1}\mid \Delta\overline{Y}_{j-1}=\Delta\overline{y}_{j-1}, \overline{D}_{j-1}=
    0,\overline{C}_{j-1}=0,\overline{L}_{j-2}=\overline{l}_{j-2},A=a) ~, \notag\\\label{eq:g-formula_direct_effect}
\end{align}
or equivalently by the IPW formula
\begin{align}
    E[\Delta Y_i^{a,\overline{c}=\overline{d}=0}] &=  
    E\bigg[ \frac{I(A=a)I(C_i=0)I(D_i=0)}{\pi_A(A) \prod_{j=0}^i \pi_{C_j}(C_j)\pi_{D_j}(D_j) } \cdot \Delta Y_i \bigg]~, \label{eq:direct_effect_IPCW_discrete}
\end{align}
where we have defined
\begin{align*}
    \pi_{D_j}(\bullet) &= P(D_j=\bullet\mid \overline{C}_{j},\overline{D}_{j-1},\overline{L}_{j-1} \overline{Y}_{j-1},A) ~.
\end{align*}
For survival outcomes ($Y_k\in \{0,1\}$), (\ref{eq:g-formula_direct_effect}) reduces to Expression~(23) in~\citet{young_causal_2020}. In the absence of death and loss to follow-up and for randomized treatment assignment, both the total effect (\ref{eq:total_effect_IPCW_discrete}) and controlled direct effect (\ref{eq:direct_effect_IPCW_discrete}) reduce to $E[\Delta Y_i\mid A=a]$.

\subsubsection{Graphical evaluation of the exchangeability conditions}

Examples of unmeasured variables which violate the exchangeability conditions (\ref{eq:exchangeability_direct_i})-(\ref{eq:exchangeability_direct_ii}) are shown in Fig.~\ref{fig:direct_effect}. Importantly, (\ref{eq:exchangeability_direct_ii}) is violated by open backdoor paths between $D$ and $Y$, such as the path $D_{k+1}^{a,\overline{c}=0,\overline{d}=0}\leftarrow U_{DY} \rightarrow Y_{k+1}^{a,\overline{c}=0,\overline{d}=0}$ through the unmeasured common cause $U_{DY}$. Therefore, the exchangeability assumption for the controlled direct effect (\ref{eq:exchangeability_direct_ii}) is stronger than the exchangeability assumption for the total effect (\ref{eq:exchangeability_total_ii}). In the data example, we have measured several important common causes of acute kidney injury and death, as we will see in Sec.~\ref{sec:examples_revisited}.

\begin{figure} 
    \centering
 \resizebox{0.8\columnwidth}{!}{
\begin{tikzpicture}
            \tikzset{line width=1.5pt, outer sep=0pt,
            ell/.style={draw,fill=white, inner sep=2pt,
            line width=1.5pt},
            swig vsplit={gap=5pt,
            inner line width right=0.5pt}};
            \node[name=L,ell, shape=ellipse] at (-5,0){$L_0$};
            \node[name=Lk1,ell, shape=ellipse] at (1,0){$L_{k-1}^{a,\overline{c}=\overline{d}=0}$};
            \node[name=UY,ell, shape=ellipse] at (8.5,5){$U_{Y}$};
            \node[name=UD,ell, shape=ellipse] at (8.5,-1.2){$U_{D}$};
            \node[name=UAY,ell, shape=ellipse] at (3,3){$U_{AY}$};
            \node[name=UDY,ell, shape=ellipse] at (13,1.5){$U_{DY}$};
            \node[name=UCY,ell, shape=ellipse] at (8,1.5){$U_{CY}$};
            \node[name=Yk,ell, shape=ellipse] at (6,3){$Y_k^{a,\overline{c}=\overline{d}=0}$};
            \node[name=Yk1,ell, shape=ellipse] at (12,3){$Y_{k+1}^{a,\overline{c}=\overline{d}=0}$};
            \node[name=Dk,shape=swig vsplit] at (6,0){
            \nodepart{left}{$D_k^{a,\overline{c}=\overline{d}=0}$}
            \nodepart{right}{$d_k=0$} };
            \node[name=Dk1,shape=swig vsplit] at (12,0){
            \nodepart{left}{$D_{k+1}^{a,\overline{c}=\overline{d}=0}$}
            \nodepart{right}{$d_{k+1}=0$} };
            \node[name=Ck,shape=swig vsplit] at (6,-3){
            \nodepart{left}{$C_k^{a,\overline{c}=\overline{d}=0}$}
            \nodepart{right}{$c_k=0$} };
            \node[name=Ck1,shape=swig vsplit] at (12,-3){
            \nodepart{left}{$C_{k+1}^{a,\overline{c}=\overline{d}=0}$}
            \nodepart{right}{$c_{k+1}=0$} };
            \node[name=A,shape=swig vsplit] at (-2,0){
            \nodepart{left}{$A$}
            \nodepart{right}{$a$} };
            \begin{scope}[>={Stealth[black]},
                          every edge/.style={draw=black,very thick}]
                \path[->,>={Stealth[blue]}]  (UY) edge[blue] (Yk);
                \path[->,>={Stealth[blue]}]  (UY) edge[blue] (Yk1);
                \path[->,>={Stealth[blue]}]  (UD) edge[blue, bend left] (Dk.240);
                \path[->,>={Stealth[blue]}]  (UD) edge[blue] (Dk1);
                \path[->,>={Stealth[red]}]  (UDY) edge[red] (Dk1.120);
                \path[->,>={Stealth[red]}]  (UDY) edge[red] (Yk1.330);
                \path[->,>={Stealth[red]}]  (UCY) edge[red] (Yk1);
                \path[->,>={Stealth[red]}]  (UCY) edge[red] (Ck1.160);
                \path[->,>={Stealth[red]}]  (UAY) edge[red, bend right] (A.140);
                \path[->,>={Stealth[red]}]  (UAY) edge[red] (Yk);
            \end{scope}
            \begin{scope}[transparency group, opacity=0.3] 
                \path[->,>={Stealth[black]}] (L) edge (Yk);
                \path[->,>={Stealth[black]}] (L) edge[bend left=35] (Yk1);
                \path[->,>={Stealth[black]}] (L) edge (Ck);
                \path[->,>={Stealth[black]}] (L) edge[bend right=35] (Ck1);
                \path[->,>={Stealth[black]}] (L) edge[bend left=12] (Dk);
                \path[->,>={Stealth[black]}] (L) edge[bend left=12] (Dk1);
                \path[->,>={Stealth[black]}] (L) edge [bend right = 20] (Lk1);
                \path[->,>={Stealth[black]}] (A) edge (Lk1);
                \path[->,>={Stealth[black]}] (A) edge[bend right=20] (Dk);
                \path[->,>={Stealth[black]}] (A) edge[bend right=15] (Dk1);
                \path[->,>={Stealth[black]}] (A) edge (Yk);
                \path[->,>={Stealth[black]}]  (A) edge (Ck);
                \path[->,>={Stealth[black]}]  (A) edge (Ck1);
                \path[->,>={Stealth[black]}]  (A) edge (Yk1);
                \path[->,>={Stealth[black]}] (Lk1) edge (Dk);
                \path[->,>={Stealth[black]}] (Lk1) edge[bend right=15] (Dk1);
                \path[->,>={Stealth[black]}] (Lk1) edge (Yk);
                \path[->,>={Stealth[black]}]  (Lk1) edge (Ck);
                \path[->,>={Stealth[black]}]  (Lk1) edge (Ck1);
                \path[->,>={Stealth[black]}]  (Lk1) edge (Yk1);
                \path[->,>={Stealth[black]}]  (Ck.60) edge (Yk);
                \path[->,>={Stealth[black]}]  (Ck) edge (Ck1);
                \path[->,>={Stealth[black]}]  (Ck.60) edge (Dk.230);
                \path[->,>={Stealth[black]}]  (Dk) edge (Ck1);
                \path[->,>={Stealth[black]}]  (Dk) edge (Dk1);
                \path[->,>={Stealth[black]}]  (Dk.60) edge (Yk);
                \path[->,>={Stealth[black]}]  (Yk) edge (Yk1);
                \path[->,>={Stealth[black]}]  (Yk) edge (Dk1);
                \path[->,>={Stealth[black]}]  (Yk) edge (Ck1);
                \path[->,>={Stealth[black]}]  (Ck1.60) edge  (Yk1);
                \path[->,>={Stealth[black]}]  (Ck1.60) edge  (Dk1.230);
                \path[->,>={Stealth[black]}]  (Dk1.60) edge (Yk1);
            \end{scope}
        \end{tikzpicture}
 }
\caption{Identification of the (controlled) direct effect. In contrast to Figures~\ref{fig:total_effect}~(A) and (B), open backdoor paths between $Y$ and $D$, exemplified by the red path $D_{k+1}^{a,\overline{c}=\overline{d}=0} \leftarrow U_{DY} \rightarrow Y_{k+1}^{a,\overline{c}=\overline{d}=0}$, can violate exchangeability (\ref{eq:exchangeability_direct_ii}).}
\label{fig:direct_effect}
\end{figure}

\subsection{Separable effects} \label{subsection: separable effects identification} 
We begin by assuming the following three identification conditions.  

\textbf{Exchangeability}
\begin{align}
    &(\overline{Y}_{K+1}^{a,\overline{c}=0},\overline{D}_{K+1}^{a,\overline{c}=0},\overline{L}_{K+1}^{a,\overline{c}=0}) \independent A\mid L_0 ~, \label{eq:exchangeability_separable_i} \\
    &(\underline{Y}_{k+1}^{a,\overline{c}=0},\underline{D}_{k+1}^{a,\overline{c}=0},\underline{L}_{k+1}^{a,\overline{c}=0})\independent C_{k+1}^{a,\overline{c}=0} \mid \overline{Y}_k^{a,\overline{c}=0}, \overline{D}^{a,\overline{c}=0}_k,\overline{C}_k^{a,\overline{c}=0},\overline{L}_k^{a,\overline{c}=0},A ~. \label{eq:exchangeability_separable_ii}
\end{align}
Expressions (\ref{eq:exchangeability_separable_i})-(\ref{eq:exchangeability_separable_ii}) imply the exchangeability conditions for total effect (\ref{eq:exchangeability_total_i})-(\ref{eq:exchangeability_total_ii}) due to the decomposition rule of conditional independence.

\textbf{Positivity}
\begin{align}
    &f_{\overline{L}_k,\overline{D}_{k+1},C_{k+1},Y_k}(\overline{l}_k,\overline{d}_{k+1},0,\overline{y}_k)>0 \implies\notag\\
    &\quad P(A=a\mid \overline{D}_{k+1}=\overline{d}_{k+1},C_{k+1}=0,\overline{Y}_{k+1}=\overline{y}_{k+1},\overline{L}_k=\overline{l}_k) >0 \label{eq:positivity_separable_ii} 
\end{align}
for all $a\in\{0,1\}$, $k\in\{0,\dots,K\}$ and $L_k\in\mathcal{L}$. We also assume the positivity and consistency assumptions (\ref{positivity_total_effect})-(\ref{positivity_total_effect_ii}) and (\ref{eq:consistency_total_effect}).  Expression (\ref{eq:positivity_separable_ii}) requires that for any possibly observed level of measured time-varying covariate history amongst those who remain uncensored through each follow-up time, there are individuals with $A=0$ and $A=1$.

Consider a setting where the $A_Y$ and $A_D$ components are assigned independently one at a time. We require the following dismissible component conditions to hold for all $k\in\{0,\dots,K\}$:
\begin{align}
    &Y_{k+1}^{\overline{c}=0} \independent A_D \mid A_Y, \overline{D}_{k+1}^{\overline{c}=0} , \overline{Y}_k^{\overline{c}=0} , \overline{L}_k^{\overline{c}=0}  ~, \label{eq:dismissible_component_i} \\
    &D_{k+1}^{\overline{c}=0}  \independent A_Y  \mid A_D , \overline{D}_k^{\overline{c}=0} ,\overline{Y}_k^{\overline{c}=0} ,\overline{L}_k^{\overline{c}=0}  ~, \label{eq:dismissible_component_ii}\\
    &L_{Y,k}^{\overline{c}=0}  \independent A_D  \mid A_Y ,\overline{Y}_k^{\overline{c}=0} ,\overline{D}_k^{\overline{c}=0} ,\overline{L}_{k-1}^{\overline{c}=0} ,L_{D,k}^{\overline{c}=0}  ~, \label{eq:dismissible_component_iii} \\
    & L_{D,k}^{\overline{c}=0}  \independent A_Y \mid A_D , \overline{D}_k^{\overline{c}=0} , \overline{Y}_k^{\overline{c}=0} ,\overline{L}_{k-1}^{\overline{c}=0}  ~, \label{eq:dismissible_component_iv}
\end{align}
where we have supposed that $L_k=(L_{Y,k},L_{D,k})$ consists of components $L_{Y,k}$ and $L_{D,k}$ satisfying (\ref{eq:dismissible_component_iii})-(\ref{eq:dismissible_component_iv}) respectively. Assumptions \eqref{eq:dismissible_component_i}-\eqref{eq:dismissible_component_iv} express independencies between quantities that are observable in a future four armed trial without loss to follow-up, and can therefore be tested in such a trial. These conditions require that $\overline{L}_{D,k}$ captures all effects of $A_D$ on $\underline{Y}^{\overline{c}=0}_{k+1}$, whereas $\overline{L}_{Y,k}$ captures all effects of $A_Y$ on $\underline{D}_{k+1}^{\overline{c}=0}$. In the example on acute kidney injury discussed in Sec.~\ref{sec:examples_revisited}, we suppose that (\ref{eq:dismissible_component_i})-(\ref{eq:dismissible_component_iv}) hold with a set of baseline covariates $L_0$, $L_{Y,k}=\emptyset$ and $L_{D,k}$ given by the latest blood pressure measurement by time $k$, which influences the cardiac risk and also the perfusion of the kidneys. The implications and plausibility of this assumption in the context of the data example are discussed in Sec.~\ref{sec:simplified_estimators_partition}. Furthermore, \citet{stensrud_generalized_2021} describes a sensitivity analysis strategy for the dismissible component conditions.

Under the identification conditions for separable effects and the modified treatment assumption (\ref{eq:modified_treatment_assumption}), we have
\begin{align}
    &E[\Delta Y_i^{a_Y,a_D,\overline{c}=0}]  \notag\\
    =& \sum_{\Delta \overline{y}_i} \sum_{\overline{d}_i}\sum_{\overline{l}_{i-1}} \prod_{j=0}^{i} \notag\\
    &\Delta y_i\cdot P(\Delta Y_j=\Delta y_j\mid \overline{D}_{j}=\overline{d}_j, \overline{C}_j=0,\overline{L}_{j-1}=\overline{l}_{j-1},  \overline{Y}_{j-1}=  \overline{y}_{j-1},A=a_Y) \notag\\
    &\quad \times P(D_{j}=d_{j}\mid C_{j}=0, \overline{L}_{j-1}=\overline{l}_{j-1},  \overline{Y}_{j-1}= \overline{y}_{j-1}, \overline{D}_{j-1}=\overline{d}_{j-1}, A=a_D ) \notag\\
    &\quad \times P(L_{Y,j-1}=l_{Y,j-1}\mid  \overline{L}_{A_D,j-1}=\overline{l}_{A_D,j-1}, \overline{Y}_{j-1}= \overline{y}_{j-1}, \overline{D}_{j-1}=\overline{d}_{j-1},\overline{C}_{j-1}=0, \notag\\
   & \qquad\qquad\qquad\qquad\qquad\qquad  \overline{L}_{j-2}=\overline{l}_{j-2}, A=a_Y) \notag\\
    &\quad \times P(L_{D,j-1}=l_{D,j-1}\mid  \overline{Y}_{j-1}= \overline{y}_{j-1}, \overline{D}_{j-1}=\overline{d}_{j-1},\overline{C}_{j-1}=0,\overline{L}_{j-2}=\overline{l}_{j-2},  A=a_D)~. \notag\\ \label{eq:g-formula_separable}
\end{align}
Expression (\ref{eq:g-formula_separable}) can also be written on IPW weighted form as
\begin{align}
    E[\Delta Y^{a_Y,a_D,\overline{c}=0}_i] =  E\left[ \frac{I(A=a_Y)}{\pi_A(A)}\cdot\frac{I(C_i=0)}{\prod_{j=0}^i \pi_{C_j}(C_j)}\cdot\frac{\prod_{l=0}^i \pi_{D_l}^{a_D}}{\prod_{m=0}^i \pi_{D_m}^{a_Y}}\cdot\frac{\prod_{n=0}^{i-1} \pi_{L_{D,n}}^{a_D}}{\prod_{q=0}^{i-1} \pi_{L_{D,q}}^{a_Y}} \cdot  \Delta Y_i\right] ~,\label{eq:IPW_separable}
\end{align}
or
\begin{align}
    E[\Delta Y^{a_Y,a_D,\overline{c}=0}_i] = E\left[ \frac{I(A=a_D)}{\pi_A(A)}\cdot\frac{I(C_i=0)}{\prod_{j=0}^i \pi_{C_j}(C_j)}\cdot\frac{\prod_{l=0}^i \pi_{Y_l}^{a_Y}}{\prod_{m=0}^i \pi_{Y_m}^{a_D}}\cdot\frac{\prod_{n=0}^{i-1} \pi_{L_{Y,n}}^{a_Y}}{\prod_{q=0}^{i-1} \pi_{L_{Y,q}}^{a_D}} \cdot  \Delta Y_i\right] ~. \label{eq:IPW_separable_ii}
\end{align}
Here, we have defined
\begin{align}
    &\pi_{X_j}^{z}=f_{X_j}^z (X_j) ~\text{where}~ f_{X_j}^{z}(x)=P(X_j=x\mid \mathcal{H}_{X_j}^{C,L,Y,D},A=z) \label{eq: pi L_D}  
\end{align}
with $z=a_D,a_Y$ and $\mathcal{H}_{X_j}^{C,L,Y,D}$ being the history of $C,L,Y,D$ prior to $X_j$ (i.e. the subset containing all variables in $\{C_{k},L_k,Y_k,D_k:k=0,\dots,K+1\}$ that are ordered topologically prior to $X_j$).

The identification formula for separable effects on the competing event with time-varying covariates was first shown in~\citet{stensrud_generalized_2021} and can also be found in Appendix~\ref{sec:proof_identification}.

When full isolation holds, the identification formulas for separable effects are equal to identification formulas derived for certain path-specific effects \citep{stensrud2020conditional,robins_alternative_2011,robins2020interventionist}. Otherwise, natural direct and indirect effects are not identified because time-varying blood pressure measurements $L_k$, which are themselves affected by treatment, act as a recanting witness.

\subsubsection{Graphical evaluation of the identification conditions}
The exchangeability conditions (\ref{eq:exchangeability_separable_i})-(\ref{eq:exchangeability_separable_ii}) can be evaluated in a similar way as for the total effect in Fig.~\ref{fig:total_effect}. However, identification of the separable effects also require the dismissible component conditions to hold. These conditions can be evaluated in a DAG representing a four armed trial where the $A_D$ and $A_Y$ components can be assigned different values~\citep{stensrud_generalized_2021}, shown in  Fig.~\ref{fig:separable_effect}. Like for the controlled direct effect, $L_k$ must contain sufficient variables to block all backdoor paths between $D$ and $Y$ in order for the dismissible component conditions to hold. In particular, unmeasured common causes of $D$ and $Y$ such as $U_{DY}$ in Fig.~\ref{fig:separable_effect} can violate the dismissible component conditions.

\begin{figure} 
    \centering
\subfloat[]{
 \resizebox{0.8\columnwidth}{!}{
\begin{tikzpicture}
            \tikzset{line width=1.5pt, outer sep=0pt,
            ell/.style={draw,fill=white, inner sep=2pt,
            line width=1.5pt},
            swig vsplit={gap=5pt,
            inner line width right=0.5pt}};
            \node[name=L,ell, shape=ellipse] at (0,1.5){$L_0$};
            \node[name=Yk,ell, shape=ellipse] at (6,3){$Y_k^{\overline{c}=0}$};
            \node[name=Yk1,ell, shape=ellipse] at (11,3){$Y_{k+1}^{\overline{c}=0}$};
            \node[name=Dk,ell, shape=ellipse] at (6,0){$D_k^{\overline{c}=0}$};
            \node[name=Dk1,ell, shape=ellipse] at (11,0){$D_{k+1}^{\overline{c}=0}$};
            \node[name=Ck,shape=swig vsplit] at (6,-3){
            \nodepart{left}{$C_k^{\overline{c}=0}$}
            \nodepart{right}{$c_k=0$} };
            \node[name=Ck1,shape=swig vsplit] at (11,-3){
            \nodepart{left}{$C_{k+1}^{\overline{c}=0}$}
            \nodepart{right}{$c_{k+1}=0$} };
            \node[name=AY,ell, shape=ellipse] at (3,3){$A_Y$};
            \node[name=AD,ell, shape=ellipse] at (3,0){$A_D$};
            \node[name=MAY,ell, shape=ellipse] at (9,3/4){$M_{A_Y}$};
            \node[name=MAD,ell, shape=ellipse] at (9,9/4){$M_{A_D}$};
            \begin{scope}[>={Stealth[black]},
                          every edge/.style={draw=black,very thick}]
                \path[->,>={Stealth[red]}]  (AY) edge[red] (MAY);
                \path[->,>={Stealth[red]}]  (MAY) edge[red] (Dk1);
                \path[->,>={Stealth[red]}]  (AD) edge[red] (MAD);
                \path[->,>={Stealth[red]}]  (MAD) edge[red] (Yk1);
            \end{scope}
            \begin{scope}[transparency group, opacity=0.3] 
                \path[->,>={Stealth[black]}] (AY) edge (Yk);
                \path[->,>={Stealth[black]}] (AY) edge[bend left=15] (Yk1);
                \path[->,>={Stealth[black]}] (AD) edge (Dk);
                \path[->,>={Stealth[black]}] (AD) edge[bend right=15] (Dk1);
                \path[->,>={Stealth[black]}] (L) edge (Yk);
                \path[->,>={Stealth[black]}] (L) edge (Yk1);
                \path[->,>={Stealth[black]}] (L) edge[bend right=15] (Ck);
                \path[->,>={Stealth[black]}] (L) edge[bend right=15] (Ck1);
                \path[->,>={Stealth[black]}] (L) edge (Dk);
                \path[->,>={Stealth[black]}] (L) edge (Dk1);
                \path[->,>={Stealth[black]}]  (Ck) edge[bend right] (Yk);
                \path[->,>={Stealth[black]}]  (Ck) edge (Ck1);
                \path[->,>={Stealth[black]}]  (Ck.70) edge[bend right] (Dk);
                \path[->,>={Stealth[black]}]  (Dk) edge (Ck1);
                \path[->,>={Stealth[black]}]  (Dk) edge (Dk1);
                \path[->,>={Stealth[black]}]  (Dk) edge (Yk);
                \path[->,>={Stealth[black]}]  (Yk) edge (Yk1);
                \path[->,>={Stealth[black]}]  (Yk) edge (Dk1);
                \path[->,>={Stealth[black]}]  (Yk) edge (Ck1);
                \path[->,>={Stealth[black]}]  (Ck1) edge[bend right]  (Yk1);
                \path[->,>={Stealth[black]}]  (Ck1.70) edge[bend right]  (Dk1);
                \path[->,>={Stealth[black]}]  (Dk1) edge (Yk1);
            \end{scope}
        \end{tikzpicture}
    }
    }
\\
\subfloat[]{
 \resizebox{0.8\columnwidth}{!}{
\begin{tikzpicture}
            \tikzset{line width=1.5pt, outer sep=0pt,
            ell/.style={draw,fill=white, inner sep=2pt,
            line width=1.5pt},
            swig vsplit={gap=5pt,
            inner line width right=0.5pt}};
            \node[name=L,ell, shape=ellipse] at (0,1.5){$L_0$};
            \node[name=UDY,ell, shape=ellipse] at (13,1.5){$U_{DY}$};
            \node[name=WDY,ell, shape=ellipse] at (9,9/4){$W_{DY}$};
            \node[name=Yk,ell, shape=ellipse] at (6,3){$Y_k^{\overline{c}=0}$};
            \node[name=Yk1,ell, shape=ellipse] at (11,3){$Y_{k+1}^{\overline{c}=0}$};
            \node[name=Dk,ell, shape=ellipse] at (6,0){$D_k^{\overline{c}=0}$};
            \node[name=Dk1,ell, shape=ellipse] at (11,0){$D_{k+1}^{\overline{c}=0}$};
            \node[name=Ck,shape=swig vsplit] at (6,-3){
            \nodepart{left}{$C_k^{\overline{c}=0}$}
            \nodepart{right}{$c_k=0$} };
            \node[name=Ck1,shape=swig vsplit] at (11,-3){
            \nodepart{left}{$C_{k+1}^{\overline{c}=0}$}
            \nodepart{right}{$c_{k+1}=0$} };
            \node[name=AY,ell, shape=ellipse] at (3,3){$A_Y$};
            \node[name=AD,ell, shape=ellipse] at (3,0){$A_D$};
            \begin{scope}[>={Stealth[black]},
                          every edge/.style={draw=black,very thick}]
                \path[->,>={Stealth[red]}]  (WDY) edge[red] (Dk1);
                \path[->,>={Stealth[red]}]  (UDY) edge[red] (Dk1);
                \path[->,>={Stealth[red]}]  (WDY) edge[red] (Yk);
                \path[->,>={Stealth[red]}]  (UDY) edge[red] (Yk1);
                \path[->,>={Stealth[red]}]  (AY) edge[red] (Yk);
                \path[->,>={Stealth[red]}]  (AD) edge[red, bend right=15] (Dk1);
            \end{scope}
            \begin{scope}[transparency group, opacity=0.3] 
                \path[->,>={Stealth[black]}] (AY) edge[bend left=15] (Yk1);
                \path[->,>={Stealth[black]}] (AD) edge (Dk);
                \path[->,>={Stealth[black]}] (L) edge (Yk);
                \path[->,>={Stealth[black]}] (L) edge (Yk1);
                \path[->,>={Stealth[black]}] (L) edge[bend right=15] (Ck);
                \path[->,>={Stealth[black]}] (L) edge[bend right=15] (Ck1);
                \path[->,>={Stealth[black]}] (L) edge (Dk);
                \path[->,>={Stealth[black]}] (L) edge (Dk1);
                \path[->,>={Stealth[black]}]  (Ck) edge[bend right] (Yk);
                \path[->,>={Stealth[black]}]  (Ck) edge (Ck1);
                \path[->,>={Stealth[black]}]  (Ck.70) edge[bend right] (Dk);
                \path[->,>={Stealth[black]}]  (Dk) edge (Ck1);
                \path[->,>={Stealth[black]}]  (Dk) edge (Dk1);
                \path[->,>={Stealth[black]}]  (Dk) edge (Yk);
                \path[->,>={Stealth[black]}]  (Yk) edge (Yk1);
                \path[->,>={Stealth[black]}]  (Yk) edge (Dk1);
                \path[->,>={Stealth[black]}]  (Yk) edge (Ck1);
                \path[->,>={Stealth[black]}]  (Ck1) edge[bend right]  (Yk1);
                \path[->,>={Stealth[black]}]  (Ck1.70) edge[bend right]  (Dk1);
                \path[->,>={Stealth[black]}]  (Dk1) edge (Yk1);
            \end{scope}
        \end{tikzpicture}
 }
} 
\caption{Graphical evaluation of the dismissible component conditions (\ref{eq:dismissible_component_i})-(\ref{eq:dismissible_component_ii}) when only baseline covariates are measured. Examples of violations of the conditions are shown as red paths. (A) Dismissible component conditions can be violated by unmeasured mediators such as $M_{A_Y}$ and $M_{A_D}$: (\ref{eq:dismissible_component_i}) is violated by the path $A_D\rightarrow M_{A_D}\rightarrow Y_{k+1}^{\overline{c}=0}$ and (\ref{eq:dismissible_component_ii}) is violated by the path $A_Y\rightarrow M_{A_Y} \rightarrow D_{k+1}^{\overline{c}=0}$. However, if $M_{A_Y}$ and $M_{A_D}$ were measured and included in $L_Y$ and $L_D$ respectively, the dismissble conditions conditions would hold. (B) Assumption (\ref{eq:dismissible_component_i}) is violated by open backdoor paths between $Y$ and $D$, such as the path $A_D\rightarrow D_{k+1}^{\overline{c}=0} \leftarrow U_{DY} \rightarrow Y_{k+1}^{\overline{c}=0}$, a collider path which opens when conditioning on $D_{k+1}^{\overline{c}=0}$. Likewise, (\ref{eq:dismissible_component_ii}) is violated by the path $A_Y\rightarrow Y_k^{\overline{c}=0}\leftarrow W_{DY}\rightarrow D_{k+1}^{\overline{c}=0}$.}
\label{fig:separable_effect}
\end{figure}

\subsection{Correspondence with continuous time estimands\label{sec:correspondence}}
Up to this section, we have considered a fixed time grid where the duration of each interval is 1 unit of time. In this section we will consider limiting cases of the identification results where we allow the grid-spacing to become arbitrarily small.  Let the endpoints of the intervals $k\in\{0,\dots,K+1\}$ correspond to times $\{0,t_1,\dots,t_{K+1}\}\subseteq[0,\infty)$. As before, we assume the duration of all intervals is equal, and denote this by $\Delta t$.

We can associate the counterfactual quantities considered thus far in discrete time with corresponding quantities in the counting process literature. An overview of the corresponding quantities is presented in Table~\ref{tab:correspondence_counterfactuals_classical}. Here, we use the term 'factual quantities' to denote variables that take their natural values, i.e. quantities that are not subject to any counterfactual intervention (see ~\citet{richardson_single_2013} for a formal definition of natural value). These are different from observed quantities, which only contain the factual events that have been recorded in subjects that are under follow-up.

\begin{table}[htbp]
  \centering
  \caption{Correspondence between discrete time quantities and continuous time quantities.  $C$ and $T^D$ are the time of loss to follow-up and time of the competing event respectively, and $L_t$ is the process of measured covariates by time $t$.  The use of the superscript $c$ in the right column represents a \emph{complete}  process, that is, a process where no individuals are lost to follow-up. }
     \resizebox{\columnwidth}{!}{
    \begin{tabular}{|p{3cm}|p{4cm}|p{5cm}|}
    \hline
          & \textbf{Discrete quantity} & \textbf{Quantity in the counting process literature} \\
    \hline
    \multirow{2}{3cm}{\textbf{Counterfactual quantity}} & $Y_k^{\overline{c}=0}$     &  \\\cline{2-3}
          & $D_k^{\overline{c}=0}$     &  \\\hline
    \multirow{3}{3.5cm}{\textbf{Factual quantity}} 
          & $Y_k$    &  $N^c_{t_k}$ \\\cline{2-3}
          & $D_k$    &  $I(T^D \leq t_k)$ \\\hline
    \multirow{2}{3cm}{\textbf{Observed quantity}}&  $C_k$    & $\int_0^{t_k} I(T^D\geq t)dC_t$ \\\cline{2-3}
          &  $\sum_{j=1}^k I(C_{j}=0) \Delta Y_j$    & $\int_0^{t_k} I(C\geq t)dN^c_t$ \\\cline{2-3}
            &   $\sum_{j=1}^k I(C_{j}=0) \Delta D_j$  
          & $\int_0^{t_k} I(C\geq t)dI(T^D \leq t)$ \\\hline
    \end{tabular}%
      }
    
  \label{tab:correspondence_counterfactuals_classical}
\end{table}%

Importantly, quantities indexed by the superscript $\overline{c}=0$ are controlled direct effects with respect to an intervention which eliminates loss to follow-up, and do not have an analog in the existing counting process literature. This includes the quantity $C_k^{\overline{c}=0}$, which is the counterfactual value of the censoring indicator for interval $k$ under an intervention that eliminates censoring in previous intervals.

\subsubsection{Correspondence of identification conditions}
In the counting process literature, it is usual (see e.g. \citet{aalen_survival_2008} and~\citet[Expression 7.22]{cook_statistical_2007}) to identify the intensity of the complete (i.e. uncensored) counting process as a function of the intensity of the observed (i.e. censored) counting process, using the independent censoring assumption
\begin{align}
    \lambda^{\mathcal{F}^c}_t = \lambda^\mathcal{G}_t ~, \label{eq:continuous_independent_censoring}
\end{align}
where
\begin{align*}
    \lambda^{\mathcal{F}^c}_tdt &= E[dN^c_t\mid \mathcal{F}^c_{t^-}]~, \quad \mathcal{F}^c_{t^-}=\sigma(L_u, A, N^c_u, I(T^D\geq u); 0\leq u < t ) ~,  \\
    \lambda^{\mathcal{G}}_t dt &= E[dN^c_t\mid \mathcal{G}_{t^-}]~, \quad \mathcal{G}_{t^-}=\sigma(L_u, A, N^c_u, I(T^D\geq u), I(C\geq u); 0\leq u < t ) ~.
\end{align*}
A corresponding formulation of (\ref{eq:continuous_independent_censoring}) within the discrete time 
framework is
\begin{align}
    \frac{1}{\Delta t} \cdot  E[\Delta Y_{j} \mid \overline{D}_{j}, \overline{L}_{j-1}, \overline{Y}_{j-1}, A ] = \frac{1}{\Delta t} \cdot
     E[\Delta Y_{j} \mid \overline{D}_{j},\overline{C}_{j}, \overline{L}_{j-1}, \overline{Y}_{j-1}, A] ~. \label{eq:discrete_correspondence_expectation}
\end{align}
We assume that the possibility of experiencing more than one recurrent event during a single interval becomes negligible, i.e. $\Delta Y_j\in\{0,1\}$, for fine discretizations. Thus, for small $\Delta t$, (\ref{eq:discrete_correspondence_expectation}) is closely related to
\begin{align}
    &\frac{1}{\Delta t} \cdot  P( \Delta Y_j =\Delta y_j \mid \overline{D}_{j}, \overline{L}_{j-1}, \overline{Y}_{j-1}, A ) = \frac{1}{\Delta t} \cdot
    P(\Delta Y_{j}=\Delta y_j \mid \overline{D}_{j},\overline{C}_{j}, \overline{L}_{j-1}, \overline{Y}_{j-1}, A) ~. \label{eq:discrete_independent_censoring}
\end{align}

We show in Appendix~\ref{sec:app_correspondence_censoring} that when the random variables in (\ref{eq:discrete_independent_censoring}) are generated under an FFRCISTG model, and when consistency (\ref{eq:consistency_total_effect}) and faithfulness hold, then exchangeability with respect to censoring (\ref{eq:exchangeability_total_ii}) is implied by (\ref{eq:discrete_independent_censoring}).
In plain English, this result states that a discrete time analog of the independent censoring assumption implies the absence of backdoor paths between $C_i^{a,\overline{c}=0}$ and $Y_{j}^{a,\overline{c}=0}$ for all $i\leq j$  in Fig.~\ref{fig:total_effect}. However, the reverse implication does not follow, as effects of $C_i$ on future $\Delta Y_j$ (i.e. the presence of a path $C_i\rightarrow \Delta Y_j$ for $i\leq j$ in a DAG) violates (\ref{eq:discrete_independent_censoring})  without violating (\ref{eq:exchangeability_total_ii}). The path $C_i\rightarrow \Delta Y_j$ for $i\leq j$ could represent the presence of concomitant care which affects the recurrent outcome, and the consequences of such a path for the interpretation of discrete versus continuous time estimands are clarified in Sec.~\ref{sec:diff_interpretation}.  A similar correspondence of identification conditions exists for the competing event~\citep{robins_correcting_2000}, and is stated in Appendix~\ref{sec:app_correspondence_censoring}.

\subsubsection{Correspondence of identification formulas}
\label{section:correspondence of identification formulas}
In this section, we consider identifying functionals in the limit of fine discretizations of time. Justifications for the results are given in  Appendix~\ref{sec:proof_identification}.

In the limit of fine discretizations,  $E[Y_k^{a,\overline{c}=0}]$ can be formulated as 
\begin{align}
  \int_0^{t_k}\prodi_{s < u}[1- dA_s^D(a)] \cdot E[W_A \mathcal{W}_{C,u-} dN_u \mid C\geq u, T^D\geq u,A=a] ~, \label{eq:total_weighted_form_continuous}
\end{align}
where 
\begin{align}
    \mathcal{W}_{C,t} &= \frac{\prodi_{u\leq t} [1-d \Lambda_u^C]}{\prodi_{u\leq t} \left[1-d \Lambda_u^{C\mid\mathcal{F}}\right]} ~, \label{eq: continuous-time dot weights}
\end{align}
$dA_t^D(a) =  P( T^D \in [t,t+dt)| T^D \geq t, C \geq t, A=a)$ and $W_A=\frac{P(A=a)}{\pi_A(A)}$. In this setting, $\Lambda_t^C$ and $\Lambda_t^{C|\mathcal F}$ are the compensators of $N^C$ with respect to $\mathcal F_{t}^{C,D,A}$ and $\mathcal F_{t}^{L,Y,C,D,A}$, which heuristically means that
\begin{align}
    d \Lambda_t^C &= P(C \in [t,t+dt)|\mathcal F_{t-}^{C,D,A}) ~, \label{eq:lambda_censoring_i}\\
    d \Lambda_t^{C|\mathcal F} &= P(C \in [t, t+dt) |\mathcal F_{t-}^{L,Y,C,D,A}) ~. \label{eq:lambda_censoring_ii}
\end{align} 
Here, $\mathcal F_{t}^{B}$ denotes the filtration generated by the collection of variables and processes $B$. Expression (\ref{eq:total_weighted_form_continuous}) is equivalent to
\begin{align}
    \int_0^{t_k} \frac{E\bigg[  \frac{I(A=a)I(C > u)}{ \pi_A(A)\prodi_{s< u} (1-d\Lambda_s^{C\mid\mathcal{F}}) } \cdot dN_u \bigg]  }{E\bigg[  \frac{I(A=a)I(C> u)}{ \pi_A(A)\prodi_{\tau< u} (1-d\Lambda_\tau^{C\mid\mathcal{F}}) }  \bigg] } ~.\label{eq:total_effect_IPCW}
\end{align}
The product-integral terms are covariate-specific survival functions with respect to the censoring event.  Expression (\ref{eq:total_effect_IPCW}) corresponds to Expression (7.29) in~\citet{cook_statistical_2007} and targets a setting commonly called 'dependent censoring' in the counting process literature. 

Under the strengthened independent censoring assumption
\begin{align}
    C\independent(N^c_t, T^D)\mid A ~, \label{eq:strong_indep_censoring}
\end{align}
which implies (\ref{eq:continuous_independent_censoring}) without any covariates ($L_t=\emptyset$), we have that $\mathcal{W}_{C,t}=1$ with $L_t=\emptyset$. Furthermore, in settings where treatment $A$ is assigned by randomization, we have that $W_A=1$. Consequently, (\ref{eq:total_weighted_form_continuous}) reduces to 
\begin{align}
  \int_0^{t_k} \prodi_{u < t} [1-dA_u^D(a)] E\left[ dN_u \mid  T^D\geq u, C \geq u,A=a \right] ~.\label{eq:total_effect_continuous}
\end{align} 
Expression (\ref{eq:total_effect_continuous}) corresponds to Expression (13) in \citet{cook_marginal_1997}. 

The controlled direct effect (with respect to interventions on the competing event) can be viewed as a special case of the total effect, where 1) we re-define the censoring event as a composite of loss to follow-up and the competing event, hence the censoring indicator takes the form $I(C\wedge T^D \leq t)$, and 2) we re-define the "competing" event as an event that never occurs. Under 1) and 2), Expression (\ref{eq:strong_indep_censoring}) becomes
\begin{align}
    N^c_t \independent C\wedge T^D\mid A ~, \label{eq:strong_indep_censoring_direct}
\end{align}
and (\ref{eq:total_effect_continuous}) reduces further to
\begin{align}
  \int_0^{t_k} E\left[ dN_u \mid  T^D \geq u, C \geq u,A=a \right] ~. \label{eq:direct_effect_continuous}
\end{align}
Expression (\ref{eq:direct_effect_continuous}) is the continuous time limit of the controlled direct effect (\ref{eq:direct_effect_IPCW_discrete}) if (\ref{eq:strong_indep_censoring_direct}) is satisfied for fine discretizations. It corresponds to the quantity $R(t)$ in~\citet{andersen_modeling_2019} and is described by \citet{cook_marginal_1997} as a measure of the expected number of events for subjects at risk over the entire observation period, under the condition that the recurrent event is independent of the competing event.

The continuous time limit of the identification formula for separable effects (\ref{eq:IPW_separable}) is given by
\begin{align*}
   \int_0^{t_k} &\prodi_{s < u}[1- dA_s^D(a_Y)] \\
   &\qquad\times E\left[  
     W_A \mathcal{W}_{C,u-}  \mathcal W_{D,u-}(a_Y,a_D)  \mathcal{W}_{L_{D,u-}}(a_Y,a_D) dN_u \mid  C\geq u,T^D \geq u,A=a_Y \right] 
\end{align*}
and
\begin{align*}
   \int_0^{t_k} &\prodi_{s < u}[1- dA_s^D(a_D)] \\
   &\qquad\times E\left[  
     W_A \mathcal{W}_{C,u-} \theta_u^{Y}  \mathcal W_{Y,u-}(a_Y,a_D)  \mathcal{W}_{L_{Y,u-}}(a_Y,a_D) dN_u \mid  C\geq u,T^D \geq u,A=a_D \right] ~.
\end{align*}
The weights $\mathcal W_D(a_Y,a_D)$ and $\mathcal W_Y(a_Y,a_D)$ take the form
\begin{align*}
    \mathcal W_{D,t}(a_Y,a_D) = \theta_t^D \cdot \frac{ \prodi_{u\leq t}  \left[1-d \Lambda_u^{D\mid\mathcal{F}}(a_D) \right] }{ \prodi_{u\leq t} \left[1-d \Lambda_u^{D\mid\mathcal{F}}(a_Y)\right]}~,~ \mathcal{W}_{Y,t}(a_Y,a_D)= \prod_{s\leq t} \theta_s^{Y} \frac{\prodi_{u\leq t} [1-d \Lambda_u^{Y\mid\mathcal{F}}(a_Y)]}{\prodi_{u\leq t} \left[1-d \Lambda_u^{Y\mid\mathcal{F}}(a_D)\right]} 
\end{align*}
where the compensators $\Lambda^{D|\mathcal{F}}$ and $\Lambda^{Y|\mathcal{F}}$ are defined analogously to (\ref{eq:lambda_censoring_ii}). $\Lambda_u^{D\mid\mathcal{F}}(a)$ is understood as the random function $\Lambda_u^{D\mid\mathcal{F}}$ evaluated in the argument $A=a$ (and likewise for $\Lambda_u^{Y|\mathcal{F}}(a)$). Furthermore, $\theta_t^D = \Big( \frac{d \Lambda_t^{D\mid\mathcal{F}}(a_D)}{ d \Lambda_t^{D\mid\mathcal{F}}(a_Y)  }\Big)^{I(T^D \leq t)}$ and $\theta_t^Y = \Big( \frac{d \Lambda_t^{Y\mid\mathcal{F}}(a_Y)}{ d \Lambda_t^{Y\mid\mathcal{F}}(a_D)  }\Big)^{N_t-N_{t-}}$.

The mathematical characterization of the limit $\mathcal W_{L_{D},t}(a_Y,a_D)$ of 
$\prod_{j=0}^i \pi_{L_{D,j}}^{a_D}/\prod_{k=0}^i \pi_{L_{D,k}}^{a_Y}$, where $\pi_{L_{D,j}}^{\bullet}$ is defined in (\ref{eq: pi L_D}),
depends on what type of process $L_{D}$ is. Many applications are covered when $L_{D}$ is a marked point process on a finite mark space. That is, $L_{D}$ takes values in a finite number of marks but can jump between marks over time. We will assume $L_{D}$ is such a process in Sec. \ref{sec:estimation_maintext}. The same considerations also apply to the limit $\mathcal{W}_{L_{Y},t}(a_Y,a_D)$ of $\prod_{j=0}^i \pi_{L_{Y,j}}^{a_Y}/\prod_{k=0}^i \pi_{L_{Y,k}}^{a_D}$. These weights are closely related to the mediation weights considered by~\citet{zheng_longitudinal_2017,mittinty_longitudinal_2020,tchetgen_tchetgen_inverse_2013}.

Finally, a product integral representation of the total effect on the competing event is given in Appendix~\ref{sec:proof_identification}. 

In Table~\ref{tab:estimands}, we show an overview of the correspondence between the causal estimands discussed in Sec.~\ref{sec:identification} and common estimands that appear in the statistical literature.
\begin{table}[h!]
    \centering
    \caption{A mapping of common recurrent events estimands in the literature to their counterfactual definition of risk. The third row shows a new proposed estimand, the \emph{separable effects for recurrent events}, based on~\citet{stensrud_separable_2020-1}.}
     \resizebox{\columnwidth}{!}{\begin{tabular}{ |p{ 2 cm }|p{ 4 cm }|p{ 6 cm }|}
          \hline
         \textbf{Definition}     & \textbf{Description}  & \textbf{Alternative terminology} \\
         \hline
         $E[Y_k^{a, \overline{c}=0}]$   & Expected event count without elimination of competing events   & Marginal mean~\citep{cook_marginal_1997}, treatment policy strategy~\citep{schmidli_estimands_2021}\\
         \hline
         $E [Y_k^{a,\overline{d}=0,\overline{c}=0}]$   &   Expected event count with elimination of competing events  &  Cumulative rate~\citep{ghosh_nonparametric_2000,cook_marginal_1997}, hypothetical strategy~\citep{schmidli_estimands_2021} \\
         \hline
         $E[Y_{k}^{a_Y,a_D, \overline{c}=0}]$   & Expected event count under a decomposed treatment &  Does not correspond to classical estimands \\
          \hline 
\end{tabular}}

    \label{tab:estimands}
\end{table}

\subsubsection{Differences in interpretation\label{sec:diff_interpretation}} 
In the counting process formalism of recurrent events, $N^c_t$ is interpreted as the count of events that would be measured if we somehow could observe every individual's future outcomes (for example by implanting a 'tracker device'), even if they withdraw from study participation or otherwise discontinue follow-up. This is a factual (as opposed to a counterfactual) quantity, because it is not subject to any counterfactual intervention to eliminate censoring. Next, the observed counting process $N_t$ is interpreted as the number of events that were recorded while the subject was alive and under follow-up, i.e. $N_t=\int_0^t I(T^D\geq s, C\geq s)dN^c_s$.

If study participants receive concomitant care by virtue of being under follow-up (e.g. additional medical exams that can lead to discovery of new conditions which trigger initiation of additional, supportive treatments), then individuals who are lost to follow-up may have different outcomes $N^c_t$ compared to subjects under follow-up due to the termination of such concomitant care. This violates the independent censoring condition (\ref{eq:continuous_independent_censoring}). Therefore, $E[N_t^c]$ is not identified when concomitant care under follow-up affects future outcomes $N^c_t$ without additional strong assumptions. In other words, one cannot make inference on individuals who are censored (who do not receive concomitant care) by only observing uncensored individuals (who do receive concomitant care). 

In contrast to $N^c_t$,  the counterfactual quantity $Y^{\overline{c}=0}_k$ is often interpreted as the number of recurrent events that would be observed by time $k$ under an intervention which \emph{prevented} individuals from being lost to follow-up, i.e. in a pseudopopulation where all individuals receive the same level of the primary intervention ($A$) and concomitant care. $E[Y^{\overline{c}=0}_k]$ is still identified under effects of concomitant care on the recurrent event, i.e. the arrows $c_k=0\rightarrow Y_k^{a,\overline{c}=0}$ in Fig.~\ref{fig:total_effect} do not violate the exchangeability condition (\ref{eq:exchangeability_total_ii}). In the special case where concomitant care does not affect future recurrent events, the interpretations of $E[Y^{\overline{c}=0}_k]$ and $E[N^c_{t_k}]$ coincide.  Similar arguments are given by \cite{young_causal_2020}, Sec. 5, for the incident event setting.

\section{Estimation \label{sec:estimation_maintext}}
The identification formulas in Sec. \ref{sec:identification} motivate a variety of estimators that have been presented in the literature; examples can be found in \citet{young_causal_2020,stensrud_generalized_2021,martinussen2021estimation}.

In survival and event history analysis, researchers have traditionally been accustomed to estimands and estimators defined in continuous time. We mapped out correspondences between the discrete time identification formulas and their continuous time limits in Sec.~\ref{section:correspondence of identification formulas}. Next, we will consider\footnote{Instead of considering estimators defined in continuous time, it would also possible to construct estimators targeting the discrete identification formulas in Sections \ref{sec:identification_total_effect}-\ref{subsection: separable effects identification}, similarly to \citet{young_causal_2020,stensrud_generalized_2021}.} the following general estimator in continuous time, applicable to several of the estimands considered above, 
\begin{align}
\begin{pmatrix}
 \hat Y_t \\ 
\hat S_t \\
\hat D_t 
\end{pmatrix} &= \begin{pmatrix}
0 \\ 1 \\ 0 
\end{pmatrix} + \int_0^t \begin{pmatrix}
\hat S_{s-} & 0 & 0  \\
0 & -\hat S_{s-} & 0  \\
0 & 0 & \hat S_{s-}  
\end{pmatrix} d\begin{pmatrix}
\hat B^Y_s \\ \hat B_s^D \\ \hat B_s^{D,w} 
\end{pmatrix}. \label{eq: mean frequency estimator}
\end{align}
Here, $\hat Y_t$ is an estimator of a counterfactual mean frequency function under interventions of interest and $\hat S_t$ is an auxiliary quantity 
used  to define the system in (\ref{eq: mean frequency estimator}). 
Finally, $\hat D_t$ is an estimator of a counterfactual competing event process under interventions of interest. 

The stochastic differential equation (\ref{eq: mean frequency estimator}) is uniquely determined by the integrators. Thus, presenting different estimators on this form amounts to presenting different integrators. We restrict the focus to the case with no tied event times in the remainder of this section for ease of presentation.

\subsection{Risk set estimators}
Identification formulas of the form (\ref{eq:total_weighted_form_continuous}), where the integrator conditions on the at-risk event $\{ T^D \geq t, C \geq t \}$, motivate the risk set estimators
\begin{align}
    \begin{pmatrix}
        \hat B^Y_t \\ \hat B_t^D \\ \hat B_t^{D,w}
    \end{pmatrix}
     &=  \sum_{i=1}^n \int_0^t \begin{pmatrix}
        \frac{\hat{\bar{\theta}}_s^i \hat R_{s-}^{i} I(A_i=a) Z_s^i }{\sum_{j=1}^n I(A_j=a) Z_s^j } & 0 \\
        0 & \frac{ \hat R_{s-}^{i,D} I(A_i=a) Z_s^i }{\sum_{j=1}^n  I(A_j=a)  Z_s^j } 
        \\
        0 & \frac{ \hat{\bar{\theta}}_s^i \hat R_{s-}^{i} I(A_i=a) Z_s^i }{\sum_{j=1}^n  I(A_j=a)  Z_s^j } 
     \end{pmatrix} d\begin{pmatrix}
            N_s^{i} \\ N_s^{D,i}
     \end{pmatrix}  ~,
    \label{eq: weighted risk set estimator}
\end{align}
where $Z^i_t=I(T^{D,i}\geq t, C^i \geq t)$ is the at-risk process. Here, $N_t^i=N_{t\wedge C_i}^{c,i}$ is the observed counting process for the recurrent event, $N^{D,i}_t = I(T^{D,i} \leq t, T^{D,i} < C^i)$ is the observed counting process for death, and $\hat R^i, \hat R^{i,D}$ are estimated weight processes of individual $i$ (see Table \ref{tab:weights}). The $\hat{\bar{\theta}}^i$ terms, specified in Table \ref{tab:weights}, are needed when the driving counting processes share jump times with the weights, which is the case for some separable effects estimands, as seen in Section \ref{section:correspondence of identification formulas}.

\subsection{Horvitz–Thompson and Hajek estimators}
Identification formulas of the form (\ref{eq:total_effect_IPCW}) (which coincides with the discrete time formulas (\ref{eq:total_effect_IPCW_discrete}) in the case of the total effect, (\ref{eq:direct_effect_IPCW_discrete}) for the controlled direct effect and (\ref{eq:IPW_separable})-(\ref{eq:IPW_separable_ii}) for the separable effect) motivate Hajek  estimators \citep{godambe_comment_1971} and Horwitz-Thompson estimators \citep{horvitz_generalization_1952}, which give the integrators
\begin{align}
    \begin{pmatrix}
        \hat B^Y_t \\ \hat B_t^D \\ \hat B_t^{D,w}
    \end{pmatrix}
     &=  \frac{1}{n} \sum_{i=1}^n \int_0^t \frac{\hat{\bar{\theta}}_s^i \hat{ \overline R}_{s-}^{i} I(A_i=a) }{ H_{s-} }  \begin{pmatrix}
        1 & 0 \\
        0 & 0
        \\
        0 & 1
     \end{pmatrix} d\begin{pmatrix}
            N_s^{i} \\ N_s^{D,i}
     \end{pmatrix} .
    \label{eq: weighted rate estimator Hajek and H-T}
\end{align}
In the above expression, $H_{t} = \frac{1}{n} \sum_{j=1}^n \hat{ \overline R}_{t}^j I(A_j=a)$ gives Hajek estimators, and $H_{t}=1$ gives Horvitz–Thompson estimators. $\hat{ \overline{R}}^i$ is an estimated weight processes for individual $i$ (see Table \ref{tab:weights}). These estimators are closely related to previously studied inverse probability weighted estimators~\citep{robins_recovery_1992,rotnitzky_semiparametric_1995,hernan_marginal_2000} and proportional odds estimators~\citep{zheng_longitudinal_2017,mittinty_longitudinal_2020,tchetgen_tchetgen_inverse_2013}.

The estimator defined by (\ref{eq: mean frequency estimator}) may be unfamiliar to some practitioners, but it has the following properties:
\begin{itemize}
    \item The estimator is generic in the sense that, given  weight estimators it can be used to estimate the total effect, the controlled direct effect, and the separable effect, and other composite estimands (e.g. the `while alive' strategy) as defined in Sec. \ref{sec:estimands_with_competing_events}.
    \item Expression (\ref{eq: mean frequency estimator}) is easy to solve on a computer, as it defines a recursive equation that can be solved using e.g. a \texttt{for} loop. General software that can be used to solve systems like (\ref{eq: mean frequency estimator}) is  available for anyone to use at \texttt{github.com/palryalen/}.
\end{itemize}

In Theorem \ref{thm:convergence} in Appendix \ref{sec:appendix estimation} we provide convergence results for the estimators in (\ref{eq: mean frequency estimator})-(\ref{eq: weighted risk set estimator}) for the case when the true weights are not known, but estimated. Convergence is guaranteed when the weight estimators $\hat R_t$, $\hat R_t^D$, and $\hat{ \overline{R}}_t$ converge in probability to the true weights for each fixed $t$, which is established for the additive hazard weight estimator we will consider in Sec. \ref{subsection: estimating the weights}
\cite[Theorem 2]{ryalen_additive_2019}.

In Table \ref{tab:weights} we present pairs of weights $R^i, R^{i,D}$, and $ \overline R^i$ as well as the parameter $\bar \theta^{i}_t $ that can be used in (\ref{eq: mean frequency estimator})-(\ref{eq: weighted risk set estimator}) to estimate the total effect, the direct effect, and the separable effects as defined in Sec. \ref{sec:estimands_with_competing_events}. Define $\mathcal{W}_{D}$, the weights associated with the intervention that prevents death from other causes, similarly to the censoring weights in (\ref{eq: continuous-time dot weights}),
\begin{align*}
    \mathcal{W}_{D,t} &= \frac{\prodi_{u\leq t} [1-d \Lambda_u^D]}{\prodi_{u\leq t} \left[1-d \Lambda_u^{D\mid\mathcal{F}}\right]} ~,
\end{align*}
where $\Lambda_t^D$ is the compensator of $N^D$ with respect to $\mathcal F_{t}^{C,D,A}$, defined analogously to (\ref{eq:lambda_censoring_i}). 
$\overline {\mathcal W}_{C}$ and $\overline{\mathcal W}_{D}$ are the unstabilized versions of these weights, defined as
\begin{align*}
    \overline {\mathcal{W}}_{C,t} = \frac{I(C > t)}{\prodi_{u\leq t} \left[1-d \Lambda_u^{C\mid\mathcal{F}}\right]}~, \hspace{2 cm} \overline {\mathcal{W}}_{D,t} = \frac{I(T^D > t)}{\prodi_{u\leq t} \left[1-d \Lambda_u^{D\mid\mathcal{F}}\right]} ~.
\end{align*}

\begin{table}[h!]
    \centering
    \caption{Weights $R^i$, $R^{i,D}$ and $\overline{R}^{i}$ as well as $\bar \theta^{i}_t $  in (\ref{eq: weighted risk set estimator}) and (\ref{eq: weighted rate estimator Hajek and H-T}) for estimating the total effect, the controlled direct effect, and the separable effects, respectively. 
}
    \resizebox{\columnwidth}{!}{
    \begin{tabular}{ |p{ 1.8 cm }|p{ 11 cm }|p{ 0.3 cm }| p{0.2 cm}|}
         \hline 
         \textbf{Estimand}     &  \scalebox{0.9}{$\boldsymbol{ R_{t}^{i}},\boldsymbol{ \overline{R}_{t}^{i}}$}  &   \scalebox{0.9}{$\boldsymbol{ R_{t}^{i,D}}$} & \scalebox{0.9}{$\boldsymbol{ \bar \theta^{i}_t }$} \\
         \hline
         $E[Y_k^{a,\overline{c}=0}]$  & $  R_{t}^{i} = {\mathcal W}_{C,t}^i \cdot   W_A^{i}$ & $1$ & 1 \\ 
         & $\overline{R}_{t}^{i} =\overline{ \mathcal W }_{C,t}^i \cdot
         \frac{1}{ \pi_A(A_i)}$ &  $-$ & 1
         \\\hline
         $E[Y_k^{a,\overline{d}=\overline{c}=0}]$  & $ R_{t}^{i}= {\mathcal W}_{C,t}^i \cdot W_A^{i} \cdot \mathcal W_{D,t}^i$ & $0$ & 1 \\ 
         & $\overline{R}_{t}^{i} = \overline{ \mathcal W }_{C,t}^i \cdot \frac{1}{ \pi_A(A_i)}\cdot \overline{\mathcal{W}}_{D,t}^i$ &  $-$ & 1
         \\\hline
         \multirow{4}{4em}{\scalebox{0.9}{$E[Y_k^{a_Y,a_D,\overline{c}=0}]$}} & $R_{t}^{i} = {\mathcal W}_{C,t}^i \cdot  I(a=a_Y) W_A^{i} \cdot  { \mathcal W}_{D,t}^i(a_Y,a_D) \cdot {\mathcal W}_{L_{D},t}^i(a_Y,a_D) $ & 1 & 1 \\
         & $\overline{R}_{t}^{i} =\overline{\mathcal W}_{C,t}^i \cdot \frac{ I(a=a_Y)}{\pi_A(A_i)} \cdot  { \mathcal W}_{D,t}^i(a_Y,a_D) \cdot {\mathcal W}_{L_{D},t}^i(a_Y,a_D) $ & $-$ & 1 \\
         \cline{2-4}
         
         & $R_{t}^{i} = {\mathcal W}_{C,t}^i \cdot  I(a=a_D) W_A^{i} \cdot  { \mathcal W}_{Y,t}^i(a_Y,a_D) \cdot {\mathcal W}_{L_{Y},t}^i(a_Y,a_D) $ & 1 & $\theta_t^{Y,i}$ \\
         & $\overline{R}_{t}^{i} =\overline{\mathcal W}_{C,t}^i \cdot \frac{ I(a=a_D)}{\pi_A(A_i)} \cdot  { \mathcal W}_{Y,t}^i(a_Y,a_D) \cdot {\mathcal W}_{L_{Y},t}^i(a_Y,a_D) $ & $-$ & $\theta_t^{Y,i}$
         \\\hline
         $E[D_k^{a,\overline{c}=0}]$  & $  R_{t}^{i} = {\mathcal W}_{C,t}^i \cdot   W_A^{i}$ & $1$ & 1 \\ 
         & $\overline{R}_{t}^{i} =\overline{ \mathcal W }_{C,t}^i \cdot
         \frac{1}{ \pi_A(A_i)}$ &  $-$ & 1
         \\\hline
         \multirow{4}{4em}{\scalebox{0.9}{$E[D_k^{a_Y,a_D,\overline{c}=0}]$}} & $R_{t}^{i} = {\mathcal W}_{C,t}^i \cdot  I(a=a_Y) W_A^{i} \cdot   { \mathcal W}_{D,t}^i(a_Y,a_D) \cdot {\mathcal W}_{L_{D},t}^i(a_Y,a_D) $ & 1 & $\theta_t^{D,i}$ \\
         & $\overline{R}_{t}^{i} =\overline{\mathcal W}_{C,t}^i \cdot \frac{I(a=a_Y)}{\pi_A(A_i)}  \cdot   { \mathcal W}_{D,t}^i(a_Y,a_D) \cdot {\mathcal W}_{L_{D},t}^i(a_Y,a_D) $ & $-$ & $\theta_t^{D,i}$ \\
         \cline{2-4}
         & $R_{t}^{i} = {\mathcal W}_{C,t}^i \cdot I(a=a_D) W_A^{i} \cdot  { \mathcal W}_{Y,t}^i(a_Y,a_D) \cdot {\mathcal W}_{L_{Y},t}^i(a_Y,a_D) $ & 1 & 1 \\
         & $\overline{R}_{t}^{i} =\overline{\mathcal W}_{C,t}^i \cdot \frac{I(a=a_D)}{\pi_A(A_i)} \cdot  { \mathcal W}_{Y,t}^i(a_Y,a_D) \cdot {\mathcal W}_{L_{Y},t}^i(a_Y,a_D) $ & $-$ & 1
         \\\hline
\end{tabular}
}
    \label{tab:weights}
\end{table}

\subsection{Estimating the weights} \label{subsection: estimating the weights}
Suppose we have a consistent estimator of the propensity score $\pi_A$, which will allow us to estimate the treatment weights in Table \ref{tab:weights}.

The time-varying weights in Table \ref{tab:weights} solve the Doléans-Dade equation
\begin{align}
    W_t^i = 1 + \int_0^t W^i_{s-}(\theta_s^i - 1) d\bar N^i_s + \int_0^t W^i_{s-} Z_s^i( \alpha_s^i - \alpha_s^{*,i})ds ~, \label{eq: stochastic exponential}
\end{align}
where $W^i$ is the weight of interest, $\bar N^i$ is a counting process, $\alpha^i$ and $\alpha^{*,i}$ are hazards, and $\theta^i = \alpha^{*,i}/ \alpha^i$. In Table \ref{tab:weight hazards}, we present $\alpha^i$, $\alpha^{*,i}$, and $\bar N^i$'s corresponding to the different weights in Table \ref{tab:weights}.  
We consider a weight estimator that is defined via plug-in of cumulative hazard estimates,
\begin{align}
    \hat W_t^i = 1 + \int_0^t \hat W^i_{s-} (\hat \theta_{s-}^i - 1) d\bar N_s^i + \int_0^t \hat W_{s-}^i Z^i_s ( d\hat{ A}_s^i - d \hat{  A}^{*,i}_s ) ~, \label{eq: continuous time weight estimator}
\end{align}
where $\hat {A}^i$ and $\hat{ A}^{*,i}$ are cumulative hazard estimates of $ A_t^i = \int_0^t   \alpha_s^i ds$ and $ A_t^{*,i} = \int_0^t   \alpha_s^{*,i} ds$ and $\hat \theta_t^i = \frac{\hat{A}^{*,i}_t - \hat{ A}^{*,i}_{t-b}}{\hat{ A}^i_t - \hat{ A}^i_{t-b}}$, where $b$ is a smoothing parameter used to obtain the hazard ratio $\hat \theta_t^i$. The solution to (\ref{eq: continuous time weight estimator}) is determined by the cumulative hazard estimates and the counting process. Thus, the smoothing parameter $b$ contributes to the estimator only when $\bar N^i$ jumps, which will not happen for the weights $\mathcal W_C^i, \overline{\mathcal{W}}_C^i, \mathcal W_D^i,$ and $\overline{\mathcal{W}}_D^i$ in the examples we consider in Table \ref{tab:weights}. The counting process term in (\ref{eq: continuous time weight estimator}) can therefore be neglected for the upper four weights in Table \ref{tab:weight hazards}. 
For the other time-varying weights, choosing $b$ requires a trade-off between bias and variance, see \citet{ryalen_additive_2019} for a discussion.

To estimate ${ \mathcal{W}}_{L_{D},t}^i$, the weights associated with $L_{D}$, we suppose that there are $m$ marks. We consider the counting processes $\{N_h^i\}_{h=1}^m$ that "counts" the occurrence of each mark of individual $i$, having intensity $Z_t^i \cdot \alpha_{h,t}^idt = E[dN^i_{h,t}| \sigma( L_s^i,N_s^i,N_s^{D,i},C_s^i,A; s < t )]$. Then,
\begin{align*}
    \mathcal W_{L_{D},t}^i = \prod_{h=1}^m \mathcal W_{L_{D},h,t}^i ~, 
\end{align*}
where $\mathcal W_{L_{D},h,t}^i$ solves (\ref{eq: stochastic exponential}) with $\alpha_t^i = \alpha_{h,t}^i|_{A=a_Y}$, $\alpha_t^{*,i} = \alpha_{h,t}^{*,i}|_{A=a_D}$ and $\bar N^i = N_h^i$. We thus obtain an estimator of $\mathcal W_{L_{D},t}^i$ by multiplying the estimators $W_{L_{D},h,t}^i$, each of which solve (\ref{eq: continuous time weight estimator}). A corresponding procedure can be used to estimate the weights $\mathcal{W}_{L_Y,t}^i$. We present the choices of $\alpha^i$ and $\alpha^{*,i}$ for the different weights in Table \ref{tab:weight hazards}. For high-dimensional covariates $L_k$, these weight estimators may give rise to erratic behavior (this is also described for the related mediation weights in~\citet{mittinty_longitudinal_2020}). In future work, one could also consider estimators motivated by the alternative odds representation of the covariate weights, given in Appendix~\ref{sec:proof_identification}, along the lines of~\citet{zheng_longitudinal_2017,stensrud_generalized_2021}, but this is beyond the scope of the current work.

 In summary, we suggest the following strategy for estimating the causal effects of interest:
 \begin{itemize}
     \item Identify the requisite weights from Table \ref{tab:weights} and specify hazard models $\alpha^i$, $\alpha^{*,i}$ from Table \ref{tab:weight hazards}.
     \item Solve (\ref{eq: continuous time weight estimator}) to obtain estimates of the weight processes.
     \item Obtain $\hat R^i, \hat R^{D,i}, \hat{\overline {R}}^i$, and $\hat{\bar{\theta}}^i$ from (\ref{eq: weighted risk set estimator}) or (\ref{eq: weighted rate estimator Hajek and H-T}) by multiplying together the weight estimates of individual $i$ according to Table \ref{tab:weights}.
     \item Solve (\ref{eq: mean frequency estimator}) to obtain $\hat Y_t$ (and/or $\hat D_t$), which estimates the expected number of events under the chosen intervention at $t$.
     \item Repeat the previous steps with a contrasting intervention on treatment to obtain the targeted causal contrast.
     \item Evaluate the uncertainty of the estimators using non-parametric bootstrap.
 \end{itemize}
 We use this estimation method in Sec.\ \ref{sec:examples_revisited}, assuming additive hazard models for the different $\alpha^i$'s and $\alpha^{*,i}$'s. The estimators are implemented in the \texttt{R} packages \texttt{transform.hazards} and \texttt{ahw} (available at \texttt{github.com/palryalen/}). The code is found in the online supplementary material.

\begin{table}[h!]
    \centering
    \caption{Hazards and counting processes that define (\ref{eq: stochastic exponential}) for the different weights.}
    \resizebox{\columnwidth}{!}{
    \begin{tabular}{ |p{ 2.5 cm }|p{ 11.5 cm }|p{ 1.7 cm }|}
         \hline 
         \textbf{Weight}     & \textbf{Hazards in (\ref{eq: stochastic exponential})} & $\bar{N}^i$ \\
         \hline
         $\overline{ \mathcal W }_{C,t}^i$  & $ \alpha_t^idt = P(t \leq C^i < t+dt|t \leq T^{D,i}, t \leq C^i, \overline L^i_{t-}, \overline N^i_{t-}, A^i) $ & $I(C^i\leq \cdot )$  \\ 
         & $\alpha_t^{*,i} dt = 0$ & 
         \\\hline
          $\mathcal {W}_{C,t}^i$ &   $ \alpha_t^idt = P(t \leq C^i < t+dt|t \leq T^{D,i}, t \leq C^i, \overline L^i_{t-}, \overline N^i_{t-}, A^i) $  & $I(C^i\leq \cdot )$  \\  
         & $\alpha_t^{*,i} dt =  P(t \leq C^i < t+dt|t \leq T^{D,i}, t \leq C^i, A^i) $ &
         \\ \hline
         $ \overline{ \mathcal W }_{D,t}^i $   &  $  \alpha_t^i dt = P(t \leq T^{D,i} < t+dt| t \leq T^{D,i}, t \leq C^i, \overline L^i_{t-}, \overline N^i_{t-},A^i)  $     & $N^{D,i}$     \\
          &  $ \alpha_t^{*,i} dt = 0 $  &   \\
         \hline
         $\mathcal{W}_{D,t}^i$  &  $  \alpha_t^i dt = P(t \leq T^{D,i} < t+dt| t \leq T^{D,i}, t \leq C^i, \overline L^i_{t-}, \overline N^i_{t-},A^i)  $     & $N^{D,i}$     \\
          &  $ \alpha_t^{*,i} dt = P(t \leq T^{D,i} < t+dt| t \leq T^{D,i}, t \leq C^i,A^i) $   &  \\
         \hline
         $\mathcal{W}_{D,t}^i (a_Y,a_D)$ &  $  \alpha_t^i dt = P(t \leq T^{D,i} < t+dt| t \leq T^{D,i}, t \leq C^i, \overline L^i_{t-}, \overline N^i_{t-},A=a_Y)  $  & $N^{D,i}$  \\
         &  $\alpha_t^{*,i}dt = P(t \leq T^{D,i} < t+dt| t \leq T^{D,i}, t \leq C^i, \overline L^i_{t-}, \overline N^i_{t-},A=a_D)$ & 
         \\
         \hline
         $\mathcal{W}_{L_{D},h,t}^i (a_Y,a_D)$ &  $  \alpha_{h,t}^i dt = E[dN_{h,t}^i| t \leq T^{D,i}, t \leq C^i, \overline{L}_{t-}^i, \overline{N}^i_{t-}, A=a_Y ]  $ & $N_h^i$   \\
         &  $ \alpha_{h,t}^{*,i} dt = E[dN_{h,t}^i | t \leq T^{D,i}, t \leq C^i, \overline{L}_{t-}^i, \overline{N}^i_{t-}, A=a_D]$ & \\
         \hline
         
         $\mathcal{W}_{Y,t}^i (a_Y,a_D)$ &  $  \alpha_t^i dt = E[dN_t^i| t \leq T^{D,i}, t \leq C^i, \overline L^i_{t-}, \overline N^i_{t-},A=a_D]  $  & $N^{i}$  \\
         &  $\alpha_t^{*,i}dt = E[dN_t^i| t \leq T^{D,i}, t \leq C^i, \overline L^i_{t-}, \overline N^i_{t-},A=a_Y]$ & 
         \\
         \hline
         $\mathcal{W}_{L_{Y},h,t}^i (a_Y,a_D)$ &  $  \alpha_{h,t}^i dt = E[dN_{h,t}^i| t \leq T^{D,i}, t \leq C^i, \overline{L}_{t-}^i, \overline{N}^i_{t-}, A=a_D ]  $ & $N_h^i$   \\
         &  $ \alpha_{h,t}^{*,i} dt = E[dN_{h,t}^i | t \leq T^{D,i}, t \leq C^i, \overline{L}_{t-}^i, \overline{N}^i_{t-}, A=a_Y]$ & \\
         \hline
\end{tabular}
}
    \label{tab:weight hazards}
\end{table}

\subsection{Estimators under assumptions on $L_k$}\label{sec:simplified_estimators_partition}

There exist two important settings where we do not need to model the densities of the covariate process $L_t$. Firstly, if the dismissible component conditions are satisfied with $L_{Y,k} \equiv L_k$ and $L_{D,k}=\emptyset$, then $\mathcal{W}_{L_{D},t}=1$ (a further elaboration on this point is found in Appendix~\ref{sec:proof_identification}). The assumption that $L_{D,k}=\emptyset$ implies that $A_D$ partial isolation holds (see Appendix~\ref{sec:app_isolation} for a definition of $A_D$ partial isolation and Lemma 6 of~\citet{stensrud_generalized_2021} for a proof of this result). This is unlikely to hold in the trial considered in Sec.~\ref{sec:examples_revisited}, because we expect that the $A_D$ component of treatment can cause acute kidney injury by lowering systemic blood pressure, i.e. through the pathway $A_D\rightarrow L_j \rightarrow Y_{k>j}$, which is not intersected by any $D_{i\leq k}$. 

Secondly, if the dismissible component conditions are satisfied under $L_{D,k} \equiv L_k$ and $L_{Y,k}=\emptyset$, we have that $\mathcal{W}_{L_{Y},t}=1$. In the trial in Sec.~\ref{sec:examples_revisited}, this assumption implies that the component that binds to receptors in the kidneys ($A_Y$) has no effect on blood pressure outside of its possible effect on the risk of acute kidney injury (see Lemma 5 of~\citet{stensrud_generalized_2021}). This is plausible and serves as a sanity check of the dismissible component conditions in this example. However, in practice we have intermittent blood pressure measurements $L_k$, and therefore the dismissible component conditions under $L_{D,k} \equiv L_k$ and $L_{Y,k}=\emptyset$ hold at best approximately.

Even in these two simplified settings, the natural direct effect is not identified as the measured covariates $L_{Y,k}$ (or $L_{D,k}$) act as a recanting witness (see Sec. \ref{subsection: separable effects identification}).

\section{Example: blood pressure treatment and acute kidney injury}\label{sec:examples_revisited}
In Sec.~\ref{sec:sep_effects_no_censoring} we described a hypothetical modified version of antihypertensive therapy that preserves the effect of existing treatments on systemic blood pressure but does not lead to dilation of efferent arterioles in the kidneys, thereby potentially avoiding a detrimental side effect of treatment which can give risk to acute kidney injury. In this section, we apply the estimators proposed in the Sec.~\ref{sec:estimation_maintext} to compute the effect of such a modified blood pressure therapy on the recurrence of acute kidney injury, as well as the total and controlled direct effect, using data from the Systolic Blood Pressure Intervention Trial \citep{sprint2015randomized}. The illustrative example considered in this section builds on  ~\citet{stensrud_generalized_2021}, but now considers the case where acute kidney injury ($Y_k$) is a recurrent outcome as opposed to the (first) incident event. Another, simulated example from a hypothetical trial on treatment discontinuation is given in Appendix~\ref{sec:sim_examples_revisited} along with R code in the Supplementary Material.

In the SPRINT trial, individuals were randomized to standard ($A=0$) or intensive ($A=1$) blood pressure (BP) lowering therapy.  We consider the effect of intensive versus standard treatment on the recurrence of acute kidney injury by time $t$ during the first 1000 days of follow-up.

We have restricted our analysis to subjects aged over 75 years of age. Furthermore, we have only considered individuals with complete baseline covariates. This led to 1312 individuals under standard treatment and 1311 individuals under intensive treatment. By the end of 1000 days, a total of 73 deaths were recorded in the standard treatment group, versus 52 in the intensive treatment group. In total, 668 individuals were lost to follow-up before day 1000. The frequencies of recorded AKI events by treatment group are given in Table \ref{tab:frequency}.
\begin{table}[htbp]
  \centering
  \caption{Frequency table for recorded AKI events by treatment group}
    \begin{tabular}{|l|r|r|r|r|}\hline
    
          & 0     & 1     & 2     & 3 \\\hline
    
    $A=0$   & 1273  & 36    & 1     & 2 \\\hline
    
    $A=1$   & 1253  & 52    & 4     & 2 \\\hline
    \end{tabular}%
  \label{tab:frequency}%
\end{table}%

We estimated total, controlled direct and separable effect using (\ref{eq: weighted rate estimator Hajek and H-T}) with additive regression models for the hazards, specified in Table \ref{tab:assumed models}. 
\begin{table}[h!]
    \centering
    \caption{The left column includes the weights, the middle column includes the hazards that define (\ref{eq: stochastic exponential}), and the right column includes the parametric hazard models that were used in the data analysis.  
    }
    \begin{tabular}{ |p{ 2 cm }|p{ 1.3 cm }|p{ 8 cm }|}
         \hline 
          \textbf{Weight}     &  \textbf{Hazards} &  \textbf{Hazard models fitted}\\
         \hline
            \multirow{2}{0.2cm}{$ \mathcal W_{C,t}^i$}
           & $ \alpha_t^idt$ &  $\beta_t^0 + \beta_t^A A + \beta_t^{L_0} L_0 + \beta_t^{L} L_t + \beta_t^{Y} Y_{t-} $    \\ \cline{2-3}
            & $\alpha_t^{*,i} dt$ &   $\beta_t^0 + \beta_t^A A$   \\
           \hline
            \multirow{2}{0.2cm}{$ \mathcal W_{D,t}^i$}
           & $ \alpha_t^idt$ &  $\beta_t^0 + \beta_t^A A + \beta_t^{L_0} L_0 + \beta_t^{L} L_t + \beta_t^{Y} Y_{t-}$   \\ \cline{2-3}
            & $\alpha_t^{*,i} dt$ &   $\beta_t^0 + \beta_t^A A$   \\
            \cline{1-3}
            \multirow{2}{0.2cm}{$ \mathcal W_{Y,t}^i(a_Y,a_D)$}
           & $ \alpha_t^idt$ &  $\beta_t^0 + \beta_t^A I(A=a_Y) + \beta^{L_0} L_0 + \beta_t^{L} L_t + \beta_t^{Y} Y_{t-}$    \\ \cline{2-3}
            & $\alpha_t^{*,i} dt$ &   $\beta_t^0 + \beta_t^A I(A=a_D) + \beta^{L_0} L_0 + \beta_t^{L} L_t + \beta_t^{Y} Y_{t-}$  \\
         \hline
\end{tabular}
    
    \label{tab:assumed models}
\end{table}
Following \citet{stensrud_generalized_2021}, we included the following baseline covariates ($L_0$): smoking status, history of clinical or subclinical cardiovascular disease, clinical or subclinical chronic kidney disease, statin use and sex. Additionally, we adjusted for the most recent measurements of mean arterial pressure ($L_{k}$). 
We truncated the stabilized weights $\mathcal{W}_{C,t},\mathcal{W}_{D,t}$ and $\mathcal{W}_{Y,t}$ outside of the interval [0.2-5]. A smoothing parameter of $b=250$ was used (analyses with parameters $b\in\{100,200,500\}$ gave similar results).

The analysis relies on the identification assumptions in Sec. \ref{sec:identification}.  The assumptions of exchangebility of baseline treatment, positivity and consistency hold by design because blood pressure treatment is assigned by randomization in a controlled experiment. We have further assumed that the measured covariates $(L_0,L_k)$ are sufficient for identification. In particular, exchangeability \eqref{eq:exchangeability_direct_ii} for the controlled direct effect and the dismissible component conditions \eqref{eq:dismissible_component_i}-\eqref{eq:dismissible_component_ii} would be violated if there are unmeasured common causes of death and recurrent AKI, such as $U_{DY}$ in Figures~\ref{fig:direct_effect} and \ref{fig:separable_effect}, that are not captured by $(L_0,L_k)$, or causal paths such as $A_Y \rightarrow M_{A_Y} \rightarrow D_{k+1}^{\overline{c}=0}$ or $A_D\rightarrow M_{A_D}\rightarrow Y_{k+1}^{\overline{c}=0}$ in Fig. \ref{fig:separable_effect} that are not intersected by $L_k$.

The resulting estimates are shown in Fig. \ref{fig:plots_motivating_example}. At 1000 days, we found a total effect of $0.017[-0.001,0.037]$ and a controlled direct effect of $0.017[-0.003,0.037]$ (95\% confidence intervals, obtained using 500 non-parametric bootstrap samples, are reported in square brackets). Thus, the total effect for individuals over 75 years old is (borderline) consistent with an increased occurrence of acute kidney injury under intensive treatment, as reported by \citet{sprint2015randomized} for the full trial population.
However, there were more deaths in the standard group compared to the intensive group (see Fig. 9a in \citet{stensrud_generalized_2021}). Thus, it is not clear whether the increased occurrence of acute kidney injury in the intensive group is due to a protective effect on survival or a direct effect on the recurrent outcome. 

To quantify the mechanism by which treatment leads to increased risk of acute kidney injury, we studied separable effects. The direct separable effect evaluated at $a_D=1$ is equal to $0.011[-0.005,0.034]$ at 1000 days, which is consistent with no reduction in the recurrence of acute kidney injury by eliminating the $A_Y$ component of treatment. This finding is also consistent with \citet{stensrud_generalized_2021}, who only studied the incidence of the first kidney failure event. To conclude, the analysis of separable effects does not provide evidence in favor of a reduction in the expected number of acute kidney injury episodes in a modified blood pressure treatment that does not dilate efferent glomerular arterioles. If a non-null effect had been found, this would strengthen the hypothesis that we could change the number of acute kidney injury occurrences by intervening on the treatment component that dilates efferent arterioles, and thereby make it more attractive to test such a hypothesis in a future randomized trial if such a treatment is developed.

\begin{figure}
    \centering
    \includegraphics[width=0.7\linewidth]{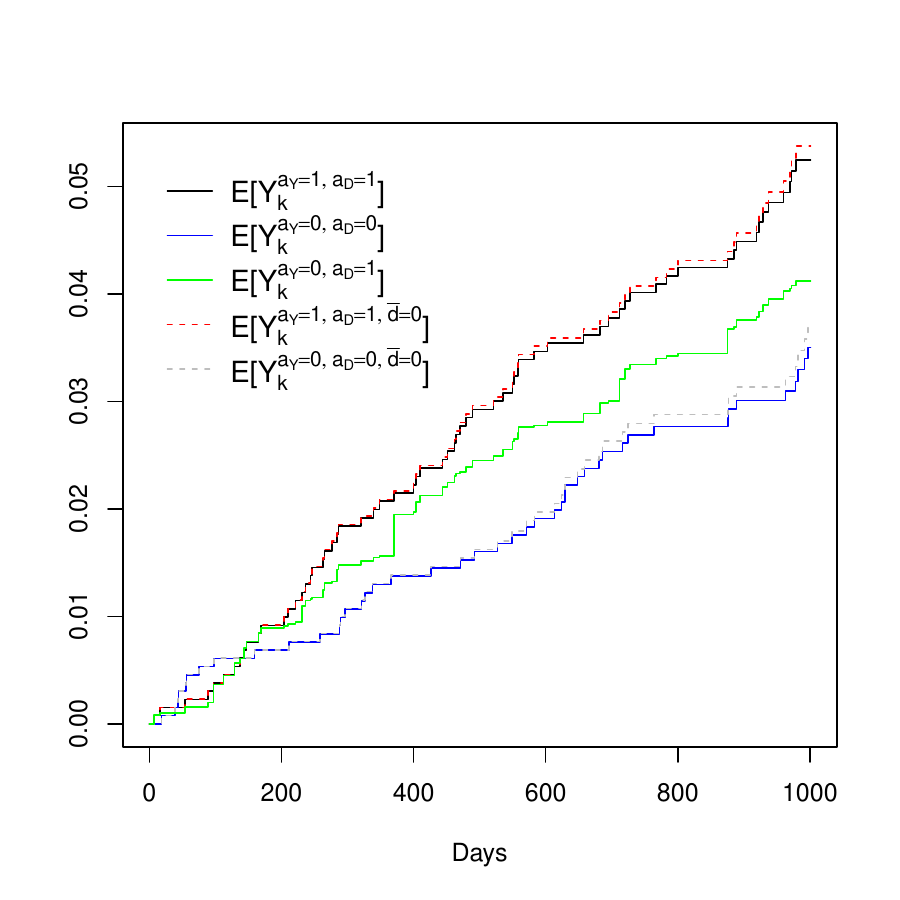}
    \caption{Shows estimates of the total effect, controlled direct effect and separable effect on recurrence of acute kidney injury up to 1000 days.  The superscript $\overline{c}=0$, denoting an intervention to prevent loss to follow-up, has been suppressed in plot legends to reduce clutter.}
    \label{fig:plots_motivating_example}
\end{figure}

\section{Discussion \label{sec:discussion}}
We have used a formal causal framework to define estimands for recurrent outcomes that differ in the way they treat competing events. The controlled direct effect is a contrast of counterfactual outcomes which implies that competing events are considered to be a form of censoring. The total effect captures all causal pathways between treatment and the recurrent event, and the separable effect quantifies contrasts in expected outcomes under independent prescription of treatment components.

Further, we have given formal conditions for identifying these effects, and demonstrated how to evaluate the identification conditions in causal graphs. This allowed us to formally describe how the causal estimands map to classical statistical estimands for recurrent events based on counting processes in the limit of fine discretizations of time. 

In settings with competing events, it is often of interest to disentangle the effect on the recurrent event from the effect on the competing event. The controlled direct effect often fails to do so in a scientifically insightful way, because it is not clear which intervention, if any, eliminates the occurrence of the competing event. The interpretation of the direct effect is therefore unclear. The separable effect corresponds (by design) to interventions on components of the original treatment, which are assigned independently of each other. The practical relevance of the estimand relies on the plausibility of modified treatments. The process of conceptualizing modified treatments can motivate  future treatment development and sharpen research questions about mechanisms \citep{robins_alternative_2011,robins2020interventionist,stensrud_separable_2020-1}.

Stronger assumptions are needed to identify the (controlled) direct effect and separable effects compared to the total effect. For example, these estimands require the investigator to measure common causes of the recurrent event and failure time, even in an ideal randomized trial such as in Sec.~\ref{sec:observed_data}. The need for stronger assumptions is far from unique to our setting, and it is analogous to the task of identifying per-protocol effects in settings with imperfect adherence and mediation effects.

The use of a formal (counterfactual) framework to define causal effects elucidates analytic choices regarding treatment recommendations. The formal causal framework makes it possible to define effects with respect to explicit interventions, and to explicitly state the conditions under which such effects can be identified from observed data. This also makes it possible to transparently assess the strength and validity of the identifying assumptions in practice.

\section*{Acknowledgements}
This manuscript was prepared using SPRINT Research Materials obtained from the NHLBI Biologic Specimen and Data Repository Information Coordinating Center and does not necessarily reflect the opinions or views of the SPRINT  or the NHLBI.

\bibliography{references}

\begin{thebibliography}{60}
\providecommand{\natexlab}[1]{#1}
\providecommand{\url}[1]{\texttt{#1}}
\expandafter\ifx\csname urlstyle\endcsname\relax
  \providecommand{\doi}[1]{doi: #1}\else
  \providecommand{\doi}{doi: \begingroup \urlstyle{rm}\Url}\fi

\bibitem[Aalen et~al.(2008)Aalen, Borgan, and Gjessing]{aalen_survival_2008}
Odd~O. Aalen, Ørnulf Borgan, and Håkon~K. Gjessing.
\newblock \emph{Survival and {Event} {History} {Analysis}}.
\newblock Statistics for {Biology} and {Health}. Springer New York, New York, NY, 2008.

\bibitem[Aalen et~al.(2015)Aalen, Cook, and R{\o}ysland]{aalen2015does}
Odd~O Aalen, Richard~J Cook, and Kjetil R{\o}ysland.
\newblock Does cox analysis of a randomized survival study yield a causal treatment effect?
\newblock \emph{Lifetime data analysis}, 21\penalty0 (4):\penalty0 579--593, 2015.

\bibitem[Andersen et~al.(2019)Andersen, Angst, and Ravn]{andersen_modeling_2019}
Per~Kragh Andersen, Jules Angst, and Henrik Ravn.
\newblock Modeling marginal features in studies of recurrent events in the presence of a terminal event.
\newblock \emph{Lifetime Data Analysis}, 25\penalty0 (4):\penalty0 681--695, 2019.

\bibitem[Anker and McMurray(2012)]{anker_time_2012}
Stefan~D. Anker and John~J.V. McMurray.
\newblock Time to move on from ‘time-to-first’: should all events be included in the analysis of clinical trials?
\newblock \emph{European Heart Journal}, 33\penalty0 (22):\penalty0 2764--2765, 2012.

\bibitem[Brunton et~al.(2018)Brunton, Knollmann, and Hilal-Dandan]{brunton_goodman_2018}
Laurence~L. Brunton, Björn~C. Knollmann, and Randa Hilal-Dandan.
\newblock \emph{Goodman \& {Gilman}'s: {The} {Pharmacological} {Basis} of {Therapeutics}}.
\newblock {McGraw}-{Hill}'s {AccessMedicine}. McGraw-Hill Education LLC, New York, N.Y, 13th ed. edition, 2018.

\bibitem[Chen and Cook(2004)]{chen_tests_2004}
Bingshu~Eric Chen and Richard~J. Cook.
\newblock Tests for multivariate recurrent events in the presence of a terminal event.
\newblock \emph{Biostatistics}, 5\penalty0 (1):\penalty0 129--143, 2004.

\bibitem[Claggett et~al.(2018)Claggett, Tian, Fu, Solomon, and Wei]{claggett_quantifying_2018}
Brian Claggett, Lu~Tian, Haoda Fu, Scott~D. Solomon, and Lee-Jen Wei.
\newblock Quantifying the totality of treatment effect with multiple event-time observations in the presence of a terminal event from a comparative clinical study.
\newblock \emph{Statistics in Medicine}, 37\penalty0 (25):\penalty0 3589--3598, 2018.

\bibitem[Cook and Lawless(1997)]{cook_marginal_1997}
Richard~J. Cook and Jerald~F. Lawless.
\newblock Marginal {Analysis} of {Recurrent} {Events} and a {Terminating} {Event}.
\newblock \emph{Statistics in Medicine}, 16\penalty0 (8):\penalty0 911--924, 1997.

\bibitem[Cook and Lawless(2007)]{cook_statistical_2007}
Richard~J. Cook and Jerald~F. Lawless.
\newblock \emph{The statistical analysis of recurrent events}.
\newblock Statistics for biology and health. Springer, New York, 2007.

\bibitem[Dawid and Didelez(2012)]{dawid2012imagine}
Philip Dawid and Vanessa Didelez.
\newblock "{I}magine a can opener"--the magic of principal stratum analysis.
\newblock \emph{The international journal of biostatistics}, 8\penalty0 (1), 2012.

\bibitem[Didelez(2019)]{didelez2019defining}
Vanessa Didelez.
\newblock Defining causal mediation with a longitudinal mediator and a survival outcome.
\newblock \emph{Lifetime data analysis}, 25\penalty0 (4):\penalty0 593--610, 2019.

\bibitem[{European Medicines Agency}(2020)]{european_medicines_agency_qualification_nodate}
{European Medicines Agency}.
\newblock Qualification opinion of clinically interpretable treatment effect measures based on recurrent event endpoints that allow for efficient statistical analysis.
\newblock 2020.

\bibitem[Frangakis and Rubin(2002)]{frangakis_principal_2002}
Constantine~E. Frangakis and Donald~B. Rubin.
\newblock Principal stratification in causal inference.
\newblock \emph{Biometrics}, 58\penalty0 (1):\penalty0 21--29, 2002.

\bibitem[Fritsch et~al.(2021)Fritsch, Schlömer, Mendolia, Mütze, and Jahn-Eimermacher]{fritsch_efficiency_2021}
Arno Fritsch, Patrick Schlömer, Franco Mendolia, Tobias Mütze, and Antje Jahn-Eimermacher.
\newblock Efficiency {Comparison} of {Analysis} {Methods} for {Recurrent} {Event} and {Time}-to-{First} {Event} {Endpoints} in the {Presence} of {Terminal} {Events}—{Application} to {Clinical} {Trials} in {Chronic} {Heart} {Failure}.
\newblock \emph{Statistics in Biopharmaceutical Research}, 0\penalty0 (0):\penalty0 1--12, 2021.

\bibitem[Gail(1975)]{gail_review_1975}
Mitchell Gail.
\newblock A {Review} and {Critique} of {Some} {Models} {Used} in {Competing} {Risk} {Analysis}.
\newblock \emph{Biometrics}, 31\penalty0 (1):\penalty0 209, March 1975.

\bibitem[Ghosh and Lin(2000)]{ghosh_nonparametric_2000}
Debashis Ghosh and D.~Y. Lin.
\newblock Nonparametric {Analysis} of {Recurrent} {Events} and {Death}.
\newblock \emph{Biometrics}, 56\penalty0 (2):\penalty0 554--562, 2000.

\bibitem[Hajek(1971)]{godambe_comment_1971}
J.~Hajek.
\newblock Comment on "{An} essay on the logical foundations of survey sampling by {D}. {Basu}".
\newblock In V.~P. Godambe and D.~A. Sprott, editors, \emph{Foundations of statistical inference}. Holt, Rinehart and Winston of Canada, Toronto, 1971.

\bibitem[Hernán(2010)]{hernan_hazards_2010}
Miguel~A. Hernán.
\newblock The {Hazards} of {Hazard} {Ratios}.
\newblock \emph{Epidemiology (Cambridge, Mass.)}, 21\penalty0 (1):\penalty0 13--15, 2010.

\bibitem[Hernán et~al.(2000)Hernán, Brumback, and Robins]{hernan_marginal_2000}
Miguel~Ángel Hernán, Babette Brumback, and James~M. Robins.
\newblock Marginal {Structural} {Models} to {Estimate} the {Causal} {Effect} of {Zidovudine} on the {Survival} of {HIV}-{Positive} {Men}.
\newblock \emph{Epidemiology}, 11\penalty0 (5):\penalty0 561--570, September 2000.

\bibitem[Horvitz and Thompson(1952)]{horvitz_generalization_1952}
D.~G. Horvitz and D.~J. Thompson.
\newblock A {Generalization} of {Sampling} {Without} {Replacement} from a {Finite} {Universe}.
\newblock \emph{Journal of the American Statistical Association}, 47\penalty0 (260):\penalty0 663--685, December 1952.

\bibitem[Jacod and Shiryaev(2003)]{JacodShiryaev}
Jean Jacod and Albert~N. Shiryaev.
\newblock \emph{Limit theorems for stochastic processes}, volume 288 of \emph{Grundlehren der Mathematischen Wissenschaften [Fundamental Principles of Mathematical Sciences]}.
\newblock Springer-Verlag, Berlin, second edition, 2003.

\bibitem[Joffe(2011)]{joffe2011principal}
Marshall Joffe.
\newblock Principal stratification and attribution prohibition: good ideas taken too far.
\newblock \emph{The International Journal of Biostatistics}, 7\penalty0 (1):\penalty0 1--22, 2011.

\bibitem[Martinussen and Stensrud(2021)]{martinussen2021estimation}
Torben Martinussen and Mats~Julius Stensrud.
\newblock Estimation of separable direct and indirect effects in continuous time.
\newblock \emph{Biometrics}, 2021.

\bibitem[Martinussen et~al.(2020)Martinussen, Vansteelandt, and Andersen]{martinussen_subtleties_2020}
Torben Martinussen, Stijn Vansteelandt, and Per~Kragh Andersen.
\newblock Subtleties in the interpretation of hazard contrasts.
\newblock \emph{Lifetime Data Analysis}, 26\penalty0 (4):\penalty0 833--855, 2020.

\bibitem[Mittinty and Vansteelandt(2020)]{mittinty_longitudinal_2020}
Murthy~N Mittinty and Stijn Vansteelandt.
\newblock Longitudinal {Mediation} {Analysis} {Using} {Natural} {Effect} {Models}.
\newblock \emph{American Journal of Epidemiology}, 189\penalty0 (11):\penalty0 1427--1435, 2020.

\bibitem[Pearl(2001)]{pearl_direct_nodate}
Judea Pearl.
\newblock Direct and {Indirect} {Effects}.
\newblock \emph{Proceedings of the Seventeenth Conference on Uncertainty in Artificial Intelligence}, pages 411--20, 2001.

\bibitem[Pearl(2009)]{pearl_causality:_2009}
Judea Pearl.
\newblock \emph{Causality: models, reasoning, and inference}.
\newblock Cambridge University Press, Cambridge, 2nd ed. edition, 2009.

\bibitem[Prentice et~al.(1978)Prentice, Kalbfleisch, Peterson~Jr, Flournoy, Farewell, and Breslow]{prentice1978analysis}
Ross~L Prentice, John~D Kalbfleisch, Arthur~V Peterson~Jr, Nancy Flournoy, Vern~T Farewell, and Norman~E Breslow.
\newblock The analysis of failure times in the presence of competing risks.
\newblock \emph{Biometrics}, pages 541--554, 1978.

\bibitem[Putter et~al.(2007)Putter, Fiocco, and Geskus]{putter_tutorial_2007}
H.~Putter, M.~Fiocco, and R.~B. Geskus.
\newblock Tutorial in biostatistics: competing risks and multi-state models.
\newblock \emph{Statistics in Medicine}, 26\penalty0 (11):\penalty0 2389--2430, 2007.

\bibitem[Reeve et~al.(2020)Reeve, Jordan, Thompson, Sawan, Todd, Gammie, Hopper, Hilmer, and Gnjidic]{reeve_withdrawal_2020}
Emily Reeve, Vanessa Jordan, Wade Thompson, Mouna Sawan, Adam Todd, Todd~M. Gammie, Ingrid Hopper, Sarah~N. Hilmer, and Danijela Gnjidic.
\newblock Withdrawal of antihypertensive drugs in older people.
\newblock \emph{Cochrane Database of Systematic Reviews}, \penalty0 (6), 2020.

\bibitem[Richardson and Robins(2013{\natexlab{a}})]{richardson_primer_2013}
Thomas~S Richardson and James~M Robins.
\newblock Single {World} {Intervention} {Graphs}: {A} {Primer}.
\newblock 2013{\natexlab{a}}.

\bibitem[Richardson and Robins(2013{\natexlab{b}})]{richardson_single_2013}
Thomas~S. Richardson and James~M. Robins.
\newblock Single {World} {Intervention} {Graphs} {(SWIGs}): {A} {Unification} of the {Counterfactual} and {Graphical} {Approaches} to {Causality}.
\newblock 2013{\natexlab{b}}.

\bibitem[Robins(1986)]{robins_new_1986}
James Robins.
\newblock A new approach to causal inference in mortality studies with a sustained exposure period—application to control of the healthy worker survivor effect.
\newblock \emph{Mathematical Modelling}, 7\penalty0 (9):\penalty0 1393--1512, 1986.

\bibitem[Robins et~al.(2007)Robins, Rotnitzky, Vansteelandt, Have, Xie, and Murphy]{robins2007discussions}
James Robins, Andrea Rotnitzky, Stijn Vansteelandt, Tom~Ten Have, Yu~Xie, and Susan Murphy.
\newblock Discussions on "{Principal} stratification designs to estimate input data missing due to death".
\newblock \emph{Biometrics}, 63\penalty0 (3):\penalty0 650--658, 2007.

\bibitem[Robins and Finkelstein(2000)]{robins_correcting_2000}
James~M. Robins and Dianne~M. Finkelstein.
\newblock Correcting for {Noncompliance} and {Dependent} {Censoring} in an {AIDS} {Clinical} {Trial} with {Inverse} {Probability} of {Censoring} {Weighted} ({IPCW}) {Log}-{Rank} {Tests}.
\newblock \emph{Biometrics}, 56\penalty0 (3):\penalty0 779--788, 2000.

\bibitem[Robins and Greenland(1992)]{robins1992identifiability}
James~M Robins and Sander Greenland.
\newblock Identifiability and exchangeability for direct and indirect effects.
\newblock \emph{Epidemiology}, pages 143--155, 1992.

\bibitem[Robins and Richardson(2011)]{robins_alternative_2011}
James~M. Robins and Thomas~S. Richardson.
\newblock Alternative {Graphical} {Causal} {Models} and the {Identification} of {Direct} {Effects}.
\newblock In \emph{Causality and {Psychopathology}}. Oxford University Press, 2011.

\bibitem[Robins and Rotnitzky(1992)]{robins_recovery_1992}
James~M. Robins and Andrea Rotnitzky.
\newblock Recovery of {Information} and {Adjustment} for {Dependent} {Censoring} {Using} {Surrogate} {Markers}.
\newblock In Nicholas~P. Jewell, Klaus Dietz, and Vernon~T. Farewell, editors, \emph{{AIDS} {Epidemiology}: {Methodological} {Issues}}, pages 297--331. Birkhäuser, Boston, MA, 1992.

\bibitem[Robins et~al.(2020)Robins, Richardson, and Shpitser]{robins2020interventionist}
James~M Robins, Thomas~S Richardson, and Ilya Shpitser.
\newblock An interventionist approach to mediation analysis.
\newblock \emph{arXiv preprint arXiv:2008.06019}, 2020.

\bibitem[Rotnitzky and Robins(1995)]{rotnitzky_semiparametric_1995}
Andrea Rotnitzky and James~M. Robins.
\newblock Semiparametric regression estimation in the presence of dependent censoring.
\newblock \emph{Biometrika}, 82\penalty0 (4):\penalty0 805--820, December 1995.

\bibitem[Ryalen et~al.(2018)Ryalen, Stensrud, and Røysland]{ryalen2018transforming}
Pål~C. Ryalen, Mats~J. Stensrud, and Kjetil Røysland.
\newblock {Transforming cumulative hazard estimates}.
\newblock \emph{Biometrika}, 105:\penalty0 905--916, 2018.

\bibitem[Ryalen et~al.(2019)Ryalen, Stensrud, and Røysland]{ryalen_additive_2019}
Pål~C. Ryalen, Mats~J. Stensrud, and Kjetil Røysland.
\newblock The additive hazard estimator is consistent for continuous-time marginal structural models.
\newblock \emph{Lifetime Data Analysis}, 25\penalty0 (4):\penalty0 611--638, 2019.

\bibitem[Sarvet et~al.(2020)Sarvet, Wanis, Stensrud, and Hernán]{sarvet_graphical_2020}
Aaron~L. Sarvet, Kerollos~Nashat Wanis, Mats~J. Stensrud, and Miguel~A. Hernán.
\newblock A {Graphical} {Description} of {Partial} {Exchangeability}.
\newblock \emph{Epidemiology}, 31\penalty0 (3):\penalty0 365--368, 2020.

\bibitem[Schmidli et~al.(2021)Schmidli, Roger, Akacha, and {on behalf of the Recurrent Event Qualification Opinion Consortium}]{schmidli_estimands_2021}
Heinz Schmidli, James~H. Roger, Mouna Akacha, and {on behalf of the Recurrent Event Qualification Opinion Consortium}.
\newblock Estimands for recurrent event endpoints in the presence of a terminal event.
\newblock \emph{Statistics in Biopharmaceutical Research}, pages 1--29, 2021.

\bibitem[Spirtes et~al.(2000)Spirtes, Glymour, and Scheines]{spirtes_causation_2000}
Peter Spirtes, Clark~N. Glymour, and Richard Scheines.
\newblock \emph{Causation, prediction, and search}.
\newblock Adaptive computation and machine learning. MIT Press, Cambridge, Mass, 2nd ed edition, 2000.

\bibitem[{SPRINT Research Group}(2015)]{sprint2015randomized}
{SPRINT Research Group}.
\newblock A randomized trial of intensive versus standard blood-pressure control.
\newblock \emph{New England Journal of Medicine}, 373\penalty0 (22):\penalty0 2103--2116, 2015.

\bibitem[Stensrud and Dukes(2022)]{stensrud_translating_2022}
Mats~J. Stensrud and Oliver Dukes.
\newblock Translating questions to estimands in randomized clinical trials with intercurrent events.
\newblock \emph{Statistics in Medicine}, 41\penalty0 (16):\penalty0 3211--3228, 2022.

\bibitem[Stensrud and Hern{\'a}n(2020)]{stensrud2020test}
Mats~J Stensrud and Miguel~A Hern{\'a}n.
\newblock Why test for proportional hazards?
\newblock \emph{JAMA}, 323\penalty0 (14):\penalty0 1401--1402, 2020.

\bibitem[Stensrud et~al.(2020)Stensrud, Young, Didelez, Robins, and Hernán]{stensrud_separable_2020-1}
Mats~J. Stensrud, Jessica~G. Young, Vanessa Didelez, James~M. Robins, and Miguel~A. Hernán.
\newblock Separable {Effects} for {Causal} {Inference} in the {Presence} of {Competing} {Events}.
\newblock \emph{Journal of the American Statistical Association}, pages 1--9, 2020.

\bibitem[Stensrud et~al.(2021{\natexlab{a}})Stensrud, Hernán, {Tchetgen Tchetgen}, Robins, Didelez, and Young]{stensrud_generalized_2021}
Mats~J. Stensrud, Miguel~A. Hernán, Eric~J {Tchetgen Tchetgen}, James~M. Robins, Vanessa Didelez, and Jessica~G. Young.
\newblock A generalized theory of separable effects in competing event settings.
\newblock \emph{Lifetime Data Analysis}, 2021{\natexlab{a}}.

\bibitem[Stensrud et~al.(2021{\natexlab{b}})Stensrud, Young, and Martinussen]{stensrud_discussion_2021}
Mats~J. Stensrud, Jessica~G. Young, and Torben Martinussen.
\newblock Discussion on “{Causal} mediation of semicompeting risks” by {Yen}-{Tsung} {Huang}.
\newblock \emph{Biometrics}, 77\penalty0 (4):\penalty0 1160--1164, 2021{\natexlab{b}}.

\bibitem[Stensrud et~al.(2022)Stensrud, Robins, Sarvet, {Tchetgen Tchetgen}, and Young]{stensrud2020conditional}
Mats~J. Stensrud, James~M. Robins, Aaron Sarvet, Eric~J. {Tchetgen Tchetgen}, and Jessica~G. Young.
\newblock Conditional separable effects.
\newblock \emph{Journal of the American Statistical Association}, 0\penalty0 (0):\penalty0 1--13, 2022.

\bibitem[{Tchetgen Tchetgen}(2013)]{tchetgen_tchetgen_inverse_2013}
Eric~J. {Tchetgen Tchetgen}.
\newblock Inverse odds ratio-weighted estimation for causal mediation analysis.
\newblock \emph{Statistics in Medicine}, 32\penalty0 (26):\penalty0 4567--4580, 2013.

\bibitem[Tsiatis(1975)]{tsiatis_nonidentifiability_1975}
A.~Tsiatis.
\newblock A nonidentifiability aspect of the problem of competing risks.
\newblock \emph{Proceedings of the National Academy of Sciences}, 72\penalty0 (1):\penalty0 20--22, January 1975.

\bibitem[Vansteelandt et~al.(2019)Vansteelandt, Linder, Vandenberghe, Steen, and Madsen]{vansteelandt_mediation_2019}
Stijn Vansteelandt, Martin Linder, Sjouke Vandenberghe, Johan Steen, and Jesper Madsen.
\newblock Mediation analysis of time‐to‐event endpoints accounting for repeatedly measured mediators subject to time‐varying confounding.
\newblock \emph{Statistics in Medicine}, 38\penalty0 (24):\penalty0 4828--4840, October 2019.

\bibitem[Wei et~al.(2021)Wei, Mütze, Jahn-Eimermacher, and Roger]{wei_properties_2021}
Jiawei Wei, Tobias Mütze, Antje Jahn-Eimermacher, and James Roger.
\newblock Properties of {Two} {While}-{Alive} {Estimands} for {Recurrent} {Events} and {Their} {Potential} {Estimators}.
\newblock \emph{Statistics in Biopharmaceutical Research}, 0\penalty0 (0):\penalty0 1--11, October 2021.

\bibitem[Xu et~al.(2022)Xu, Scharfstein, Müller, and Daniels]{xu_bayesian_2022}
Yanxun Xu, Daniel Scharfstein, Peter Müller, and Michael Daniels.
\newblock A {Bayesian} nonparametric approach for evaluating the causal effect of treatment in randomized trials with semi-competing risks.
\newblock \emph{Biostatistics}, 23\penalty0 (1):\penalty0 34--49, 2022.

\bibitem[Young and Stensrud(2021)]{young2021identified}
Jessica~G Young and Mats~J Stensrud.
\newblock Identified versus interesting causal effects in fertility trials and other settings with competing or truncation events.
\newblock \emph{Epidemiology}, 32\penalty0 (4):\penalty0 569--572, 2021.

\bibitem[Young et~al.(2020)Young, Stensrud, {Tchetgen Tchetgen}, and Hernán]{young_causal_2020}
Jessica~G. Young, Mats~J. Stensrud, Eric~J. {Tchetgen Tchetgen}, and Miguel~A. Hernán.
\newblock A causal framework for classical statistical estimands in failure-time settings with competing events.
\newblock \emph{Statistics in Medicine}, 39\penalty0 (8):\penalty0 1199--1236, 2020.

\bibitem[Zheng and {van der Laan}(2017)]{zheng_longitudinal_2017}
Wenjing Zheng and Mark {van der Laan}.
\newblock Longitudinal {Mediation} {Analysis} with {Time}-varying {Mediators} and {Exposures}, with {Application} to {Survival} {Outcomes}.
\newblock \emph{Journal of Causal Inference}, 5\penalty0 (2), June 2017.

\end{thebibliography}
\bibliographystyle{plainnat}

\appendix

\section{Illustrative example: a simulated trial on treatment discontinuation \label{sec:sim_examples_revisited}}
In this section, we illustrate an application of the concepts and estimators outlined in Secs.~\ref{sec:identification}-\ref{sec:estimation_maintext} for the total effect, controlled direct effect and separable effect, using a simulated data example. Consider investigators concerned with the effects of over-treatment of older adults with antihypertensive agents ($A_Y$) and aspirin ($A_D$). Over-treatment might lead to episodes ($Y_k$) of syncope (dizziness) caused by blood pressure becoming too low (in turn, possibly leading to injurious falls), with all-cause mortality ($D_k$) as a competing risk. Suppose these investigators conduct a randomized controlled trial in a sample of patients admitted to nursing homes with a history of
cardiovascular disease and currently taking antihypertensives and aspirin. Patients are then randomly assigned to either discontinue or to continue  \emph{both} treatments ($A=0$ indicates assignment to discontinuation of both aspirin and antihypertensives, $A=1$ denotes assignment to continuing both medications). A similar intervention is considered in \citet{reeve_withdrawal_2020}. Thus, $(A_Y,A_D)$ is a physical decomposition of treatment $A$ such that receiving both components, i.e. $A_Y=A_D=1$, is equivalent to receiving $A=1$ (and conversely receiving neither component, $A_Y=A_D=0$, is equivalent to $A=0$), and therefore satisfies the modified treatment assumption (\ref{eq:modified_treatment_assumption}).

We consider a simplified setting where there is one binary pre-treatment and post-treatment common cause of future events (syncope and death), denoted by $L_0$ (old age at baseline) and $L_1$ (binarized blood pressure after treatment initiation) respectively.

In the data generating model, we first sampled $L_0,L_1$ according to
\begin{align*}
    P(L_0=1)&=\frac{1}{2} ~, \\
    P(L_1=1\mid A_Y) &= \frac{1}{2} + (2A_Y-1)\cdot\beta_{L_1,A_Y} ~.
\end{align*}
Next, we generated the  processes $(Y_k,D_k,C_k)$ on a discrete time grid using the hazards
\begin{align*}
    P(C_{k+1}=1|C_k=D_k=0,L_0,L_1,\overline{Y}_k,A_Y,A_D) &= \beta_{C,0} ~,\\
    P(D_{k+1}=1|C_{k+1}= D_k=0,L_0,L_1,\overline{Y}_k,A_Y,A_D) &= \beta_{D,0} + A_D\cdot\beta_{D,A}+ L_0 \cdot\beta_{D,L_0}  \\
    &\qquad + L_1 \cdot \beta_{D,L_1} + Y_{k}\cdot\beta_{D,Y} ~, \\
    P(\Delta Y_{k+1}=1|C_{k+1}=D_{k+1}=0,L_0,L_1,\overline{Y}_k,A_Y,A_D) &= \beta_{Y,0} + L_0 \cdot\beta_{Y,L_0} + L_1\cdot \beta_{Y,L_1} ~.
\end{align*}

The data generating model is constructed such that it satisfies all of the identification conditions for total effect, controlled direct effect and separable effect, and is consistent with the causal DAGs in Fig.~\ref{fig:DAG_motivating_example}. The implementation of the data generating model is shown in the online Supplementary Material.

\begin{figure} 
    \centering
\subfloat[]{
 \resizebox{0.5\columnwidth}{!}{
\begin{tikzpicture}
    \begin{scope}[every node/.style={thick,draw=none}]
    \node[name=Yk] at (2,1.5){$Y_k$};
    \node[name=Dk] at (2,-1.5){$D_k$};
    \node[name=Dk1] at (5.5,-1.5){$D_{k+1}$};
    \node[name=Yk1] at (5.5,1.5){$Y_{k+1}$};
    \node[name=AY] at (-1,1){$A_Y$};
    \node[name=AD] at (-1,-1){$A_D$};
    \node[name=L0] at (-3,0){$L_0$};
    \node[name=L0prime] at (1,0){$L_1$};
\end{scope}
\begin{scope}[>={Stealth[black]},
              every node/.style={fill=white,circle},
              every edge/.style={draw=black,very thick}]
    \path [->] (AY) edge (L0prime);
    \path [->] (L0prime) edge (Yk);
    \path [->] (L0prime) edge (Yk1);
    \path [->] (L0prime) edge (Dk);
    \path [->] (L0prime) edge (Dk1);
    \path [->] (AD) edge (Dk);
    \path [->] (AD) edge[bend right=25] (Dk1);
    \path [->] (Dk) edge (Yk);
    \path[->] (Yk) edge (Yk1);
    \path[->] (Yk) edge (Dk1);
    \path[->] (Dk1) edge (Yk1);
    \path[->] (Dk) edge (Dk1);
    \path [->] (L0) edge[bend right] (Dk);
    \path [->] (L0) edge[bend left] (Yk);
    \path [->] (L0) edge[bend right=35] (Dk1);
    \path [->] (L0) edge[bend left=35] (Yk1);
\end{scope}
\end{tikzpicture}
    }
    }
\subfloat[]{
 \resizebox{0.5\columnwidth}{!}{
\begin{tikzpicture}
    \begin{scope}[every node/.style={thick,draw=none}]
    \node[name=Yk] at (2,1.5){$Y_k$};
    \node[name=Dk] at (2,-1.5){$D_k$};
    \node[name=Dk1] at (5.5,-1.5){$D_{k+1}$};
    \node[name=Yk1] at (5.5,1.5){$Y_{k+1}$};
    \node[name=A] at (-1,0){$A$};
    \node[name=L0] at (-3,0){$L_0$};
    \node[name=L0prime] at (1,0){$L_1$};
\end{scope}
\begin{scope}[>={Stealth[black]},
              every node/.style={fill=white,circle},
              every edge/.style={draw=black,very thick}]
    \path [->] (A) edge (L0prime);
    \path [->] (A) edge (Dk);
    \path [->] (A) edge (Dk1);
    \path [->] (L0prime) edge (Yk);
    \path [->] (L0prime) edge (Yk1);
    \path [->] (L0prime) edge (Dk);
    \path [->] (L0prime) edge (Dk1);
    \path [->] (Dk) edge (Yk);
    \path[->] (Yk) edge (Yk1);
    \path[->] (Yk) edge (Dk1);
    \path[->] (Dk1) edge (Yk1);
    \path[->] (Dk) edge (Dk1);
    \path [->] (L0) edge[bend right] (Dk);
    \path [->] (L0) edge[bend left] (Yk);
    \path [->] (L0) edge[bend right=35] (Dk1);
    \path [->] (L0) edge[bend left=35] (Yk1);
\end{scope}
\end{tikzpicture}
    }
    }

\caption{(A) Causal graph illustrating a hypothetical four armed trial where $A_Y$ (antihypertensive agents) and $A_D$ (aspirin) are assigned freely. (B)  Shows the observed two armed trial where $A_Y$ and $A_D$ are assigned jointly ($A_Y\equiv A_D\equiv A$).  Censoring nodes $C_k$ are not shown because they are not connected to any other nodes under the chosen data generating model.}
\label{fig:DAG_motivating_example}
\end{figure}

Fig.~\ref{fig:DAG_motivating_example} encodes the assumption that only the $A_Y$ component (antihypertensive treatment) affects $L_1$ and does not directly affect death while the $A_D$ component (aspirin) acts directly on survival  and has no effect on the recurrence of syncope except through pathways intersected by survival  ($D_k$).  The parameters of the data generating model were chosen such that antihypertensive treatment ($A_Y=1$) increases the risk of death through the pathway $A_Y\rightarrow L_1\rightarrow Y_k\rightarrow D_{k+1}$, i.e. by lowering blood pressure, which in turn may lead to syncope and subsequent injurious falls.
This is seen in Fig.~\ref{fig:sim_plots_motivating_example} (A), as individuals who received antihypertensives ($A_Y=1$), shown by the black and red curves, experience a larger number of syncope episodes. Next, treatment with aspirin decreases the risk of  death through cardiovascular protection via the pathway $A_D\rightarrow D_{k+1}$. As illustrated by the crossing of the black and blue curves in Fig.~\ref{fig:sim_plots_motivating_example} (B),  the decreased risk of death due  to aspirin through pathway $A_D\rightarrow D_{k+1}$ is compensated by the increased risk of death under antihypertensive treatment through $A_Y\rightarrow L_1\rightarrow Y_k\rightarrow D_{k+1}$.  Therefore, discontinuation of antihypertensives only, i.e. $(A_Y=0,A_D=1)$, gives the highest survival in this example. This illustrates the role and interpretation of the separable effect; even though a trial investigator only observes individuals in treatment levels $A\in\{0,1\}$, the separable effects allows us to make inference under the hypothetical decomposed intervention $(A_Y=0,A_D=1)$, which is not possible using conventional estimands such as the total effect or controlled direct effect.

\begin{figure}
    \centering
    
\subfloat[]{
    \resizebox{0.5\columnwidth}{!}{
        \includegraphics{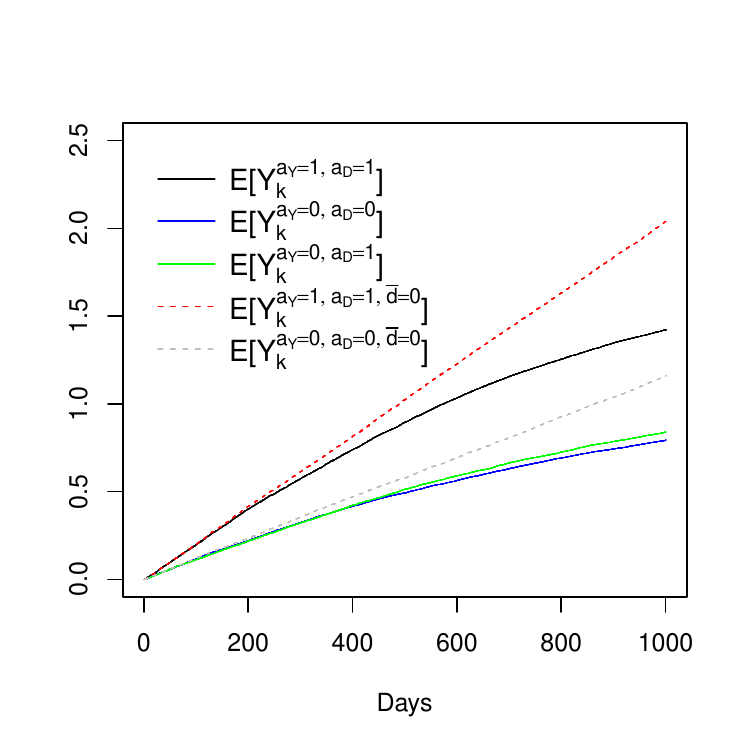}
    }
} 
\subfloat[]{
    \resizebox{0.5\columnwidth}{!}{
        \includegraphics{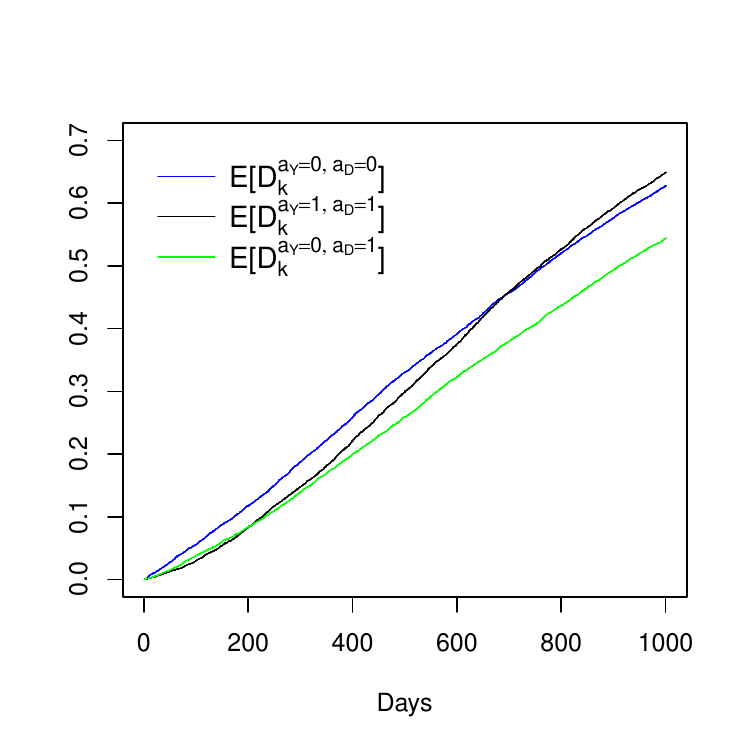}
    }
} 
\\
\subfloat[]{
    \resizebox{0.5\columnwidth}{!}{
        \includegraphics{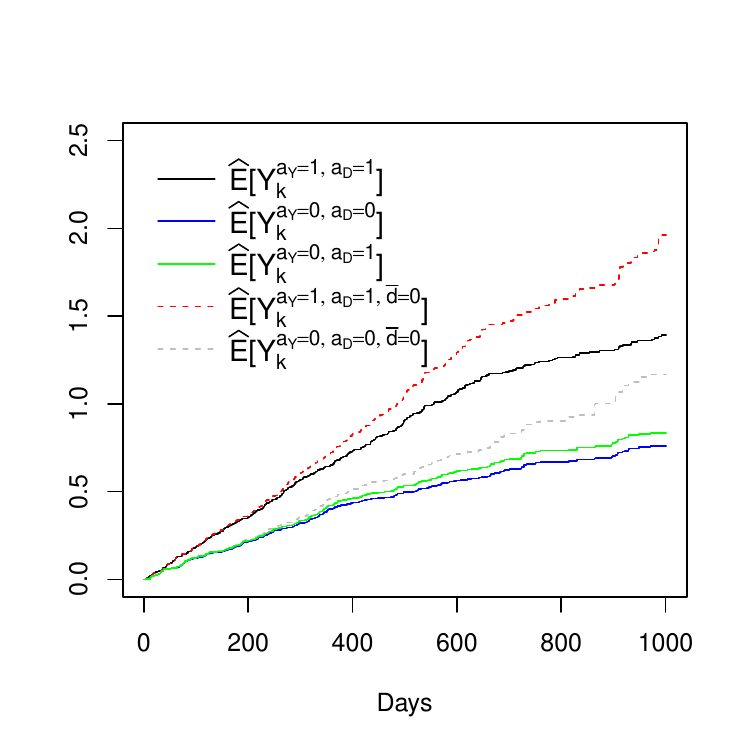}
    }
} 
\subfloat[]{
    \resizebox{0.5\columnwidth}{!}{
        \includegraphics{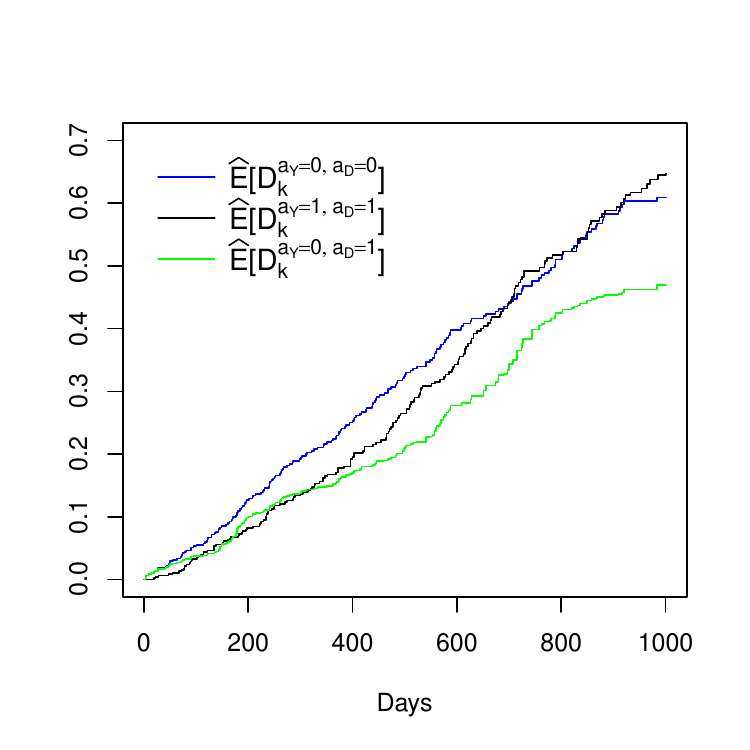}
    }
} 
    \caption{Different effects of treatment (dis)continuation of antihypertensives ($A_Y$) and aspirin ($A_D$) on recurrent syncope and survival are shown in (A) and (B), under a hypothetical data generating model. The superscript $\overline{c}=0$, denoting an intervention to prevent loss to follow-up, has been suppressed in plot legends to reduce clutter. Estimates using the risk set estimator in (\ref{eq: weighted risk set estimator}) 
    with a sample size of 500 individuals in each treatment arm are shown in (C) and (D). }
    \label{fig:sim_plots_motivating_example}
\end{figure}

By transforming the graphs in Fig.~\ref{fig:DAG_motivating_example} to single world intervention graphs~\citep{richardson_single_2013} corresponding to Figures~\ref{fig:total_effect}-\ref{fig:separable_effect}, it is straightforward to verify that the exchangeability (\ref{eq:exchangeability_total_i})-(\ref{eq:exchangeability_total_ii}), (\ref{eq:exchangeability_direct_i})-(\ref{eq:exchangeability_direct_ii}) and (\ref{eq:exchangeability_separable_i})-(\ref{eq:exchangeability_separable_ii}) and dismissible component conditions (\ref{eq:dismissible_component_i})-(\ref{eq:dismissible_component_iv}) are satisfied by the causal model. Because all positivity and consistency conditions also hold by construction in the data generating model, the total, controlled direct and separable effects are identified by the respective functionals given in Secs.~\ref{sec:identification_total_effect}-\ref{subsection: separable effects identification}. Furthermore, we can see from Fig.~\ref{fig:DAG_motivating_example} that strong $A_Y$ partial isolation (\ref{eq:isolation}) is violated by the path $A_Y\rightarrow Y_k \rightarrow D_{k+1}$, and thus the effect of antihypertensives  ($A_Y$) on recurrent syncope ($Y_k$) cannot be interpreted as a direct effect outside of pathways intersected by survival ($D_k$).

\begin{table}[h!]
    \centering
    \caption{The left column includes the weights, the middle column includes the hazards that define (\ref{eq: stochastic exponential}), and the right column includes the parametric hazard models that were used in the data analysis.  
    }
    \begin{tabular}{ |p{ 2 cm }|p{ 1.3 cm }|p{ 8 cm }|}
         \hline 
          \textbf{Weight}     &  \textbf{Hazards} &  \textbf{Hazard models fitted}\\
         \hline
            \multirow{2}{0.2cm}{$ \mathcal W_{C,t}^i$}
           & $ \alpha_t^idt$ &  $\beta_t^0 + \beta_t^A A + \beta_t^{L_0} L_0 + \beta_t^{L_1} L_1 + \beta_t^{Y} Y_{t-} $    \\ \cline{2-3}
            & $\alpha_t^{*,i} dt$ &   $\beta_t^0 + \beta_t^A A$   \\
           \hline
            \multirow{2}{0.2cm}{$ \mathcal W_{D,t}^i$}
           & $ \alpha_t^idt$ &  $\beta_t^0 + \beta_t^A A + \beta_t^{L_0} L_0 + \beta_t^{L_1} L_1 + \beta_t^{Y} Y_{t-}$   \\ \cline{2-3}
            & $\alpha_t^{*,i} dt$ &   $\beta_t^0 + \beta_t^A A$   \\
            \cline{1-3}
            \multirow{2}{0.2cm}{$ \mathcal W_{D,t}^i(a_Y,a_D)$}
           & $ \alpha_t^idt$ &  $\beta_t^0 + \beta_t^A I(A=a_Y) + \beta_t^{L_0} L_0 + \beta_t^{L_1} L_1 + \beta_t^{Y} Y_{t-}$    \\ \cline{2-3}
            & $\alpha_t^{*,i} dt$ &   $\beta_t^0 + \beta_t^A I(A=a_D) + \beta_t^{L_0} L_0 + \beta_t^{L_1} L_1 + \beta_t^{Y} Y_{t-}$  \\
         \hline
\end{tabular}
    
    \label{tab:sim_assumed models}
\end{table}

\subsection{Estimates}
Figures~\ref{fig:sim_plots_motivating_example}~(c) and (d) present estimates of $E[Y_k^{a,\overline{c}=0}]$ for $a=a_Y=a_D\in\{0,1\}$ and $E[Y_k^{a_Y,a_D,\overline{c}=0}]$ for $a_Y\neq a_D$ using the estimators described in Sec.~\ref{sec:estimation_maintext} for 500 simulated individuals in each of the treatment groups $A=0$ and $A=1$. It follows from the data generating model defined in the beginning of Appendix  \ref{sec:sim_examples_revisited} and Table \ref{tab:sim_assumed models} that the assumed hazard models are correctly specified. In principle, we could allow for a more involved data generating model with time-varying coefficients. The assumed hazard models would still provide consistent estimates if the additive structure is correctly specified. Thus, an investigator can adapt the estimators in the supplementary R material to other recurrent event problems.

\section{Isolation conditions\label{sec:app_isolation}}
Following~\citet{stensrud_generalized_2021}, we define the $A_Y$ separable effect on $Y_k$ as
\begin{align}
    E[Y_k^{a_Y=1,a_D}]~\textrm{vs}~E[Y_k^{a_Y=0,a_D}] ~. \label{eq:app_AY_effect}
\end{align}
Likewise, we define the $A_D$ separable effect on $Y_k$ as 
\begin{align}
    E[Y_k^{a_Y,a_D=1}]~\textrm{vs}~E[Y_k^{a_Y,a_D=0}] ~. \label{eq:app_AD_effect}
\end{align}

Next, following~\citet{stensrud_generalized_2021}, we define two isolation conditions:
\begin{definition}[Strong $A_Y$ partial isolation]
A treatment decomposition satisfies strong $A_Y$ partial isolation if
\begin{align}
    \textrm{There are no causal paths from $A_Y$ to $D_k$ for all $k\in\{0,\dots,K+1\}$ ~.} \label{eq:app_a_Y_isolation}
\end{align}
\end{definition}

\begin{definition}[$A_D$ partial isolation]
A treatment decomposition satisfies $A_D$ partial isolation if
\begin{align}
    &\textrm{The causal paths from $A_D$ to $Y_{k+1}$, $k=0,\dots,K$ are directed}\notag\\
    &\textrm{paths intersected by $D_{j+1}$, $j\in\{0,\dots,K\}$ ~.} \label{eq:app_a_D_isolation}
\end{align}
\end{definition}

Under strong $A_Y$ partial isolation, the $A_Y$ separable effects \emph{only} capture direct effects of $A_Y$ on $Y_k$, i.e. only pathways from $A_Y$ to $Y_k$ \emph{not} intersected by $D$. Under $A_D$ partial isolation, the $A_D$ separable effects only capture indirect effects of $A_D$ on $Y_k$, that is only pathways from $A_D$ to $Y_k$ that are intersected by $D$.

If a treatment decomposition satisfies both (\ref{eq:app_a_Y_isolation}) and~\eqref{eq:app_a_D_isolation}, it is said to satisfy full isolation. Under full isolation, (\ref{eq:app_AY_effect})-(\ref{eq:app_AD_effect}) are the separable direct and indirect effects on $Y_k$ respectively. In this case, (\ref{eq:app_AY_effect}) captures \emph{all} pathways from $A$ to $Y_k$ not intersected by $D$, and (\ref{eq:app_AD_effect}) captures \emph{all} pathways from $A$ to $Y_k$ intersected by $D$.

\section{Proof of identification results \label{sec:proof_identification}}

\subsection{Total effect\label{sec:app_total}}
Assume the following identification conditions hold for $k\in\{0,\dots,K\}$.

\textbf{Exchangeability}
\begin{align*}
    \overline{Y}_{K+1}^{a,\overline{c}=0}   &\independent A | L_0    ~,\\
    \underline{Y}_{k+1}^{a,\overline{c}=0} &\independent C_{k+1}^{a,\overline{c}=0} | \overline{L}_k^{a,\overline{c}=0}, \overline{Y}_k^{a,\overline{c}=0}, \overline{D}_k^{a,\overline{c}=0},\overline{C}_k^{a,\overline{c}=0}, A ~.  
\end{align*}

\textbf{Positivity}
\begin{align*}
&P(L_0=l_0) > 0 \implies P(A=a\mid L_0=l_0) >0 ~,  \\
    &f_{A,\overline{L}_k,\overline{D}_k,\overline{C}_k,\overline{Y}_k}(a,\overline{l}_k,0,0,\overline{y}_k)>0  \notag\\
    &\quad \implies P(C_{k+1}=0\mid A=a, \overline{L}_k=\overline{l}_k,\overline{D}_k=0,\overline{C}_k=0,\overline{Y}_k=\overline{y}_k) > 0 ~.  
\end{align*}

\textbf{Consistency}
\begin{align*}
    &\text{If $A=a$ and $\overline{C}_{k+1}=0$,} \notag\\
    &\text{then $\overline{L}_{k+1}=\overline{L}_{k+1}^{a,\overline{c}=0},\overline{D}_{k+1}=\overline{D}_{k+1}^{a,\overline{c}=0}, \overline{Y}_{k+1}=\overline{Y}_{k+1}^{a,\overline{c}=0}$}, \overline{C}_{k+1}=\overline{C}_{k+1}^{a,\overline{c}=0}~. 
\end{align*}

Next, we use the identification conditions and the law of total probability (LOTP) to add variables sequentially to the conditioning set in temporal order. We have that $P(\Delta Y_k^{a,\overline{c}=0}=\Delta y_k)$ is given by
\begin{align*}
    & \sum_{l_0} P(\Delta Y_k^{a,\overline{c}=0}=\Delta y_k\mid L_0=l_0)P(L_0=l_0) \\
    &\stackrel{ \overline{Y}_{K+1}^{a,\overline{c}=0}\independent A\mid L_0}{=} \sum_{l_0} P(\Delta Y_k^{a,\overline{c}=0}=\Delta y_k\mid A=a,L_0=l_0)P(L_0=l_0) \\
    &\stackrel{ \underline{Y}_0^{a,\overline{c}=0} \independent C_0^a\mid A,L_0 }{=} \sum_{l_0} P(\Delta Y_k^{a,\overline{c}=0}=\Delta y_k\mid C_0^a=0, A=a,L_0=l_0)P(L_0=l_0) \\
    &\stackrel{\text{LOTP}}{=} \sum_{d_0}\sum_{l_0} P(\Delta Y_k^{a,\overline{c}=0}=\Delta y_k\mid D_0^{a,c_0=0}=0,  C_0^a=0, A=a,L_0=l_0) \\
    &\qquad \times P( D_0^{a,c_0=0}=0 \mid C_0^a=0, A=a, L_0=l_0)P(L_0=l_0) \\
    &\stackrel{\text{LOTP}}{=} \sum_{\Delta y_0} \sum_{d_0}\sum_{l_0} P(\Delta Y_k^{a,c_0=0}=\Delta y_k\mid \Delta Y_0^{a,\overline{c}=0}=\Delta y_0, D_0^{a,c_0=0}=0,  C_0^a=0, A=a,L_0=l_0) \\
    &\qquad P(\Delta Y_0^{a,c_0=0}=\Delta y_0\mid  D_0^{a,c_0=0}=0,  C_0^a=0, A=a,L_0=l_0)\\
    &\qquad \times P( D_0^{a,c_0=0}=0 \mid C_0^a=0, A=a, L_0=l_0)P(L_0=l_0) ~. \\
\end{align*}
The conditional independence relation in the third line follows because $C_0\equiv 0$ deterministically. Iterating this procedures for time indices $k^\prime\in\{1,\dots, k\}$ gives
\begin{align*}
    &P(\Delta Y_k^{a,\overline{c}=0}=\Delta y_k) \\
    &=\sum_{\Delta\overline{y}_{k-1}}\sum_{\overline{d}_k}\sum_{\overline{l}_{k-1}}\prod_{j=0}^{k} \\
    &\qquad P(\Delta Y_j^{a,\overline{c}=0}=\Delta y_j\mid \overline{D}_{j}^{a,\overline{c}=0}=\overline{d}_j, \overline{C}_j^{a,\overline{c}=0}=0,\overline{L}_{j-1}^{a,\overline{c}=0}=\overline{l}_{j-1},\Delta \overline{Y}_{j-1}^{a,\overline{c}=0}=\Delta \overline{y}_{j-1},A=a) \\
    &\qquad \times P(D_{j}^{a,\overline{c}=0}=d_{j}\mid \overline{C}_{j}^{a,\overline{c}=0}=0, \overline{L}_{j-1}^{a,\overline{c}=0}=\overline{l}_{j-1}, \Delta\overline{Y}_{j-1}^{a,\overline{c}=0}=\Delta\overline{y}_{j-1}, \overline{D}_{j-1}^{a,\overline{c}=0}=\overline{d}_{j-1}, A=a ) \\
    &\qquad \times P(L_{j-1}^{a,\overline{c}=0}=l_{j-1}\mid \Delta\overline{Y}_{j-1}^{a,\overline{c}=0}=\Delta\overline{y}_{j-1}, \overline{D}_{j-1}^{a,\overline{c}=0}=\overline{d}_{j-1},\overline{C}_{j-1}^{a,\overline{c}=0}=0,\overline{L}_{j-2}^{a,\overline{c}=0}=\overline{l}_{j-2},A=a) ~.
\end{align*}
Positivity ensures that the conditioning sets on RHS have a non-zero probability. During the iterative procedure, we have use exchangeability with respect to censoring (\ref{eq:exchangeability_total_ii}) to add the censoring indicator to the conditioning set.  Finally, by consistency we have that
\begin{align*}
    P(&\Delta Y_k^{a,\overline{c}=0}=\Delta y_k) \\
    &=\sum_{\Delta\overline{y}_{k-1}}\sum_{\overline{d}_k}\sum_{\overline{l}_{k-1}}\prod_{j=0}^{k} \\
    &\qquad P(\Delta Y_j=\Delta y_j\mid \overline{D}_{j}=\overline{d}_j, \overline{C}_j=0,\overline{L}_{j-1}=\overline{l}_{j-1},\Delta \overline{Y}_{j-1}=\Delta \overline{y}_{j-1},A=a) \\
    &\qquad \times P(D_{j}=d_{j}\mid \overline{C}_{j}=0, \overline{L}_{j-1}=\overline{l}_{j-1}, \Delta\overline{Y}_{j-1}=\Delta\overline{y}_{j-1}, \overline{D}_{j-1}=\overline{d}_{j-1}, A=a ) \\
    &\qquad \times P(L_{j-1}=l_{j-1}\mid \Delta\overline{Y}_{j-1}=\Delta\overline{y}_{j-1}, \overline{D}_{j-1}=\overline{d}_{j-1},\overline{C}_{j-1}=0,\overline{L}_{j-2}=\overline{l}_{j-2},A=a) ~.
\end{align*}

Next, we derive (\ref{eq:total_effect_IPCW_discrete}).  By the presence of the indicator functions in (\ref{eq:total_effect_IPCW_discrete}) and by consistency (\ref{eq:consistency_total_effect}) we have that
\begin{align*}
    &E\left[ \frac{I(A=a)I(C_i=0)\Delta Y_i}{\pi_A(A)\prod_{j=0}^{i} \pi_{C_j}(C_j) }\right] \\
    =&E\left[ \frac{I(A=a)I(C_i=0)\Delta Y_i^{a,\overline{c}=0}}{\pi_A(A)\prod_{j=0}^{i} \pi_{C_j}(C_j) }\right]\\
    =&E\left[ \frac{I(A=a)I(C_{i}^{a,\overline{c}=0}=0)\Delta Y_i^{a,\overline{c}=0}}{P(A=a\mid L_0)\prod_{j=0}^i P(C_j^{a,\overline{c}=0}=0\mid \overline{L}_{j-1}^{a,\overline{c}=0},\overline{Y}_{j-1}^{a,\overline{c}=0},\overline{D}_{j-1}^{a,\overline{c}=0},\overline{C}_{j-1}^{a,\overline{c}=0},A) } \right] ~.
\end{align*}
Next, using the law of total expectation, the above is equal to
\begin{align*}
    &E\bigg[E\bigg[ \frac{I(A=a)I(C_{i}^{a,\overline{c}=0}=0)\Delta Y_i^{a,\overline{c}=0}}{P(A=a\mid L_0)\prod_{j=0}^i P(C_j^{a,\overline{c}=0}=0\mid \overline{L}_{j-1}^{a,\overline{c}=0},\overline{Y}_{j-1}^{a,\overline{c}=0},\overline{D}_{j-1}^{a,\overline{c}=0},\overline{C}_{j-1}^{a,\overline{c}=0},A) } \\
    & \qquad\qquad\qquad ~\bigg|~ \overline{Y}_{i}^{a,\overline{c}=0},\overline{L}_{i-1}^{a,\overline{c}=0},\overline{D}_{i-1}^{a,\overline{c}=0},\overline{C}_{i-1}^{a,\overline{c}=0},A \bigg]\bigg] \\
    = &E\bigg[\frac{I(A=a)I(C_{i-1}^{a,\overline{c}=0}=0)\Delta Y_i^{a,\overline{c}=0}}{P(A=a\mid L_0)\prod_{j=0}^{i-1} P(C_j^{a,\overline{c}=0}=0\mid \overline{L}_{j-1}^{a,\overline{c}=0},\overline{Y}_{j-1}^{a,\overline{c}=0},\overline{D}_{j-1}^{a,\overline{c}=0},\overline{C}_{j-1}^{a,\overline{c}=0},A)}\\
    &\qquad \times \frac{E[I(C_i^{a,\overline{c}=0}=0)\mid \overline{Y}_{i}^{a,\overline{c}=0},\overline{L}_{i-1}^{a,\overline{c}=0},\overline{D}_{i-1}^{a,\overline{c}=0},\overline{C}_{i-1}^{a,\overline{c}=0},A]}{P(C_i^{a,\overline{c}=0}=0\mid \overline{L}_{i-1}^{a,\overline{c}=0}, \overline{Y}_{i-1}^{a,\overline{c}=0},\overline{D}_{i-1}^{a,\overline{c}=0},\overline{C}_{i-1}^{a,\overline{c}=0},A)}\bigg]\bigg] ~.
\end{align*}
The numerator and denominator of the fraction in the final line differ only by the time index of $\overline{Y}^{a,\overline{c}=0}$ in the conditioning set. Using the fact that
\begin{align*}
    Y_i^{a,\overline{c}=0}\independent I(C_i^{a,\overline{c}=0}=0)\mid Y_{i-1}^{a,\overline{c}=0}, L_{i-1}^{a,\overline{c}=0},D_{i-1}^{a,\overline{c}=0},C_{i-1}^{a,\overline{c}=0}=0,A=a ~,
\end{align*}
(which follows from (\ref{eq:exchangeability_total_ii})) we have that the fraction is equal to 1, and thus by consistency (\ref{eq:consistency_total_effect}),
\begin{align*}
    E\left[ \frac{I(A=a)I(C_i=0)\Delta Y_i^{a,\overline{c}=0}}{\pi_A(A)\prod_{j=0}^{i} \pi_{C_j}(C_j) }\right]=E\left[ \frac{I(A=a)I(C_{i-1}=0)\Delta Y_i^{a,\overline{c}=0}}{\pi_A(A)\prod_{j=0}^{i-1} \pi_{C_j}(C_j) }\right] ~.
\end{align*}
Iterating this procedure from $j=i-1$ to $j=0$ gives
\begin{align*}
    E\left[ \frac{I(A=a)I(C_i=0)\Delta Y_i^{a,\overline{c}=0}}{\pi_A(A)\prod_{j=0}^{i} \pi_{C_j}(C_j) }\right]=E\left[ \frac{I(A=a)\Delta Y_i^{a,\overline{c}=0}}{\pi_A(A) }\right] ~.
\end{align*}
Using the law of total expectation again, RHS is equal to
\begin{align*}
    &E\left[E\left[ \frac{\Delta Y_i^{a,\overline{c}=0}}{P(A=a\mid L_0)}\cdot I(A=a) ~\bigg|~ \Delta Y_i^{a,\overline{c}=0}, L_0 \right]\right] \\
    =&E\left[ \frac{\Delta Y_i^{a,\overline{c}=0}}{P(A=a\mid L_0) }\cdot E\left[ I(A=a) \mid \Delta Y_i^{a,\overline{c}=0}, L_0  \right]\right] \\
    \stackrel{\Delta Y_i^{a,\overline{c}=0}\independent I(A=a)\mid L_0}{=} &E \left[ \Delta Y_i^{a,\overline{c}=0} \cdot \frac{E[I(A=a)\mid L_0]}{P(A=a\mid L_0)}\right] \\
    =&E[\Delta Y_i^{a,\overline{c}=0}] ~.
\end{align*}
Another IPW representation also exists. We have that
\begin{align*}
    &E\bigg[ \frac{I(A=a)I(C_i=0)}{\pi_A(A)\prod_{j=0}^i \pi_{C_j}(C_j)}   \bigg] \\
    = &E\bigg[E\bigg[\frac{I(A=a)I(C_{i-1}=0)}{\pi_A(A)\prod_{j=0}^{i-1} \pi_{C_j}(C_j)} \cdot \frac{I(C_i=0)}{P(C_i=0\mid \overline{L}_{i-1},\overline{Y}_{i-1},\overline{D}_{i-1},\overline{C}_{i-1},A=a)} \\
    &\qquad\qquad\qquad   ~\bigg|~ \overline{L}_{i-1},\overline{Y}_{i-1},\overline{D}_{i-1},\overline{C}_{i-1},A=a \bigg]\bigg] \\
    =&E\bigg[\frac{I(A=a)I(C_{i-1}=0)}{\pi_A(A)\prod_{j=0}^{i-1} \pi_{C_j}(C_j)} \cdot\frac{E[I(C_{i}=0)\mid \overline{L}_{i-1},\overline{Y}_{i-1},\overline{D}_{i-1},\overline{C}_{i-1},A=a]}{P(C_i=0\mid \overline{L}_{i-1},\overline{Y}_{i-1},\overline{D}_{i-1},\overline{C}_{i-1},A=a)} \bigg]\\
    =&E\bigg[ \frac{I(A=a)I(C_{i-1}=0)}{\pi_A(A)\prod_{j=0}^{i-1} \pi_{C_j}(C_j)}   \bigg]~.
\end{align*}
Arguing iteratively from $j=i-1$ to $j=0$, the RHS is equal to
\begin{align*}
    E\left[ \frac{I(A=a)}{P(A=a\mid L_0)} \right] &= E\left[E\left[ \frac{I(A=a)}{P(A=a\mid L_0)} ~\bigg|~L_0 \right]\right] \\
    &=E\left[ \frac{E[I(A=a)\mid L_0]}{P(A=a\mid L_0)} \right] \\
    &=1 ~.
\end{align*}
Putting everything together, we have that
\begin{align}
    E[Y_k^{a,\overline{c}=0}] &= \sum_{i=0}^k \frac{E\bigg[ \frac{I(A=a)I(C_i=0)}{\pi_A(A)\prod_{j=0}^i \pi_{C_j}(C_j)} \cdot \Delta Y_i \bigg] }{E\bigg[ \frac{I(A=a)I(C_i=0)}{\pi_A(A)\prod_{j=0}^i \pi_{C_j}(C_j)}   \bigg] } ~. \label{eq:app_total_effect_IPW}
\end{align}
To conclude, we remark that the exchangeability conditions (\ref{eq:exchangeability_total_i})-(\ref{eq:exchangeability_total_ii}) and identification formulas (\ref{eq:g-formula_total_effect})-(\ref{eq:total_effect_IPCW_discrete}) follow directly from a general identification result, Theorem~31 of ~\citet{richardson_single_2013}, by choosing outcome $Y^\ast_{k}\equiv(Y_{k}, D_{K+1})$, intervention set $\overline{A}^\ast_k\equiv (A,C_k)$ for $k\in\{0,\dots,K+1\}$ and time-varying covariates $\overline{L}^\ast_k\equiv (\overline{L}_k, L)$ for $k\in\{0,\dots,K\}$.

\subsubsection{Limit of fine discretizations}
We begin by noting that $\Delta Y_k$ in (\ref{eq:app_total_effect_IPW}) can only be non-zero if the individual has not experienced the competing event by the beginning of time interval $k-1$. Therefore,
\begin{align*}
    E[\Delta Y_i^{a,\overline{c}=0}] &= E\left[\frac{I(A=a)I(\overline{C}_i=0)I(\overline{D}_{i-1}=0)\Delta Y_i}{\pi_A\prod_{j=0}^i \pi_{C_j}(C_j)}\right] \\
    &=P(\overline{C}_i=0,\overline{D}_{i-1}=0,A=a) \\
    &\qquad\times E\left[\frac{\Delta Y_i}{\pi_A\prod_{j=0}^i \pi_{C_k}(C_k)}~\bigg|~ A=a,C_i=0,D_{i-1}=0\right] ~,
\end{align*}
where we have used the laws of probability in the second line. Using Bayes' law sequentially, we have that
\begin{align*}
    P(\overline{C}_i=0,\overline{D}_{i-1}=0,A=a) &=\prod_{j=0}^{i-1}\bigg[P(D_j=0\mid C_j=0,D_{j-1}=0,A=a)\bigg] \\
    &\quad\times\prod_{n=0}^{i} \bigg[P(C_n=0\mid D_{n-1}=0,C_{n-1}=0,A=a) \bigg] P(A=a) ~.
\end{align*}

To proceed, we define modified intensities of the recurrent event process
\begin{align*}
    \Delta\Uplambda_i^{C} &= P(C_i=1\mid \overline{D}_{i-1},\overline{C}_{i-1},A) \\
    \Delta\Uplambda_i^{C\mid\mathcal{F}} &= P(C_i=1\mid \overline{L}_{i-1},\overline{Y}_{i-1}, \overline{D}_{i-1},\overline{C}_{i-1},A) ~.
\end{align*}
Next, let $\pi(\bullet)=P(A=\bullet)$ and consider the stabilized weights
\begin{align*}
    W_A&= \frac{\pi(A)}{\pi_A(A)} ~, \\
    W_{C,i}&= \frac{\prod_{j=0}^i [1-\Delta \Uplambda_j^C]}{\prod_{k=0}^i \left[1-\Delta \Uplambda_k^{C\mid\mathcal{F}}\right]} ~.
\end{align*}
The weight $W_{C,i}$ is a ratio of Kaplan-Meier survival terms with the respect to the censoring event. Let us also define the hazard of the competing event by
\begin{align*}
    \Delta A^D_i(a) = P(D_i=1\mid C_i=0,D_{i-1}=0,A=a) ~.
\end{align*}
Putting everything together, we have that
\begin{align}
    E[Y_k^{a,\overline{c}=0}] = \sum_{i=0}^k\prod_{j=0}^{i-1}[1-\Delta A_j^D(a)] \cdot E[W_A W_{C,i}\Delta Y_i\mid C_i=0,D_{i-1}=0,A=a] ~. \label{eq:weighted_canonical_total_effect} 
\end{align}
Expression (\ref{eq:weighted_canonical_total_effect}) enables us to establish a correspondence with estimands in the counting process literature, as discussed in Sec.~\ref{section:correspondence of identification formulas}, and also motivates estimators that we described in Sec.~\ref{sec:estimation_maintext}. The product term in (\ref{eq:weighted_canonical_total_effect}) is a survival term with respect to the competing event, and the expectation is over weighted increments of recurrent acute kidney injury. In the limit of fine discretization of time, (\ref{eq:weighted_canonical_total_effect}) converges to
\begin{align*}
  \int_0^{t_k}\prodi_{s < u}[1- dA_s^D(a)] \cdot E[W_A \mathcal{W}_{C,u-} dN_u \mid C\geq u, T^D\geq u,A=a] ~.
\end{align*}

\subsubsection{Competing event \label{sec_app_total_effect_competing}}
In order to identify $E[D_k^{a,\overline{c}=0}]$ from the observed data, we require the following two exchangeability assumptions instead of (\ref{eq:exchangeability_total_i}) and (\ref{eq:exchangeability_total_ii})
\begin{align}
    \overline{D}_{K+1}^{a,\overline{c}=0}   &\independent A | L_0         \label{eq:exchangeability_total_terminating_baseline}  ~,\\
    \underline{D}_{k+1}^{a,\overline{c}=0} &\independent C_{k+1}^{a,\overline{c}=0} | \overline{L}_k^{a,\overline{c}=0}, \overline{Y}_k^{a,\overline{c}=0}, \overline{D}_k^{a,\overline{c}=0},\overline{C}_k^{a,\overline{c}=0}, A ~. \label{eq:exchangeability_total_terminating}
\end{align}
Using analogous arguments as for the recurrent event $Y$, identification of $E[D_k^{a,\overline{c}=0}]$ is achieved under (\ref{eq:exchangeability_total_terminating_baseline})-(\ref{eq:exchangeability_total_terminating}) and (\ref{positivity_total_effect})-(\ref{eq:consistency_total_effect})  by
\begin{align}
    E[D_k^{a,\overline{c}=0}] = \sum_{i=0}^k\prod_{j=0}^{i-1}[1-\Delta A_j^D(a)] \cdot E[W_A W_{C,i}\Delta D_i\mid C_i=0,D_{i-1}=0,A=a] ~. \label{eq:survival_weighted_canonical_total_effect} 
\end{align}
Likewise, in the limit of fine discretizations of time, the cumulative incidence of the competing event is given by
\begin{align}
  \int_0^{t_k}\prodi_{s < u}[1- dA_s^D(a)] \cdot E[W_A \mathcal{W}_{C,u-} dN_u^D \mid C\geq u, T^D\geq u,A=a] ~. \label{eq:survival_total_weighted_form_continuous}
\end{align}
When treatment $A$ is randomly assigned and (\ref{eq:survival_continuous_independent_censoring}) holds with $L(t)= N^c(t)= \emptyset$ (which is the usual independent censoring condition in survival analysis without any covariates~\citep{aalen_survival_2008}), then $W_A=\mathcal{W}_{C,t}=1$ and (\ref{eq:survival_total_weighted_form_continuous}) reduces to
\begin{align*}
     \int_0^{t_k}  \prodi_{u< t} [1-dA^D_u(a)]dA^D_u(a) ~.
\end{align*}
This demonstrates sufficient conditions under which the discrete time identification formula given by Expression (29) in \citet{young_causal_2020} converges to the usual representation of the cumulative incidence function in survival analysis.

\subsection{Controlled direct effect}
The identification  conditions and identification formulas for the controlled direct effect  are a special case of the identification results for total effect, redefining the censoring indicator as $\max (C_i,D_i)$ (i.e. the first occurrence of the competing event and loss to follow-up), and re-defining the competing event as an event that almost surely does not occur. This gives us

\begin{align}
    E[\Delta Y_i^{a,\overline{c}=\overline{d}=0}] &=  
    E\bigg[ \frac{I(A=a)I(C_i=0)I(D_i=0)}{\pi_A(A) \prod_{j=0}^i \pi_{C_j}(C_j)\pi_{D_j}(D_j) } \cdot \Delta Y_i \bigg]~, \label{eq:app_direct_effect_IPCW_discrete}
\end{align}

Next, we define
\begin{align*}
    &\Delta\Uplambda_j^{D}(\bullet)= P(D_j=1\mid \overline{C}_j,\overline{D}_{j-1},A=\bullet) ~,\\
    &\Delta\Uplambda_j^{D\mid\mathcal{F}}(\bullet)= P(D_j=1\mid \overline{C}_j,\overline{L}_{j-1},\overline{Y}_{j-1},\overline{D}_{j-1},A=\bullet)
\end{align*}
to be modified intensities of the competing event process. This allows us to re-write (\ref{eq:app_direct_effect_IPCW_discrete}) as
\begin{align}
    E[Y_k^{a,\overline{c}=\overline{d}=0}] = \sum_{i=0}^k E\left[ W_A W_{C,i}W_{D,i}\Delta Y_i ~\bigg|~ C_i=0,D_{i-1}=0,A=a \right] ~,\label{eq:weighted_canonical_direct_effect}
\end{align}
where
\begin{align}
    W_{D,i} &= \frac{\prod_{j=0}^i [1-\Delta \Uplambda_j^D]}{\prod_{k=0}^i [1-\Delta \Uplambda_k^{D\mid\mathcal{F}}]} ~. \label{eq: controlled direct discrete D weight}
\end{align}

Under randomization of $A$ and under the strong independent censoring assumption (\ref{eq:strong_indep_censoring}), $W_A=W_{C,i}=W_{D,i}=1$ and thus (\ref{eq:weighted_canonical_direct_effect}) converges to
\begin{align*}
  \int_0^{t_k} E\left[ dN_u \mid  T^D \geq u, C \geq u,A=a \right] 
\end{align*}
in the limit of fine discretizations of time.

\subsection{Separable effects}
We begin by assuming the modified treatment assumption (\ref{eq:modified_treatment_assumption}) and the following identification conditions for all $a\in\{0,1\}$, $k\in\{0,\dots,K+1\}$.

\textbf{Exchangeability}
\begin{align*}
    &(\overline{Y}_{K+1}^{a,\overline{c}=0},\overline{D}_{K+1}^{a,\overline{c}=0},\overline{L}_{K+1}^{a,\overline{c}=0}) \independent A\mid L_0 ~,  \\
    &(\underline{Y}_{k+1}^{a,\overline{c}=0},\underline{D}_{k+1}^{a,\overline{c}=0},\underline{L}_{k+1}^{a,\overline{c}=0})\independent C_{k+1}^{a,\overline{c}=0} \mid \overline{Y}_k^{a,\overline{c}=0}, \overline{D}^{a,\overline{c}=0}_k,\overline{C}_k^{a,\overline{c}=0},\overline{L}_k^{a,\overline{c}=0},A ~.
\end{align*}

\textbf{Positivity}
\begin{align}
    &P(L_0=l_0) > 0 \notag\\
    &\quad\implies P(A=a\mid L_0=l_0)>0 ~,\label{eq:app_positivity_separable_i}\\
    &f_{\overline{L}_k,\overline{D}_{k+1},C_{k+1},Y_k}(\overline{l}_k,0,0,\overline{y}_k)>0 \implies\notag\\
    &\quad P(A=a\mid \overline{D}_{k+1}=0,C_{k+1}=0,\overline{Y}_{k}=\overline{y}_{k},\overline{L}_k=\overline{l}_k) >0 \label{eq:app_positivity_separable_ii} \\
     &f_{A,\overline{L}_k,\overline{D}_k,\overline{C}_k,\overline{Y}_k}(a,\overline{l}_k,0,0,\overline{y}_k)>0 \notag\\
     &\quad \implies P(C_{k+1}=0\mid \overline{L}_k=\overline{l}_k,\overline{D}_k=0,\overline{C}_k=0,\overline{Y}_k=\overline{y}_k) ~. \label{eq:app_positivity_separable_iii}
\end{align}

\textbf{Consistency}
\begin{align}
    &\text{If $A=a$ and $\overline{C}_{k+1}=0$,} \notag\\
    &\text{then $\overline{L}_{k+1}=\overline{L}_{k+1}^{a,\overline{c}=0},\overline{D}_{k+1}=\overline{D}_{k+1}^{a,\overline{c}=0}, \overline{Y}_{k+1}=\overline{Y}_{k+1}^{a,\overline{c}=0}$}~.\label{eq:app_consistency_separable}
\end{align}

Consider a four armed trial where the $A_Y$ and $A_D$ are randomly assigned, independently of each other. We require the following dismissible component conditions to hold in the four armed trial
\begin{align*}
    &Y_{k+1}^{\overline{c}=0} \independent A_D \mid A_Y, \overline{D}_{k+1}^{\overline{c}=0} , \overline{Y}_k^{\overline{c}=0} , \overline{L}_k^{\overline{c}=0}  ~,  \\
    &D_{k+1}^{\overline{c}=0}  \independent A_Y  \mid A_D , \overline{D}_k^{\overline{c}=0} ,\overline{Y}_k^{\overline{c}=0} ,\overline{L}_k^{\overline{c}=0}  ~,\\
    &L_{Y,k}^{\overline{c}=0}  \independent A_D  \mid A_Y ,\overline{Y}_k^{\overline{c}=0} ,\overline{D}_k^{\overline{c}=0} ,\overline{L}_{k-1}^{\overline{c}=0} ,L_{D,k}^{\overline{c}=0}  ~,  \\
    & L_{D,k}^{\overline{c}=0}  \independent A_Y \mid A_D , \overline{D}_k^{\overline{c}=0} , \overline{Y}_k^{\overline{c}=0} ,\overline{L}_{k-1}^{\overline{c}=0}  ~.
\end{align*}
 
To proceed, we introduce the following lemmas:
\begin{lemma}\label{lemma:dismiss_counterfactual_component}
Under a FFRCISTG model, the dismissible component conditions (\ref{eq:dismissible_component_i})-(\ref{eq:dismissible_component_iv}) imply the following equalities for $a_Y,a_D\in\{0,1\}$
\begin{align}
    &P(Y_{k+1}^{a_Y,a_D=0,\overline{c}=0}= y_{k+1}\mid \overline{Y}_k^{a_Y,a_D=0,\overline{c}=0}, \overline{D}_{k+1}^{a_Y,a_D=0,\overline{c}=0}, \overline{L}_k^{a_Y,a_D=0,\overline{c}=0} ) \notag\\
   =&P(Y_{k+1}^{a_Y,a_D=1,\overline{c}=0}= y_{k+1}\mid \overline{Y}_k^{a_Y,a_D=1,\overline{c}=0}, \overline{D}_{k+1}^{a_Y,a_D=1,\overline{c}=0}, \overline{L}_k^{a_Y,a_D=1,\overline{c}=0} )
     ~,\label{eq:lemma_dismiss_counterfactual_i}\\ 
     \notag\\
    &P(D_{k+1}^{a_Y,a_D=0,\overline{c}=0}=d_{k=1} \mid \overline{Y}_k^{a_Y,a_D=0,\overline{c}=0},\overline{D}_k^{a_Y,a_D=0,\overline{c}=0},\overline{L}_k^{a_Y,a_D=0,\overline{c}=0} ) \notag\\
    =&P(D_{k+1}^{a_Y,a_D=1,\overline{c}=0}=d_{k=1} \mid \overline{Y}_k^{a_Y,a_D=1,\overline{c}=0},\overline{D}_k^{a_Y,a_D=1,\overline{c}=0},\overline{L}_k^{a_Y,a_D=1,\overline{c}=0} ) ~,\label{eq:lemma_dismiss_counterfactual_ii}\\\notag \\
    &P(L_{Y,k+1}^{a_Y,a_D=0,\overline{c}=0}=l_{Y,k+1}\mid \overline{Y}_{k+1}^{a_Y,a_D=0,\overline{c}=0},\overline{D}_{k+1}^{a_Y,a_D=0,\overline{c}=0},\overline{L}_k^{a_Y,a_D=0,\overline{c}=0},L_{D,k+1}^{a_Y,a_D=0,\overline{c}=0}) \notag\\
    =&P(L_{Y,k+1}^{a_Y,a_D=1,\overline{c}=0}=l_{Y,k+1}\mid \overline{Y}_{k+1}^{a_Y,a_D=1,\overline{c}=0},\overline{D}_{k+1}^{a_Y,a_D=1,\overline{c}=0},\overline{L}_k^{a_Y,a_D=1,\overline{c}=0},L_{D,k+1}^{a_Y,a_D=1,\overline{c}=0}) ~,\label{eq:lemma_dismiss_counterfactual_iii}\\ \notag\\
    &P(L_{D,k+1}^{a_Y,a_D=0,\overline{c}=0}=l_{D,k+1}\mid \overline{Y}_{k+1}^{a_Y,a_D=0,\overline{c}=0},\overline{D}_{k+1}^{a_Y,a_D=0,\overline{c}=0},\overline{L}_k^{a_Y,a_D=0,\overline{c}=0}) \notag\\
    =&P(L_{D,k+1}^{a_Y,a_D=1,\overline{c}=0}=l_{D,k+1}\mid \overline{Y}_{k+1}^{a_Y,a_D=1,\overline{c}=0},\overline{D}_{k+1}^{a_Y,a_D=1,\overline{c}=0},\overline{L}_k^{a_Y,a_D=1,\overline{c}=0}) ~,\label{eq:lemma_dismiss_counterfactual_iv}
\end{align}
\end{lemma}
\begin{proof}
We show the equality of Expression (\ref{eq:lemma_dismiss_counterfactual_i}), as (\ref{eq:lemma_dismiss_counterfactual_ii})-(\ref{eq:lemma_dismiss_counterfactual_iv}) follow from analogous arguments, using (\ref{eq:dismissible_component_ii})-(\ref{eq:dismissible_component_iv}) instead of (\ref{eq:dismissible_component_i}).
\begin{align*}
    &P(Y_{k+1}^{a_Y,a_D=0,\overline{c}=0}= y_{k+1}\mid \overline{Y}_k^{a_Y,a_D=0,\overline{c}=0}, \overline{D}_{k+1}^{a_Y,a_D=0,\overline{c}=0}, \overline{L}_k^{a_Y,a_D=0,\overline{c}=0} ) \\
    =&P(Y_{k+1}^{\overline{c}=0}= y_{k+1}\mid \overline{Y}_k^{\overline{c}=0}, \overline{D}_{k+1}^{\overline{c}=0}, \overline{L}_k^{\overline{c}=0}, A_Y=a_Y,A_D=0) \\
    \stackrel{\text{(\ref{eq:dismissible_component_i})}}{=}&P(Y_{k+1}^{\overline{c}=0}= y_{k+1}\mid \overline{Y}_k^{\overline{c}=0}, \overline{D}_{k+1}^{\overline{c}=0}, \overline{L}_k^{\overline{c}=0}, A_Y=a_Y,A_D=1) \\
    &P(Y_{k+1}^{a_Y,a_D=1,\overline{c}=0}= y_{k+1}\mid \overline{Y}_k^{a_Y,a_D=1,\overline{c}=0}, \overline{D}_{k+1}^{a_Y,a_D=1,\overline{c}=0}, \overline{L}_k^{a_Y,a_D=1,\overline{c}=0} ) ~.
\end{align*}
The second and fourth line hold by consistency and by randomization of $A_Y$ and $A_D$ in the four armed trial.
\end{proof}

\begin{lemma} \label{lemma:expand_conditioning_set}
 Suppose the exchangeability and positivity conditions  (\ref{eq:exchangeability_separable_i})-(\ref{eq:exchangeability_separable_ii}) and (\ref{eq:app_positivity_separable_i})-(\ref{eq:app_positivity_separable_iii}) hold. Define $\overline{A}=(A_Y,A_D)$ and $\overline{a}=(a_Y,a_D)$. We then have for all $j\in\{0,\dots,K+1\}$ that
 \begin{align}
    &P(\Delta Y_{j}^{\overline{a},\overline{c}=0}=\Delta y_{j}\mid \overline{D}_{j}^{\overline{a},\overline{c}=0},\overline{C}^{\overline{a},\overline{c}=0}_{j}, \overline{L}_{j-1}^{\overline{a},\overline{c}=0}, \overline{Y}_{j-1}^{\overline{a},\overline{c}=0}, \overline{A}) \notag\\
    &\qquad = P(\Delta Y_{j}^{\overline{a},\overline{c}=0}=\Delta y_{j}\mid \overline{D}_{j}^{\overline{a},\overline{c}=0}, \overline{L}_{j-1}^{\overline{a},\overline{c}=0}, \overline{Y}_{j-1}^{\overline{a},\overline{c}=0} ) ~,\label{eq:app_lemma1_i}\\\notag\\
    &P(D_j^{\overline{a},\overline{c}=0}=d_j\mid C_j^{\overline{a},\overline{c}=0},\overline{L}_{j-1}^{\overline{a},\overline{c}=0}, \overline{Y}_{j-1}^{\overline{a},\overline{c}=0},\overline{D}_{j-1}^{\overline{a},\overline{c}=0},\overline{A}) \notag\\
    &\qquad = P(D_j^{\overline{a},\overline{c}=0}=d_j\mid \overline{L}_{j-1}^{\overline{a},\overline{c}=0}, \overline{Y}_{j-1}^{\overline{a},\overline{c}=0},\overline{D}_{j-1}^{\overline{a},\overline{c}=0}) ~,\label{eq:app_lemma1_ii} \\\notag\\
    &P(L_{Y,j-1}^{\overline{a},\overline{c}=0}=l_{Y,j-1}\mid \overline{L}_{D,j-1}^{\overline{a},\overline{c}=0}=l_{D,j-1}, \overline{Y}_{j-1}^{\overline{a},\overline{c}=0},\overline{D}_{j-1}^{\overline{a},\overline{c}=0},\overline{C}_{j-1}^{\overline{a},\overline{c}=0},\overline{L}_{j-2}^{\overline{a},\overline{c}=0},\overline{A}) \notag\\
    &\qquad =P(L_{Y,j-1}^{\overline{a},\overline{c}=0}=l_{Y,j-1}^{\overline{a},\overline{c}=0}\mid \overline{L}_{D,j-1}^{\overline{a},\overline{c}=0}=l_{D,j-1},\overline{Y}_{j-1}^{\overline{a},\overline{c}=0},\overline{D}_{j-1}^{\overline{a},\overline{c}=0},\overline{L}_{j-2}^{\overline{a},\overline{c}=0}) ~,\label{eq:app_lemma1_iii} \\\notag\\
    &P(L_{D,j-1}^{\overline{a},\overline{c}=0}=l_{D,j-1}\mid \overline{Y}_{j-1}^{\overline{a},\overline{c}=0},\overline{D}_{j-1}^{\overline{a},\overline{c}=0},\overline{C}_{j-1}^{\overline{a},\overline{c}=0},\overline{L}_{j-2}^{\overline{a},\overline{c}=0},\overline{A}) \notag\\
    &\qquad =P(L_{D,j-1}^{\overline{a},\overline{c}=0}=l_{D,j-1}\mid \overline{Y}_{j-1}^{\overline{a},\overline{c}=0},\overline{D}_{j-1}^{\overline{a},\overline{c}=0},\overline{L}_{j-2}^{\overline{a},\overline{c}=0}) ~.\label{eq:app_lemma1_iv}
\end{align}
\end{lemma} 
\begin{proof}
We show the equality for Expression (\ref{eq:app_lemma1_i}), as (\ref{eq:app_lemma1_ii})-(\ref{eq:app_lemma1_iv}) follow from analogous arguments, using (\ref{eq:dismissible_component_ii})-(\ref{eq:dismissible_component_iv}) instead of (\ref{eq:dismissible_component_i}). We have that
\begin{align*}
    &P(\Delta Y_{j}^{\overline{a},\overline{c}=0}=\Delta y_{j}\mid \overline{D}_{j}^{\overline{a},\overline{c}=0}=\overline{d}_j, \overline{L}_{j-1}^{\overline{a},\overline{c}=0}=\overline{l}_{j-1}, \overline{Y}_{j-1}^{\overline{a},\overline{c}=0}=\overline{y}_{j-1}) \\
    =&\frac{P(\Delta \overline{Y}_{j}^{\overline{a},\overline{c}=0}=\Delta \overline{y}_{j}, \overline{D}_{j}^{\overline{a},\overline{c}=0}=\overline{d}_j, \overline{L}_{j-1}^{\overline{a},\overline{c}=0}=\overline{l}_{j-1}  \mid L_0=l_0)}{P(\overline{D}_{j}^{\overline{a},\overline{c}=0}=\overline{d}_j, \overline{L}_{j-1}^{\overline{a},\overline{c}=0}=\overline{l}_{j-1},\Delta \overline{Y}_{j-1}^{\overline{a},\overline{c}=0}=\Delta \overline{y}_{j-1} \mid L_0=l_0)}\\
    =&\frac{P(\Delta \overline{Y}_{j}^{\overline{a},\overline{c}=0}=\Delta \overline{y}_{j}, \overline{D}_{j}^{\overline{a},\overline{c}=0}=\overline{d}_j, \overline{L}_{j-1}^{\overline{a},\overline{c}=0}=\overline{l}_{j-1}  \mid \overline{A}=\overline{a}, L_0=l_0)}{P(\overline{D}_{j}^{\overline{a},\overline{c}=0}=\overline{d}_j, \overline{L}_{j-1}^{\overline{a},\overline{c}=0}=\overline{l}_{j-1},\Delta \overline{Y}_{j-1}^{\overline{a},\overline{c}=0}=\Delta \overline{y}_{j-1} \mid \overline{A}=\overline{a}, L_0=l_0)}\\
    =&\frac{P(\Delta \overline{Y}_{j}^{\overline{a},\overline{c}=0}=\Delta \overline{y}_{j}, \overline{D}_{j}^{\overline{a},\overline{c}=0}=\overline{d}_j, \overline{L}_{j-1}^{\overline{a},\overline{c}=0}=\overline{l}_{j-1}  \mid C_0^{\overline{a}}=0, \overline{A}=\overline{a},L_0=l_0)}{P(\overline{D}_{j}^{\overline{a},\overline{c}=0}=\overline{d}_j, \overline{L}_{j-1}^{\overline{a},\overline{c}=0}=\overline{l}_{j-1},\Delta \overline{Y}_{j-1}^{\overline{a},\overline{c}=0}=\Delta \overline{y}_{j-1} \mid C_0^{\overline{a}}=0, \overline{A}=\overline{a},L_0=l_0)} ~,
\end{align*}
where we have used Bayes' law in the first line,  (\ref{eq:exchangeability_separable_i}) and (\ref{eq:app_positivity_separable_i}) in the second line (expression (\ref{eq:app_positivity_separable_i}) ensures that the conditioning sets have non-zero probability) and the fact that all individuals are uncensored at time $k=0$ in the third line. Next, using Bayes' law again, we have that the above is equal to
\begin{align*}
    \scriptstyle\frac{P(\Delta \overline{Y}_{j}^{\overline{a},\overline{c}=0}=\Delta \overline{y}_{j}, \overline{D}_{j}^{\overline{a},\overline{c}=0}=\overline{d}_j, \overline{L}_{j-1}^{\overline{a},\overline{c}=0}=\overline{l}_{j-1}  \mid  Y_0^{\overline{a},\overline{c}=0}=y_0, D_0^{\overline{a},\overline{c}=0}=d_0,  C_0^{\overline{a}}=0, \overline{A}=\overline{a},L_0=l_0)}{P(\overline{D}_{j}^{\overline{a},\overline{c}=0}=\overline{d}_j, \overline{L}_{j-1}^{\overline{a},\overline{c}=0}=\overline{l}_{j-1},\Delta \overline{Y}_{j-1}^{\overline{a},\overline{c}=0}=\Delta \overline{y}_{j-1} \mid Y_0^{\overline{a},\overline{c}=0}=y_0,D_0^{\overline{a},\overline{c}=0}=d_0,  C_0^{\overline{a}}=0, \overline{A}=\overline{a},L_0=l_0)} ~.
\end{align*}
Using (\ref{eq:exchangeability_separable_ii}), the above is equal to
\begin{align*}
    \scriptstyle\frac{P(\Delta \overline{Y}_{j}^{\overline{a},\overline{c}=0}=\Delta \overline{y}_{j}, \overline{D}_{j}^{\overline{a},\overline{c}=0}=\overline{d}_j, \overline{L}_{j-1}^{\overline{a},\overline{c}=0}=\overline{l}_{j-1}  \mid  \overline{C}_1^{\overline{a},\overline{c}=0}=0, L_0=l_0, Y_0^{\overline{a},\overline{c}=0}=y_0,D_0^{\overline{a},\overline{c}=0}=d_0, \overline{A}=\overline{a})}{P(\overline{D}_{j}^{\overline{a},\overline{c}=0}=\overline{d}_j, \overline{L}_{j-1}^{\overline{a},\overline{c}=0}=\overline{l}_{j-1},\Delta \overline{Y}_{j-1}^{\overline{a},\overline{c}=0}=\Delta \overline{y}_{j-1} \mid   \overline{C}_1^{\overline{a},\overline{c}=0}=0, L_0=l_0, Y_0^{\overline{a},\overline{c}=0}=y_0,D_0^{\overline{a},\overline{c}=0}=d_0, \overline{A}=\overline{a}) 
} ~.
\end{align*}
The conditioning set has non-zero probability by positivity (\ref{eq:app_positivity_separable_iii}). After iterating this procedure, we obtain
\begin{align*}
    \scriptstyle\frac{P(\Delta \overline{Y}_{j}^{\overline{a},\overline{c}=0}=\Delta \overline{y}_{j}, \overline{D}_{j}^{\overline{a},\overline{c}=0}=\overline{d}_j  \mid C_j^{\overline{a},\overline{c}=0}=0, \overline{L}_{j-1}^{\overline{a},\overline{c}=0}=l_{j-1}, \overline{Y}_{j-1}^{\overline{a},\overline{c}=0}=\overline{y}_{j-1},\overline{D}_{j-1}^{\overline{a},\overline{c}=0}=\overline{d}_{j-1},   \overline{A}=\overline{a}
)}{P(\overline{D}_{j}^{\overline{a},\overline{c}=0}=\overline{d}_j \mid C_j^{\overline{a},\overline{c}=0}=0, \overline{L}_{j-1}^{\overline{a},\overline{c}=0}=l_{j-1}, \overline{Y}_{j-1}^{\overline{a},\overline{c}=0}=\overline{y}_{j-1},\overline{D}_{j-1}^{\overline{a},\overline{c}=0}=\overline{d}_{j-1}, \overline{A}=\overline{a})} ~.
\end{align*}
Finally, using Bayes law again, the above is equal to
\begin{align*}
    P(\Delta Y_{j}^{\overline{a},\overline{c}=0}=\Delta y_{j}\mid \overline{D}_{j}^{\overline{a},\overline{c}=0}=\overline{d}_j,
    \overline{C}_{j}^{\overline{a},\overline{c}=0}=0,
    \overline{L}_{j-1}^{\overline{a},\overline{c}=0}=\overline{l}_{j-1}, \overline{Y}_{j-1}^{\overline{a},\overline{c}=0}=\overline{y}_{j-1}, \overline{A}=\overline{a}) ~.
\end{align*}
The final result follows because the above equality holds for any choice of $\Delta \overline{y}_j,\overline{d}_j,\overline{c}_j,\overline{l}_{j-1}, \overline{a}$.
\end{proof}

\begin{lemma} \label{lemma:dismiss_observed_component}
Suppose that the identification conditions for separable effects (\ref{eq:exchangeability_separable_i})-(\ref{eq:exchangeability_separable_ii}), (\ref{eq:dismissible_component_i})-(\ref{eq:dismissible_component_iv}), (\ref{eq:app_positivity_separable_i})-(\ref{eq:app_consistency_separable}) and the modified treatment assumption (\ref{eq:modified_treatment_assumption}) hold. We then have that
\begin{align}
   & P(\Delta Y_j=\Delta y_j\mid \overline{D}_{j}=\overline{d}_j, \overline{C}_j=0,\overline{L}_{j-1}=\overline{l}_{j-1},  \overline{Y}_{j-1}=  \overline{y}_{j-1},A_Y=a_Y,A_D=a_D) \notag\\
    &\qquad =P(  \Delta Y_j=  \Delta y_j\mid \overline{D}_{j}=\overline{d}_j, \overline{C}_j=0,\overline{L}_{j-1}=\overline{l}_{j-1},  \overline{Y}_{j-1}= \overline{y}_{j-1},A=a_Y) \label{eq:app_lemma2_i}~, \\\notag\\
   & P(D_{j}=d_{j}\mid \overline{C}_{j}=0, \overline{L}_{j-1}=\overline{l}_{j-1},  \overline{Y}_{j-1}= \overline{y}_{j-1}, \overline{D}_{j-1}=\overline{d}_{j-1}, A_Y=a_Y,A_D=a_D ) \notag\\
   &\qquad =P(D_{j}=d_{j}\mid \overline{C}_{j}=0, \overline{L}_{j-1}=\overline{l}_{j-1},  \overline{Y}_{j-1}= \overline{y}_{j-1}, \overline{D}_{j-1}=\overline{d}_{j-1},A=a_D )~, \label{eq:app_lemma2_ii}\\\notag\\
    &P(L_{Y,j-1}=l_{Y,j-1}\mid  \overline{Y}_{j-1}= \overline{y}_{j-1}, \overline{D}_{j-1}=\overline{d}_{j-1},\overline{C}_{j-1}=0,\overline{L}_{j-2}=\overline{l}_{j-2}, \notag\\
   & \qquad\qquad\qquad\qquad\qquad\qquad L_{D,j-1}=l_{D,j-1}, A_Y=a_Y,A_D=a_D) \notag\\
   &\qquad =P(L_{Y,j-1}=l_{Y,j-1}\mid  \overline{Y}_{j-1}= \overline{y}_{j-1}, \overline{D}_{j-1}=\overline{d}_{j-1},\overline{C}_{j-1}=0,\overline{L}_{j-2}=\overline{l}_{j-2}, \notag\\
   & \qquad \qquad\qquad\qquad\qquad\qquad\qquad L_{D,j-1}=l_{D,j-1}, A=a_Y) ~,\label{eq:app_lemma2_iii}\\\notag\\ 
    &P(L_{D,j-1}=l_{D,j-1}\mid  \overline{Y}_{j-1}= \overline{y}_{j-1}, \overline{D}_{j-1}=\overline{d}_{j-1},\overline{C}_{j-1}=0,\overline{L}_{j-2}=\overline{l}_{j-2}, \notag\\
    & \qquad\qquad\qquad\qquad\qquad\qquad A_Y=a_Y,A_D=a_D) \notag\\
    &\qquad =P(L_{D,j-1}=l_{D,j-1}\mid  \overline{Y}_{j-1}= \overline{y}_{j-1}, \overline{D}_{j-1}=\overline{d}_{j-1},\overline{C}_{j-1}=0,\overline{L}_{j-2}=\overline{l}_{j-2}, \notag\\
    & \qquad\qquad\qquad\qquad\qquad\qquad\qquad A=a_D) ~.\label{eq:app_lemma2_iv}
\end{align}
\end{lemma} 
The quantities on the LHS are identified in the four armed trial, whereas the quantities on the RHS are identified in the two armed trial.
\begin{proof}
We show the equality for (\ref{eq:app_lemma2_i}), as (\ref{eq:app_lemma2_ii})-(\ref{eq:app_lemma2_iv}) follow from analogous arguments using (\ref{eq:dismissible_component_ii})-(\ref{eq:dismissible_component_iv}) instead of (\ref{eq:dismissible_component_i}). We have that
\begin{align*}
    &\quad P(\Delta Y_j=\Delta y_j\mid \overline{D}_{j}=\overline{d}_j, \overline{C}_j=0,\overline{L}_{j-1}=\overline{l}_{j-1},  \overline{Y}_{j-1}=  \overline{y}_{j-1},A_Y=a_Y,A_D=a_D)\\
    &\stackrel{\text{(\ref{eq:app_consistency_separable})}}{=} P(\Delta Y_j^{a_Y,a_D,\overline{c}=0}=\Delta y_j\mid \overline{D}_{j}^{a_Y,a_D,\overline{c}=0}=\overline{d}_j, \overline{C}_j^{a_Y,a_D,\overline{c}=0}=0,\overline{L}_{j-1}^{a_Y,a_D,\overline{c}=0} =\overline{l}_{j-1}, \notag\\
    &\qquad\qquad\qquad\qquad\qquad\qquad\qquad \overline{Y}_{j-1}^{a_Y,a_D,\overline{c}=0}=  \overline{y}_{j-1},A_Y=a_Y,A_D=a_D) \\
    &\stackrel{\text{Lemma~\ref{lemma:expand_conditioning_set}}}{=} P(\Delta Y_j^{a_Y,a_D,\overline{c}=0}=\Delta y_j\mid \overline{D}_{j}^{a_Y,a_D,\overline{c}=0}=\overline{d}_j, \overline{L}_{j-1}^{a_Y,a_D,\overline{c}=0} =\overline{l}_{j-1}, \overline{Y}_{j-1}^{a_Y,a_D,\overline{c}=0}=  \overline{y}_{j-1}) \\
    &\stackrel{\text{Lemma~\ref{lemma:dismiss_counterfactual_component}}}{=} P(\Delta Y_j^{\overline{a}=(a_Y,a_Y),\overline{c}=0}=\Delta y_j\mid \overline{D}_{j}^{\overline{a}=(a_Y,a_Y),\overline{c}=0}=\overline{d}_j, \overline{L}_{j-1}^{\overline{a}=(a_Y,a_Y),\overline{c}=0} =\overline{l}_{j-1}, \\
    &\qquad\qquad\qquad\qquad\qquad\qquad\qquad\qquad \overline{Y}_{j-1}^{\overline{a}=(a_Y,a_Y),\overline{c}=0}=  \overline{y}_{j-1}) \\
    &\stackrel{\text{Lemma~\ref{lemma:expand_conditioning_set}}}{=}
    P(\Delta Y_j^{\overline{a}=(a_Y,a_Y),\overline{c}=0}=\Delta y_j\mid \overline{D}_{j}^{\overline{a}=(a_Y,a_Y),\overline{c}=0}=\overline{d}_j, \overline{C}_j^{\overline{a}=(a_Y,a_Y),\overline{c}=0}=0, \notag\\
    &\qquad\qquad\qquad\qquad\qquad\qquad\qquad\qquad \overline{L}_{j-1}^{\overline{a}=(a_Y,a_Y),\overline{c}=0} =\overline{l}_{j-1}, \overline{Y}_{j-1}^{\overline{a}=(a_Y,a_Y),\overline{c}=0}=  \overline{y}_{j-1}, \\
   &\qquad\qquad\qquad\qquad\qquad\qquad\qquad\qquad\qquad A_Y=a_Y,A_D=a_Y) \\
    &\stackrel{\text{(\ref{eq:modified_treatment_assumption}),(\ref{eq:app_positivity_separable_ii})}}{=} P(\Delta Y_j^{a=a_Y,\overline{c}=0}=\Delta y_j\mid \overline{D}_{j}^{a=a_Y,\overline{c}=0}=\overline{d}_j, \overline{C}_j^{a=a_Y,\overline{c}=0}=0,\overline{L}_{j-1}^{a=a_Y,\overline{c}=0} =\overline{l}_{j-1}, \notag\\
    &\qquad\qquad\qquad\qquad\qquad\qquad\quad  \overline{Y}_{j-1}^{a=a_Y,\overline{c}=0}= \overline{y}_{j-1},A=a_Y) \\
    &\stackrel{\text{(\ref{eq:app_consistency_separable})}}{=} P(\Delta Y_j=\Delta y_j\mid \overline{D}_{j}=\overline{d}_j, \overline{C}_j=0, \overline{L}_{j-1} =\overline{l}_{j-1}, \overline{Y}_{j-1}=\overline{y}_{j-1},A=a_Y) ~.
\end{align*}
\end{proof}
 
To derive the identification formula for separable effects, we proceed by sequential application of Bayes' theorem
\begin{align*}
    &P(\Delta \overline{Y}_k^{a_Y,a_D,\overline{c}=0}=\Delta \overline{y}_k, \overline{D}_k^{a_Y,a_D,\overline{c}=0}=\overline{d}_k,\overline{L}^{a_Y,a_D,\overline{c}=0}_k=\overline{l}_k ) \notag\\
    =& \prod_{j=0}^k P( \Delta Y^{a_Y,a_D,\overline{c}=0}_j=\Delta y_j\mid \overline{D}^{a_Y,a_D,\overline{c}=0}_{j}=\overline{d}_j, \overline{L}^{a_Y,a_D,\overline{c}=0}_{j-1}=\overline{l}_{j-1}, \overline{Y}^{a_Y,a_D,\overline{c}=0}_{j-1}=  \overline{y}_{j-1}) \notag\\
    &\quad \times P(D_{j}^{a_Y,a_D,\overline{c}=0}=d_{j}\mid  \overline{L}_{j-1}^{a_Y,a_D,\overline{c}=0}=\overline{l}_{j-1},  \overline{Y}_{j-1}^{a_Y,a_D,\overline{c}=0}= \overline{y}_{j-1}, \overline{D}_{j-1}^{a_Y,a_D,\overline{c}=0}=\overline{d}_{j-1} ) \notag\\
    &\quad \times P(L_{Y,j-1}^{a_Y,a_D,\overline{c}=0}=l_{Y,j-1}\mid  \overline{Y}_{j-1}^{a_Y,a_D,\overline{c}=0}= \overline{y}_{j-1}, \overline{D}_{j-1}^{a_Y,a_D,\overline{c}=0}=\overline{d}_{j-1},\overline{L}_{j-2}^{a_Y,a_D,\overline{c}=0}=\overline{l}_{j-2}, \notag\\
   & \qquad\qquad\qquad\qquad\qquad\qquad L_{D,j-1}^{a_Y,a_D,\overline{c}=0}=l_{D,j-1}) \notag\\
    &\quad \times P(L_{D,j-1}^{a_Y,a_D,\overline{c}=0}=l_{D,j-1}\mid  \overline{Y}_{j-1}^{a_Y,a_D,\overline{c}=0}= \overline{y}_{j-1}, \overline{D}_{j-1}^{a_Y,a_D,\overline{c}=0}=\overline{d}_{j-1},\overline{L}^{a_Y,a_D,\overline{c}=0}_{j-2}=\overline{l}_{j-2}) ~.
    \end{align*}
Using Lemma~\ref{lemma:expand_conditioning_set} and (\ref{eq:app_consistency_separable}), the above is equal to
    \begin{align}
     &\prod_{j=0}^k P(\Delta Y_j=\Delta y_j\mid \overline{D}_{j}=\overline{d}_j, \overline{C}_j=0,\overline{L}_{j-1}=\overline{l}_{j-1},  \overline{Y}_{j-1}=  \overline{y}_{j-1},A_Y=a_Y,A_D=a_D) \notag\\
    &\quad \times P(D_{j}=d_{j}\mid \overline{C}_{j}=0, \overline{L}_{j-1}=\overline{l}_{j-1},  \overline{Y}_{j-1}= \overline{y}_{j-1}, \overline{D}_{j-1}=\overline{d}_{j-1}, A_Y=a_Y,A_D=a_D ) \notag\\
    &\quad \times P(L_{Y,j-1}=l_{Y,j-1}\mid  \overline{Y}_{j-1}= \overline{y}_{j-1}, \overline{D}_{j-1}=\overline{d}_{j-1},\overline{C}_{j-1}=0,\overline{L}_{j-2}=\overline{l}_{j-2}, \notag\\
   & \qquad\qquad\qquad\qquad\qquad\qquad \overline{L}_{A_D,j-1}=\overline{l}_{A_D,j-1}, A_Y=a_Y,A_D=a_D) \notag\\
    &\quad \times P(L_{D,j-1}=l_{D,j-1}\mid  \overline{Y}_{j-1}= \overline{y}_{j-1}, \overline{D}_{j-1}=\overline{d}_{j-1},\overline{C}_{j-1}=0,\overline{L}_{j-2}=\overline{l}_{j-2}, \notag\\
    & \qquad\qquad\qquad\qquad\qquad\qquad A_Y=a_Y,A_D=a_D) ~.
    \label{eq:app_g-formula_4_levels}
\end{align} 
 
The quantities on RHS of (\ref{eq:app_g-formula_4_levels}) are identified in the four armed trial.  The final identification formula for separable effects, which is a function of observed quantities in the two armed trial, follows directly from application of Lemma~\ref{lemma:dismiss_observed_component}, which gives
\begin{align}
    &P(\Delta \overline{Y}_k^{a_Y,a_D,\overline{c}=0}=\Delta \overline{y}_k, \overline{D}_k^{a_Y,a_D,\overline{c}=0}=\overline{d}_k,\overline{L}^{a_Y,a_D,\overline{c}=0}_{k-1}=\overline{l}_{k-1} )  \notag\\
    =&\prod_{j=0}^{k} P(\Delta Y_j=\Delta y_j\mid \overline{D}_{j}=\overline{d}_j, \overline{C}_j=0,\overline{L}_{j-1}=\overline{l}_{j-1},  \overline{Y}_{j-1}=  \overline{y}_{j-1},A=a_Y) \notag\\
    &\quad \times P(D_{j}=d_{j}\mid C_{j}=0, \overline{L}_{j-1}=\overline{l}_{j-1},  \overline{Y}_{j-1}= \overline{y}_{j-1}, \overline{D}_{j-1}=\overline{d}_{j-1}, A=a_D ) \notag\\
    &\quad \times P(L_{Y,j-1}=l_{Y,j-1}\mid  \overline{L}_{A_D,j-1}=\overline{l}_{A_D,j-1}, \overline{Y}_{j-1}= \overline{y}_{j-1}, \overline{D}_{j-1}=\overline{d}_{j-1},\overline{C}_{j-1}=0, \notag\\
   & \qquad\qquad\qquad\qquad\qquad\qquad  \overline{L}_{j-2}=\overline{l}_{j-2}, A=a_Y) \notag\\
    &\quad \times P(L_{D,j-1}=l_{D,j-1}\mid  \overline{Y}_{j-1}= \overline{y}_{j-1}, \overline{D}_{j-1}=\overline{d}_{j-1},\overline{C}_{j-1}=0,\overline{L}_{j-2}=\overline{l}_{j-2},  A=a_D)~. \notag\\ \label{eq:app_g-formula_separable}
\end{align}

\subsubsection{IPW representation}
Next, we will show that
\begin{align}
    E[\Delta Y^{a_Y,a_D,\overline{c}=0}_i] = E\left[ \frac{I(A=a_Y)}{\pi_A(A)}\cdot\frac{I(C_i=0)}{\prod_{j=0}^i \pi_{C_j}(C_j)}\cdot\frac{\prod_{j=0}^i \pi_{D_j}^{a_D}}{\prod_{j=0}^i \pi_{D_j}^{a_Y}}\cdot\frac{\prod_{j=0}^{i-1} \pi_{L_{D,j}}^{a_D}}{\prod_{j=0}^{i-1} \pi_{L_{D,j}}^{a_Y}} \cdot  \Delta Y_i\right] ~. \label{eq:app_IPW_separable}
\end{align}
To begin, we use Bayes' theorem sequentially to write out the joint density
\begin{align}
    &P(A=a,\overline{C}_i=\overline{c}_i,\Delta \overline{Y}_i=\Delta\overline{y}_i,\overline{D}_i=\overline{d}_i,\overline{L}_{i-1}=\overline{l}_{i-1}) = \notag\\ 
    &\prod_{j=0}^{i} P(\Delta Y_j=\Delta y_j\mid \overline{D}_{j}=\overline{d}_j, \overline{C}_j=\overline{c}_j,\overline{L}_{j-1}=\overline{l}_{j-1},  \overline{Y}_{j-1}=  \overline{y}_{j-1},A=a) \notag\\
    &\quad \times P(D_{j}=d_{j}\mid \overline{C}_{j}=\overline{c}_j, \overline{L}_{j-1}=\overline{l}_{j-1},  \overline{Y}_{j-1}= \overline{y}_{j-1}, \overline{D}_{j-1}=\overline{d}_{j-1}, A=a ) \notag\\
    &\quad \times P(C_{j}=c_{j}\mid \overline{L}_{j-1}=\overline{l}_{j-1},  \overline{Y}_{j-1}= \overline{y}_{j-1}, \overline{D}_{j-1}=\overline{d}_{j-1},\overline{C}_{j-1}=\overline{c}_{j-1}, A=a ) \notag\\
    &\quad \times P(L_{Y,j-1}=l_{Y,j-1}\mid  \overline{Y}_{j-1}= \overline{y}_{j-1}, \overline{D}_{j-1}=\overline{d}_{j-1},\overline{C}_{j-1}=\overline{c}_{j-1},\overline{L}_{j-2}=\overline{l}_{j-2}, \notag\\
   & \qquad\qquad\qquad\qquad\qquad\qquad \overline{L}_{A_D,j-1}=\overline{l}_{A_D,j-1}, A=a) \notag\\
    &\quad \times P(L_{D,j-1}=l_{D,j-1}\mid  \overline{Y}_{j-1}= \overline{y}_{j-1}, \overline{D}_{j-1}=\overline{d}_{j-1},\overline{C}_{j-1}=\overline{c}_{j-1},\overline{L}_{j-2}=\overline{l}_{j-2},  A=a) \notag\\
    &\quad \times P(A=a\mid L_0=l_0) ~. \notag\\\label{eq:app_bayes_joint}
\end{align}
Writing out RHS of (\ref{eq:app_IPW_separable}) as a discrete sum over the density in (\ref{eq:app_bayes_joint}), we have that RHS of (\ref{eq:app_IPW_separable}) is equal to
\begin{align*}
    &\sum_a\sum_{\overline{c}_i}\sum_{\Delta\overline{y}_i}\sum_{\overline{d}_i}\sum_{\overline{l}_{i-1}} P(A=a,\overline{C}_i=\overline{c}_i,\Delta \overline{Y}_i=\Delta\overline{y}_i,\overline{D}_i=\overline{d}_i,\overline{L}_{i-1}=\overline{l}_{i-1})\cdot \Delta y_i \notag\\
    &\quad \times \frac{I(a=a_Y)}{P(A=a_Y\mid L_0=l_0)} \\
    &\quad\times \frac{I(\overline{c}_i=0)}{\prod_{j=0}^i P(C_j=0\mid C_{j-1}=0,\overline{D}_{j-1}=\overline{d}_{j-1},\overline{L}_{j-1}=\overline{l}_{j-1}, \overline{Y}_{j-1}=\overline{y}_{j-1},A=a_Y)} \\
    & \quad\times \frac{\prod_{j=0}^i P(D_j= d_j \mid \overline{C}_j=0,\overline{L}_{j-1}=\overline{l}_{j-1},\overline{Y}_{j-1}=\overline{y}_{j-1},\overline{D}_{j-1}=\overline{d}_{j-1},A=a_D)}{\prod_{j=0}^i P(D_j= d_j \mid \overline{C}_j=0,\overline{L}_{j-1}=\overline{l}_{j-1},\overline{Y}_{j-1}=\overline{y}_{j-1},\overline{D}_{j-1}=\overline{d}_{j-1},A=a_Y)} \\
    &\quad\times \frac{\prod_{j=0}^{i-1} P(L_{D,j}=l_{D,j}\mid \overline{L}_{j-1}=\overline{l}_{j-1},\overline{Y}_{j}=\overline{y}_j,\overline{D}_{j}=\overline{d}_j,\overline{C}_{j}=\overline{c}_j,A=a_D)}{\prod_{j=0}^{i-1} P(L_{D,j}=l_{D,j}\mid \overline{L}_{j-1}=\overline{l}_{j-1},\overline{Y}_{j}=\overline{y}_j,\overline{D}_{j}=\overline{d}_j,\overline{C}_{j}=\overline{c}_j,A=a_Y)} \\
    \stackrel{\text{(\ref{eq:app_g-formula_separable})}}{=}& \sum_{\Delta\overline{y}_i}\sum_{\overline{d}_i}\sum_{\overline{l}_{i-1}} P(\Delta \overline{Y}_i^{a_Y,a_D,\overline{c}=0}=\Delta \overline{y}_i, \overline{D}_i^{a_Y,a_D,\overline{c}=0}=\overline{d}_i,\overline{L}^{a_Y,a_D,\overline{c}=0}_{i-1}=\overline{l}_{i-1} )  \cdot \Delta y_i \\
    =&E[\Delta Y_i^{a_Y,a_D,\overline{c}=0}] ~.
\end{align*}

\subsubsection{Limit of fine discretizations}
To proceed, we define the weights
\begin{align}
    W_{L_{D},i}(a_Y,a_D)=\frac{\prod_{j=0}^i \pi_{L_{D,j}}^{a_D}}{\prod_{k=0}^i \pi_{L_{D,k}}^{a_Y}}, \label{eq: separable LAD discrete weights}
\end{align}
and 
\begin{align}
    W_{D,i}(a_Y,a_D) &= \frac{\prod_{j=0}^{i}[1-\Delta \Uplambda_j^{D\mid\mathcal{F}}(a_D)]^{1-D_i}}{\prod_{j=0}^i [1-\Delta \Uplambda_j^{D\mid\mathcal{F}}(a_Y)]^{1-D_i}} \bigg(\frac{\Delta \Uplambda_j^{D\mid\mathcal{F}}(a_D)}{\Delta \Uplambda_j^{D\mid\mathcal{F}}(a_Y)} \bigg)^{D_i} ~. \label{eq: separable D discrete weights}
\end{align}

Using the laws of probability, we may write (\ref{eq:app_IPW_separable}) as
\begin{align}
   &E[Y_k^{a_Y,a_D,\overline{c}=0}] =\notag\\
   &\sum_{i=0}^k\prod_{j=0}^{i-1}[1-\Delta A_j^D(a_Y)] \notag\\
   &\quad \times E\left[W_A  W_{C,i} W_{D,i}(a_Y,a_D)W_{L_{D},i-1}(a_Y,a_D)  \Delta Y_i ~\bigg|~ C_i=0,D_{i-1}=0,A=a_Y\right] ~. \label{eq:app_separable_weighted_form} 
\end{align}

Using analogous arguments where we start from (\ref{eq:app_bayes_joint}) and replace $\Delta Y_k$ by $\Delta D_k$, we have that $E[D_k^{a_Y,a_D,\overline{c}=0}]$ is identified by
\begin{align}
    &E[D_k^{a_Y,a_D,\overline{c}=0}]= \notag\\
    &\sum_{i=0}^k\prod_{j=0}^{i-1}[1-\Delta A_j^D(a_Y)] \notag\\
    &\quad \times E\left[W_A  W_{C,i} W_{D,i}(a_Y,a_D)W_{L_{D},i-1}(a_Y,a_D)  \Delta D_i ~\bigg|~ C_i=0,D_{i-1}=0,A=a_Y\right] ~. \label{eq:survival_separable_weighted_form}
\end{align}

As expected, we recover the identification formulas for total effect when choosing $a_Y=a_D=a$ in (\ref{eq:app_separable_weighted_form}) and (\ref{eq:survival_separable_weighted_form}).

\subsubsection{An alternative IPW representation \label{sec:alternative_IPW_separable}}
Through analogous arguments that we used to derive (\ref{eq:app_IPW_separable}), it follows that
\begin{align*}
    E[\Delta Y^{a_Y,a_D,\overline{c}=0}_i] = E\left[ \frac{I(A=a_D)}{\pi_A(A)}\cdot\frac{I(C_i=0)}{\prod_{j=0}^i \pi_{C_j}(C_j)}\cdot\frac{\prod_{j=0}^i \pi_{Y_j}^{a_Y}}{\prod_{j=0}^i \pi_{Y_j}^{a_D}}\cdot\frac{\prod_{j=0}^{i-1} \pi_{L_{Y,j}}^{a_Y}}{\prod_{j=0}^{i-1} \pi_{L_{Y,j}}^{a_D}} \cdot  \Delta Y_i\right] ~.
\end{align*}
Next, we define the weights
\begin{align*}
    W_{L_{Y},i}(a_Y,a_D)=\frac{\prod_{j=0}^i \pi_{L_{Y,j}}^{a_Y}}{\prod_{k=0}^i \pi_{L_{Y,k}}^{a_D}} ~.
\end{align*}
By close analogy with the argument found in Appendix D of \citet{stensrud_generalized_2021}, we can re-express the covariate weights as
\begin{align*}
    W_{L_{D},i}(a_Y,a_D) = \begin{split} &\frac{\prod_{j=0}^i P(A=a_D\mid L_{D,j},\overline{L}_{j-1},\overline{Y}_j,\overline{D}_j,\overline{C}_j)}{\prod_{j=0}^i P(A=a_Y\mid L_{D,j},\overline{L}_{j-1},\overline{Y}_j,\overline{D}_j,\overline{C}_j)} \\
    &\quad\times \frac{\prod_{j=0}^i P(A=a_Y\mid \overline{L}_{j-1},\overline{Y}_j,\overline{D}_j,\overline{C}_j)}{\prod_{j=0}^i P(A=a_D\mid \overline{L}_{j-1},\overline{Y}_j,\overline{D}_j,\overline{C}_j)} ~, 
    \end{split} 
\end{align*}
and
\begin{align*}
    W_{L_{Y},i}(a_Y,a_D) = \begin{split} &\frac{\prod_{j=0}^i P(A=a_Y\mid \overline{L}_j,\overline{Y}_j,\overline{D}_j,\overline{C}_j)}{\prod_{j=0}^i P(A=a_D\mid \overline{L}_j,\overline{Y}_j,\overline{D}_j,\overline{C}_j)} \\
    &\quad\times \frac{\prod_{j=0}^i P(A=a_D\mid L_{D,j},\overline{L}_{j-1},\overline{Y}_j,\overline{D}_j,\overline{C}_j)}{\prod_{j=0}^i P(A=a_Y\mid L_{D,j},\overline{L}_{j-1},\overline{Y}_j,\overline{D}_j,\overline{C}_j)} ~. 
    \end{split} 
\end{align*}
Thus, when the dismissible component conditions holds with $L_{Y,k}=L_k$ and $L_{D,k}=\emptyset$, it follows that $W_{L_{D},k}(a_Y,a_D)=1$. Conversely, when the dismissible component conditions hold with $L_{D,k}=L_k$ and $L_{Y,k}=\emptyset$, it follows that $W_{L_{Y},k}(a_Y,a_D)=1$.

\section{Correspondence of the independent censoring assumption\label{sec:app_correspondence_censoring}}
We begin by introducing the notion of faithfulness, following~\citet{spirtes_causation_2000}.
\begin{definition}\label{def:faithfulness}
A law $P$ is faithful to a causal directed acyclic graph $\mathcal{G}$ if for any disjoint set of nodes $A,B,C$ we have that $A\independent C\mid B$ under $P$ implies $(A\independent C\mid B)_{\mathcal{G}}$, where $(\bullet)_\mathcal{G}$ is used to denote graphical d-separation.
\end{definition}

In the following result, we establish a correspondence between the exchangeability assumption and the classical independent censoring assumption in event history analysis.

\begin{proposition}\label{prop:equivalent_independent_censoring}
Let the factual data in Sec.~\ref{sec:observed_data} be generated by an FFRCISTG model, and assume consistency (\ref{eq:consistency_total_effect}) and faithfulness (Definition~\ref{def:faithfulness}) hold. Then, (\ref{eq:discrete_independent_censoring}) implies exchangeability with respect to censoring (\ref{eq:exchangeability_total_ii}).

\end{proposition}

\begin{proof}[Proof]
Expression (\ref{eq:discrete_independent_censoring}) is equivalent to the statement
\begin{align}
    \Delta Y_{j} \independent \overline{C}_j \mid \overline{D}_{j}, \overline{L}_{j-1}, \overline{Y}_{j-1}, A \quad\text{for}\quad j\in\{1,\dots,K+1\} ~. \label{eq:discrete_independence}
\end{align}
Under faithfulness, a violation of (\ref{eq:discrete_independence}) is equivalent to the existence of one of the three paths
\begin{enumerate}
    \item $\overline{C}_k \leftarrow U \rightarrow Y_k$
    \item $\overline{C}_k \leftarrow U_1 \rightarrow X \leftarrow U_2 \rightarrow Y_k$
    \item $\overline{C}_k\rightarrow Y_k$
\end{enumerate}
for some $k\in\{1,\dots,K+1\}$, where $X\in\{\overline{L}_k,\overline{Y}_{k-1},\overline{D}_k,A\}$.
Likewise, the violation of \eqref{eq:exchangeability_total_ii} is equivalent to the existence of one of the paths 
\begin{enumerate}[(1')]
    \item $\overline{C}_k^{a,\overline{c}=0} \leftarrow U \rightarrow Y_k^{a,\overline{c}=0}$
    \item $\overline{C}_k^{a,\overline{c}=0} \leftarrow U_1 \rightarrow X^\ast \leftarrow U_2 \rightarrow Y_k^{a,\overline{c}=0}$
\end{enumerate}
for some $k\in\{1,\dots,K+1\}$, where $X^\ast\in\{\overline{L}_{k-1}^{a,\overline{c}=0},\overline{Y}_{k-1}^{a,\overline{c}=0},\overline{D}_k^{a,\overline{c}=0},A\}$.

By the properties of transforming a DAG into a SWIG~\citep{richardson_primer_2013}, and by consistency (\ref{eq:consistency_total_effect}), the existence of (1') implies the existence of (1), and the existence of (2') implies the existence of (2).\footnote{The reverse implications do not hold. A counterexample is given by partial exchangeability~\citep{sarvet_graphical_2020}.} It follows that violation of (\ref{eq:exchangeability_total_ii}) implies violation of (\ref{eq:discrete_independence}), and consequently of (\ref{eq:discrete_independent_censoring}).
\end{proof}

An analogous relation exists between identification conditions for the competing event. The classical independent censoring assumption for the competing event takes the form
\begin{align}
    \lambda_{D,t}^{\mathcal{F}^c} = \lambda_{D,t}^\mathcal{G} ~, \label{eq:survival_continuous_independent_censoring}
\end{align}
where
\begin{align*}
    \lambda_{D,t}^{\mathcal{F}^c}dt &= E[dN_t^D\mid \mathcal{F}^c_{t^-}]~,  \\
    \lambda_{D,t}^{\mathcal{G}}dt &= E[dN_t^D\mid \mathcal{G}_{t^-}] ~.
\end{align*}
A corresponding relation to (\ref{eq:survival_continuous_independent_censoring}) in discrete time is
\begin{align*}
    \frac{1}{\Delta t} \cdot 
    E[\Delta D_{j} \mid  \overline{L}_{j-1}, \overline{Y}_{j-1}, \overline{D}_{j-1}, A] = 
    \frac{1}{\Delta t} \cdot 
    E[\Delta D_{j} \mid \overline{C}_{j}, \overline{L}_{j-1}, \overline{Y}_{j-1},\overline{D}_{j-1}, A] ~. 
\end{align*}
Since $\Delta D_j\in\{0,1\}$, this can be written as
\begin{align}
    \frac{1}{\Delta t} \cdot 
    P(\Delta D_{j}=1 \mid  \overline{L}_{j-1}, \overline{Y}_{j-1}, \overline{D}_{j-1}, A) = 
    \frac{1}{\Delta t} \cdot 
    P(\Delta D_{j}=1 \mid \overline{C}_{j}, \overline{L}_{j-1}, \overline{Y}_{j-1},\overline{D}_{j-1}, A) ~. \label{eq:survival_discrete_independent_censoring}
\end{align}
Under faithfulness, exchangeability (\ref{eq:exchangeability_total_terminating}) is implied by (\ref{eq:survival_discrete_independent_censoring}). The contrast of Expressions (2) and (4) in~\citet{robins_correcting_2000} is similar to this correspondence.

\section{Estimation \label{sec:appendix estimation}}

\begin{theorem}\label{thm:convergence}
Suppose $P(Z_{t_{K+1}}=1) >0$. We let $\hat R^{(n,i)}$ be as in (\ref{eq: continuous time weight estimator}) (originally defined in \citet{ryalen_additive_2019}); an estimator of the weights $R^i$  based on additive hazard models with finite third moments on the covariates (see \citet[Theorem 2]{ryalen_additive_2019}). Suppose $R^{(n,i)}, R^i$ are uniformly bounded 
and that the $\mathcal F^{N,N^D,L,A}$-intensity of $N^i$ satisfies $E[\int_0^{t_{K+1}}\lambda_s ds]<\infty$. Then,
\begin{itemize}
    \item the estimator defined by (\ref{eq: weighted risk set estimator})-(\ref{eq: weighted rate estimator Hajek and H-T}) are consistent and predictably uniformly tight (P-UT).
    \item the estimator defined by the system (\ref{eq: mean frequency estimator}) is consistent.
\end{itemize}
\end{theorem}

\begin{proof}
We show the result for the integrators in (\ref{eq: weighted risk set estimator}); the result for (\ref{eq: weighted rate estimator Hajek and H-T}) follows by similar arguments. Assume first that $R^i, R^{(n,i)}$ are orthogonal to $N^i$. Define $X^{(n)} = \frac{1}{n} \sum_{i=1}^n R^{(n,i)} N^i, Y^{(n)} = \frac{n}{\sum_{j=1}^n Z^j}$. Then,  \citet[Lemma 2]{ryalen_additive_2019} and the law of large numbers imply that $X^{(n)}$ converges to $E[ R^i N^i ] = E[\int R_{s-}^i dN_s^i]$. $Y^{(n)}$ converges to $1/E[Z^1]$ by the law of large numbers. Furthermore, $$X^{(n)} = \frac{1}{n} \sum_{i=1}^n \int_0^\cdot R_{s-}^{(n,i)} dN_s^i + \frac{1}{n} \sum_{i=1}^n \int_0^\cdot N_{s-}^i dR_{s}^{(n,i)}$$ is PUT (see \citet[VI, 6.6.a]{JacodShiryaev}), as it is a sum of two processes that are P-UT. The first term is P-UT from \citet[Proposition 1]{ryalen_additive_2019}.  The latter term is P-UT because $R^{(n,i)}$ is driven by a P-UT process  $K^{(n,i)}$ (see \citet{ryalen_additive_2019}) and $N_{-}^i$ is predictable, by  \citet[Corollary VI.6.20]{JacodShiryaev}. 
The estimator of interest is $ \int_0^\cdot Y^{(n)}_{s-} d\tilde X^{(n)}_s$, where $\tilde X^{(n)} = X^{(n)}- \frac{1}{n}\sum_{i=1}^n \int_0^t N_{s-}^i dR_{s}^{(n,i)}$. Now, the last term on the right hand side has the following decomposition:
\begin{align*}
    \frac{1}{n}\sum_{i=1}^n \int_0^{\cdot} N_{s-}^i dR_{s}^{(n,i)} &= \frac{1}{n}\sum_{i=1}^n N^i (R^{(n,i)} - R^{i}) - \frac{1}{n}\sum_{i=1}^n\int_0^{\cdot} (R_{s-}^{(n,i)} - R_{s-}^{i} ) dN_s^i \\
    &\qquad + \frac{1}{n}\sum_{i=1}^n\int_0^\cdot N_{s-}^i dR_s^i ~.
\end{align*}
By \citet[Lemma 2]{ryalen_additive_2019}, the first term on the right hand side converges to zero, while the third term converges to zero by the law of large numbers as it is a mean zero martingale. Finally,
\begin{align*}
    \frac{1}{n}\sum_{i=1}^n\int_0^{\cdot} (R_{s-}^{(n,i)} - R_{s-}^{i} ) dN_s^i &=  \frac{1}{n}\sum_{i=1}^n\int_0^{\cdot} (R_{s-}^{(n,i)} - R_{s-}^{i} ) dM_s^i \\
    &\qquad +  \frac{1}{n}\sum_{i=1}^n\int_0^{\cdot} (R_{s-}^{(n,i)} - R_{s-}^{i} ) \lambda_s^i ds ~.
\end{align*}
The first term on the right hand side converges to zero by the law of large numbers, while the second term converges to zero by dominated convergence, as $R^{(n,i)}_t - R_t^i$ converges in probability to zero for each $t$.

Thus, $\tilde X^{(n)}$ and $X^{(n)}$ converge to the same limit. Because $E[Z^1]$ is continuous, \citet[Corollary VI 3.33]{JacodShiryaev} implies that $(\tilde X^{(n)}, Y^{(n)})$ converges weakly (with respect to the Skorohod metric),  and \citet[Theorem VI 6.22]{JacodShiryaev} implies that also $ \int_0^\cdot Y_{s-}^{(n)} d\tilde X_s^{(n)} $ converges weakly to the deterministic limit $\int_0^\cdot  E[R_s^1 dN_s^1| Z_s^1 > 0]$.

If $N^i$ is not orthogonal to $R^i$ and $R^{(n,i)}$, we have
$$ R^i N^i = 
\int_0^\cdot \theta_s^i R_{s-}^i dN_s^i + \int_0^\cdot N_{s-}^i dR_s^i ~,$$
and likewise 
\begin{align*}
    \frac{1}{n} \sum_{i=1}^n R^{(n,i)} N^i &= \frac{1}{n} \sum_{i=1}^n \int_0^\cdot \hat \theta_{s-}^{i} R_{s-}^{(n,i)} dN_s^i + \frac{1}{n} \sum_{i=1}^n \int_0^\cdot N_{s-}^i dR_{s}^{(n,i)} \\
    &\qquad + \frac{1}{n}\sum_{i=1}^n \int_0^\cdot R_{s-}^{(n,i)}Z_s^i d[N^i,\hat A^i-\hat A^{*,i}]_s~,
\end{align*}
where the last term on the right-hand side can be neglected. We can then build on the argument above, replacing $R^i$ and $R^{(n,i)}$ with $\theta^i R^i$ and $\hat \theta^{i} R^{(n,i)}$ when necessary, to show the convergence.

The consistency of (\ref{eq: mean frequency estimator}) follows from \citet[Theorem 1]{ryalen2018transforming} because $\int_0^{\cdot} Y_{s-}^{(n)} d\tilde X_s^{(n)}$ is consistent and P-UT.
\end{proof}

\end{document}